\documentclass[11pt]{article}
\usepackage[utf8]{inputenc}
\usepackage{fullpage}
\usepackage{amsmath,amsthm,amsfonts,amssymb}
\usepackage{mathtools}
\usepackage{tikz}
\usepackage{enumitem}
\usepackage[colorlinks=true, allcolors=blue]{hyperref}
\usepackage[capitalize,nameinlink]{cleveref}
\usepackage{thmtools}
\usepackage{thm-restate}
\usepackage{complexity}
\usepackage[color=red!20, textwidth=20mm, colorinlistoftodos]{todonotes}
\usepackage{comment}

\newcommand{\eps}{\varepsilon}
\newcommand{\N}{\mathbb{N}}
\newcommand{\F}{\mathbb{F}}
\renewcommand{\R}{\mathbb{R}}
\newcommand{\Z}{\mathbb{Z}}

\newcommand{\floor}[1]{{\lfloor#1\rfloor}}
\newcommand{\ceil}[1]{{\lceil#1\rceil}}

\newtheorem{theorem}{Theorem}[section]
\newtheorem{lemma}[theorem]{Lemma}

\newtheorem{corollary}[theorem]{Corollary}
\newtheorem{definition}[theorem]{Definition}
\newtheorem{remark}[theorem]{Remark}

\newtheorem{proposition}[theorem]{Proposition}

\newcommand{\problem}[1]{\textsl{#1}}

\begin{document}
\sloppy
\pagenumbering{roman}

\title{Polynomial formulations as a~barrier for~reduction-based hardness~proofs}
\date{}
\author{
    Tatiana Belova
    \thanks{Steklov Institute of Mathematics at St. Petersburg. Email: \texttt{yukikomodo@gmail.com}.}
    \and
	Alexander Golovnev
	\thanks{Georgetown University. Email: \texttt{alexgolovnev@gmail.com}.}
	\and
	Alexander~S. Kulikov
	\thanks{Steklov Institute of Mathematics at St. Petersburg and St.~Petersburg State University. Email: \texttt{alexander.s.kulikov@gmail.com}.}
	\and
	Ivan Mihajlin
	\thanks{Leonhard Euler International Mathematical Institute in St.~Petersburg. Email: \texttt{ivmihajlin@gmail.com}.}
	\and
	Denil Sharipov 
	\thanks{St.~Petersburg State University. Email: \texttt{d.i.sharipov@yandex.ru}.}
}

\maketitle

\begin{abstract}
    The Strong Exponential Time Hypothesis (SETH) 
    asserts that for every $\varepsilon>0$ there exists~$k$
    such that $k$-SAT requires time
    $(2-\varepsilon)^n$. The field of fine-grained complexity has leveraged SETH to prove quite tight conditional lower bounds for dozens of problems in various domains and complexity classes, including Edit Distance, Graph Diameter, Hitting Set, Independent Set, and Orthogonal Vectors.
    Yet, it~has been repeatedly asked in~the literature 
    whether SETH-hardness results can be~proven for other fundamental problems such as Hamiltonian Path, Independent Set, Chromatic Number,
    MAX-$k$-SAT, and Set Cover.
    
    In this paper, we show that fine-grained reductions implying even $\lambda^n$-hardness of these problems from SETH for \emph{any} $\lambda>1$, would imply new circuit lower bounds: super-linear lower bounds for Boolean series-parallel circuits or~polynomial lower bounds for arithmetic circuits (each of which is a four-decade open question).
    
    We also extend this barrier result to the class of parameterized problems. Namely, for every~$\lambda>1$ we conditionally rule out fine-grained reductions implying SETH-based lower bounds of~$\lambda^k$ for a number of problems parameterized by the solution size~$k$.
    
    Our main technical tool is a new concept called polynomial formulations. In particular, we show that many problems can be represented by relatively succinct low-degree polynomials, and that any problem with such a representation cannot be proven SETH-hard (without proving new circuit lower bounds).
\end{abstract}

\thispagestyle{empty}
\newpage
\setcounter{tocdepth}{2}
\tableofcontents
\newpage
\pagenumbering{arabic}

\section{Introduction}
In this paper, we explain the lack of hardness results based on the Strong Exponential Time Hypothesis for a~large class of problems by proving that such hardness results would lead to new strong circuit lower bounds.
\subsection{Background}
The central question in~complexity theory is to find the minimum time required to~solve a~given computational problem. Answering such a~question involves
proving lower bounds on~computational complexity. Unconditional lower bounds
remain elusive: for example, we do~not know how to~solve \problem{CNF-SAT} 
in~time $(2-\varepsilon)^n$ (where $n$~is the number of~variables in~an~input CNF formula and $\varepsilon>0$ is a~constant),
and at~the same time we~have no~tools to~exclude the possibility of even an~$O(n)$ time
algorithm. Super-linear lower bounds are only known for restricted models of computation.

For this reason, all existing lower bounds are conditional. Classical complexity theory, founded in the 1970's, considers polynomial-time reductions:
we say that $P$~reduces to~$Q$ and write $P \le Q$, if a~polynomial-time algorithm for~$Q$ can be~used to~solve~$P$ in polynomial time. 
Such a reduction may be~viewed as~a~conditional lower bound:
if the~problem~$P$ cannot be~solved in~polynomial time, then neither can the problem~$Q$. 
While polynomial-time reductions (conditionally) rule out polynomial-time algorithms for many problems, they say little about \emph{quantitative} hardness of computational problems.

The recently developed field of fine-grained complexity aims to establish tighter connections between complexities of computational problems. By using fine-grained reductions, one can leverage algorithmic hardness assumptions to prove \emph{quantitative} lower bounds for wide classes of problems. A~fine-grained reduction, denoted $(P, p(n)) \le (Q, q(n))$, implies that a~faster than $q(n)$-time
algorithm for~$Q$ leads to a~faster than $p(n)$-time algorithm for~$P$.
The standard assumptions in this field (see Vassilevska Williams~\cite{DBLP:conf/iwpec/Williams15,vw18} for excellent surveys on this topic) are hardness of \problem{CNF-SAT}~\cite{ip99,ipz98}, \problem{$3$-SUM}~\cite{go95,e99}, \problem{Orthogonal Vectors}~\cite{w05}, \problem{All Pairs Shortest Paths}~\cite{vww10}, \problem{Online Matrix-Vector Multiplication}~\cite{hkns15}, and \problem{Set Cover}~\cite{DBLP:journals/talg/CyganDLMNOPSW16}.

One of the most popular fine-grained assumptions, the Strong Exponential Time Hypothesis (SETH), postulates that for every $\eps>0$ there exists a~$k$ such that \problem{$k$-SAT} cannot be solved in time $(2-\eps)^n$. The Strong explanatory power of SETH is confirmed by many tight lower bounds for computational problems both in \P{} and \NP{}. We refer the reader to \cite[Section~3]{vw18} for an extensive list of such results, and we list a~few notable representatives below. The following upper bounds are known to~be~tight (up to small multiplicative factors) under SETH:
\begin{itemize}
    \item $n^2$ for \problem{Orthogonal Vectors}~\cite{w05} (where $n$~is the number of~vectors), $3/2$-approximate \problem{Graph Diameter}~\cite{roditty2013fast} (where $n$~is the number of nodes in the input graph), and \problem{Edit Distance}~\cite{DBLP:journals/siamcomp/BackursI18} (where $n$~is the length of~the input strings);
    \item $2^{n}$ for \problem{Hitting~Set} (where $n$~is the size of the universe); and \problem{NAE-SAT} (where $n$~is the number of the variables)~\cite{DBLP:journals/talg/CyganDLMNOPSW16};
    \item $n^k$ for \problem{$k$-Dominating Set}~\cite{DBLP:conf/soda/PatrascuW10} (where $n$~is the number of~nodes in~the input graph and $k\geq7$);
    \item $2^{\operatorname{tw}}$ for \problem{Independent Set}~\cite{DBLP:journals/talg/LokshtanovMS18} (where $\operatorname{tw}$ is the treewidth of the input graph).
\end{itemize}
Given such an extensive list of tight conditional lower bounds, one may speculate that SETH can explain many other current algorithmic barriers. 
\begin{restatable}{openproblem}{finegrained}\label{op:finegrained}
Can we prove $\lambda^n$-SETH hardness results for $\lambda>1$ for any of the following problems:  
\problem{$k$-SAT} (asked in~\cite[problem~5]{DBLP:journals/talg/CyganDLMNOPSW16}), 
        \problem{Hamiltonian Path} (asked in~\cite[chapter~12]{DBLP:series/txtcs/FominK10}),
        \problem{Chromatic Number} (asked in~~\cite[problem~5]{DBLP:journals/eatcs/LokshtanovMS11}, \cite[problem~43]{DBLP:conf/icalp/Zamir21}, \cite[problem~2]{DBLP:journals/talg/CyganDLMNOPSW16}),
        \problem{Set Cover} (asked in \cite[problem~1]{DBLP:journals/talg/CyganDLMNOPSW16}),
        \problem{Independent Set} (asked in \cite[problem~4]{DBLP:journals/eatcs/LokshtanovMS11}),
        \problem{Clique},
        \problem{Vertex Cover},
        \problem{MAX-$k$-SAT}, 
        \problem{$3$d-Matching}? 
        
        Can we prove $\lambda^k$-SETH hardness results \cref{sec:problemsparam} for $\lambda>1$ for any of the following parameterized problems (asked in~\cite[problem~2]{DBLP:journals/eatcs/LokshtanovMS11}): 
        \problem{$k$-Path}, 
        \problem{$k$-Vertex Cover}, 
        \problem{$k$-Tree}, 
        \problem{$k$-Steiner Tree}, 
        \problem{$k$-Internal Spanning Tree}, 
        \problem{$k$-Leaf Spanning Tree}, 
        \problem{$k$-Nonblocker}, 
        \problem{$k$-Path Contractibility}, 
        \problem{$k$-Cluster Editing}, 
        \problem{$k$-Set Splitting}?
\end{restatable}

\paragraph{Barriers for hardness proofs.} The first (and the only prior to this work) conditional barrier for proving SETH-hardness results was shown by Carmosino 
et~al.~\cite{DBLP:conf/innovations/CarmosinoGIMPS16}. For 
an~integer~$k$, \problem{$k$-TAUT}
    is~the language of~all $k$-DNFs that are tautologies (note that a $k$-DNF formula $\phi$ is in $k$-TAUT if and only if the $k$-CNF formula resulting from negating~$\phi$ is not in $k$-SAT). \cite{DBLP:conf/innovations/CarmosinoGIMPS16} defines a stronger version of SETH---Non-deterministic Strong Exponential Time Hypothesis (NSETH)---which postulates that for every $\eps>0$ there exists $k$ such that even non-deterministic algorithms cannot solve \problem{$k$-TAUT} on $n$-variate formulas in time $2^{(1-\eps)n}$. While this conjecture is stronger than SETH, refuting even NSETH would imply strong lower bounds against Boolean series-parallel circuits~\cite{DBLP:journals/iandc/JahanjouMV18}. Carmosino et al.~\cite{DBLP:conf/innovations/CarmosinoGIMPS16} proved that tight SETH-hardness results for \problem{$3$-SUM}, \problem{APSP}, and some other problems would refute NSETH, and, thus, imply new circuit lower bounds (resolving a four-decade open question). 

\paragraph{Circuit lower bounds.} The barriers we show for SETH-hardness proofs also lie in the field of circuit complexity. Below we review two of the main challenges in this field, and we give rigorous definitions of the circuit models in \cref{sec:spckts,sec:ackts}.
The best known lower bound on the size of Boolean circuits computing functions in~\P{} is~${3.1n-o(n)}$~\cite{31n}. In fact, this bound remains the best known even for the much larger class of functions from $\E^{\NP}$ even against the~restricted model of~series-parallel circuits. A long-standing open problem in Boolean circuit complexity is to find an explicit language that cannot be computed by linear-size circuits from various restricted circuit classes~\cite[Frontier~3]{Valiant, AB2009}.
\begin{restatable}{openproblem}{spckts}\label{op:spckts}
Prove a lower bound of $\omega(n)$ on the size of Boolean series-parallel circuits computing a language from $\E^{\NP}$.
\end{restatable}
In contrast to the case of Boolean circuits, in the model of arithmetic circuits we have a super-linear lower bound of $\Omega(n\log{n})$~\cite{s73,bs83}. One of the biggest challenges in this area is to prove a stronger lower bound, for example, a lower bound of $n^{\gamma}$ for a constant $\gamma>1$.
\begin{restatable}{openproblem}{ackts}\label{op:ackts}
For a constant $\gamma>1$, prove a lower bound of $n^\gamma$ on the arithmetic circuit complexity of a constant-degree polynomial that can be constructed in polynomial time~$n^{O(1)}$.
\end{restatable}

\subsection{Our Contribution}
Despite much effort, we still do not have \emph{any} SETH-hardness results for the problems listed in \cref{op:finegrained}. For example, an algorithm solving the \problem{Hamiltonian Path} problem in time $2^n n^{O(1)}$ has been known for $60$ years~\cite{bellman1962dynamic,held1962dynamic}, yet we don't have any improvements on the algorithm nor conditional lower bounds $\lambda^n$ on the complexity of this problem.\footnote{The undirected version of \problem{Hamiltonian Path} can be solved in time~$1.66^n$ by Bj{\"{o}}rklund's algorithm~\cite{B10}. In this paper, we only consider the directed version of the \problem{Hamiltonian Path} and \problem{$k$-Path} problems. Since we're proving barriers for lower bounds, considering the harder directed version of a problem only makes our barrier results stronger.} In this paper, we show that this barrier is no coincidence. Namely, we show that a resolution of~\cref{op:finegrained} would resolve \cref{op:spckts} or \cref{op:ackts}.
More specifically, \emph{any} SETH-based exponential lower bound for \problem{Hamiltonian Path} or any other problem from \cref{op:finegrained} would imply a super-linear lower bound for series-parallel circuits or an arbitrarily large polynomial lower bound for arithmetic circuits. 

Our first main result says that for a number of well-studied problems, any SETH lower bound of the form $\lambda^n$ for a constant $\lambda>1$ would imply new circuit lower bounds. (In the end of this section we clarify what we mean by SETH lower bounds, and our definition is quite general.)
\begin{restatable}{theorem}{intromain}\label{thm:intromain}
If at least one of the following problems
    \begin{quote}
        \problem{$k$-SAT}, 
        \problem{MAX-$k$-SAT}, 
        \problem{Hamiltonian Path},
        \problem{Graph Coloring},
        \problem{Set Cover},
        \problem{Independent Set},
        \problem{Clique},
        \problem{Vertex Cover},
        \problem{$3$d-Matching}
    \end{quote}
is $\lambda^n$-SETH-hard for a constant $\lambda>1$, then at~least one of~the following circuit lower bounds holds:
	\begin{itemize}
		\item $\E^{\NP}$ requires series-parallel Boolean circuits of~size~$\omega(n)$;
		\item for every constant $\gamma>1$, there exists an~explicit family of constant-degree polynomials over~$\Z$ that requires arithmetic circuits of~size $\Omega(n^{\gamma})$.
	\end{itemize}
\end{restatable}

While this result conditionally rules out, say, $1.1^n$-hardness of \problem{Hamiltonian Path}, it is still possible that the parameterized version of \problem{Hamiltonian Path}, \problem{$k$-Path}, is $1.1^k$-hard for \emph{some} function $k\vcentcolon= k(n)$. In our second result, we conditionally rule out even such hardness results (which might have seemed easier to obtain).
\begin{restatable}{theorem}{intromainparam}\label{thm:intromainparam}
If at least one of the following parameterized problems
    \begin{quote}
         \problem{$k$-Path}, 
         \problem{$k$-Vertex Cover}, 
         \problem{$k$-Tree}, 
         \problem{$k$-Steiner Tree}, 
         \problem{$k$-Internal Spanning Tree}, 
         \problem{$k$-Leaf Spanning Tree}, 
         \problem{$k$-Nonblocker}, 
         \problem{$k$-Path Contractibility}, 
         \problem{$k$-Cluster Editing}, 
         \problem{$k$-Set Splitting}
    \end{quote}
is $\lambda^k$-SETH-hard for a constant $\lambda>1$, then at~least one of~the following circuit lower bounds holds:
	\begin{itemize}
		\item $\E^{\NP}$ requires series-parallel Boolean circuits of~size~$\omega(n)$;
		\item for every constant $\gamma>1$, there exists an~explicit family of constant-degree polynomials over~$\Z$ that requires arithmetic circuits of~size $\Omega(n^{\gamma})$.
	\end{itemize}
\end{restatable}
A couple of remarks about the main results of this work are in order. 
While an SETH-lower bound for any problem in the premise of \cref{thm:intromain,thm:intromainparam} would imply that at least one of \cref{op:spckts} or \cref{op:ackts} must have an affirmative answer, it would not say which~one!

The first part of the conclusion of  \cref{thm:intromain,thm:intromainparam} is really that ``NSETH fails'', and the series-parallel circuit lower bound comes as an already known consequence of that~\cite{DBLP:journals/iandc/JahanjouMV18}.

\paragraph{On the definition of SETH-hardness.}
The Exponential Time Hypothesis (ETH), which follows from SETH~\cite{ipz98}, asserts that there is some non-explicit $\lambda>1$ such that \problem{$3$-SAT} on~$n$ variables requires time~$\lambda^n$. While SETH implies the existence of such $\lambda>1$, it doesn't tell us anything about $\lambda$. In particular, it doesn't provide us with a lower bound on $\lambda$ greater than~$1$.  
Each of the problems discussed above has a reduction from \problem{$3$-SAT} that preserves the size of the instance up to a constant factor~\cite[Theorem~14.6]{ipz98,cygan2015parameterized}. Thus, under SETH (or even ETH), none of these problems can be solved by algorithms running in time $(\lambda')^n$ for some \emph{non-explicit} $\lambda'>1$ which we cannot even bound away from~$1$.

The only currently known way to prove fine-grained and SETH-hardness results is via fine-grained reductions. Such a reduction gives a conditional lower bound of $\lambda^n$ on the complexity of a computational problem for an \emph{explicit} $\lambda>1$. In this work, we (conditionally) rule out reduction-based proofs of SETH-hardness which we emphasize in the title of the paper. 

To simplify the presentation, we follow the terminology of previous works in this area and say ``a computational problem~$P$ is $\lambda^n$-SETH hard for $\lambda>1$'' when we mean the following. There exists an \emph{explicit} constant~$\lambda>1$ and a function $\delta\vcentcolon=\delta(\eps)$ such that for every $k$ and $\eps>0$, there exists an algorithm $A$ for \problem{$k$-SAT} running in time $2^{(1-\delta(\eps))n}$, making~$t$ calls to an oracle for~$P$ with input sizes $n_1,\ldots,n_t$ satisfying
\[
\sum_{i=1}^t \lambda^{(1-\eps)n_i} \leq 2^{(1-\delta(\eps))n} \;.
\]
In particular, this captures all known $\lambda^n$-SETH lower bounds in fine-grained complexity for explicit~$\lambda>1$.

\paragraph{On randomized reductions.}
While this work focuses on deterministic (i.e., non-randomized) reductions, now we discuss why randomized SETH-hardness reductions would also be surprising. If one of the problems under consideration is SETH-hard under deterministic reductions, then \cref{thm:intromain,thm:intromainparam} imply that either NSETH is false (implying series-parallel circuit lower bounds) or we have high arithmetic circuit lower bounds. In fact, the same result holds for zero-error probabilistic reductions. 

A randomized SETH-hardness reduction for any of the problems in the premises of \cref{thm:intromain,thm:intromainparam} would either give us high arithmetic circuit lower bounds or a faster than $2^n$ two-round \AM{} protocol for \problem{$k$-TAUT}. While faster than $2^n$ two-round \AM{} protocols are not known to imply circuit lower bounds, designing such a protocol for \problem{$k$-TAUT} would still be a big achievement. The~celebrated work of Williams~\cite{W16} gives a constant-round \AM{} protocol running in time $2^{n/2}$ for \problem{$k$-TAUT}.\footnote{The standard transformation of a constant-round \AM{} protocol into a two-round \AM{} protocol suffers a quadratic blow-up in size which results in a two-round protocol of trivial length $2^n$.}

Also, our result, together with the standard trick of simulating randomness by non-uniformity (see~\cite[Lemma~4]{DBLP:conf/innovations/CarmosinoGIMPS16}), gives us that a randomized SETH-hardness reduction for any of the problems in the premises of \cref{thm:intromain,thm:intromainparam} implies either high arithmetic circuit lower bounds or that NUNSETH (Non-unifrom NSETH, see \cite[Definition~4]{DBLP:conf/innovations/CarmosinoGIMPS16})) is false.  

\subsection{Proof Overview}
\subsubsection{Non-parameterized Problems}
We demonstrate our main result on the following example. Assume that the \problem{Hamiltonian Path} problem is $1.5^n$-SETH hard. We'll prove one of the two circuit lower bounds: a super-linear lower bound for Boolean series-parallel circuits computing a language from $\E^\NP$ or an $n^{1.2}$ lower bound on the size of arithmetic circuits computing an explicit $n$-variate polynomial of constant degree. We note that the actual result proven in~\cref{thm:intromain} rules out $\lambda^n$-hardness for every $\lambda>1$ and is capable of proving arithmetic circuit lower bounds of $n^\gamma$ for every $\gamma>1$.

First we find a low-degree polynomial on exponentially-many variables that can be used to solve \emph{any} \problem{Hamiltonian path} instance of a given size. We call this a polynomial formulation of \problem{Hamiltonian Path}.\footnote{In fact, our main result (\cref{thm:main}) rules out SETH-hardness results for all problems that admit polynomial formulations on $c^n$ variables of instances of size~$n$ for every~$c>1$.}
Then we show how to find a small arithmetic circuit computing this polynomial formulation. After that, we give a non-deterministic algorithm for the \problem{$k$-TAUT} problem, and conclude with circuit lower bounds.
\paragraph{Polynomial formulations of \problem{Hamiltonian Path}.}
First we give the following polynomial formulation of \problem{Hamiltonian Path} on a graph with node set $[\ell]$. Without loss of generality, we assume that $\ell$ is a multiple of~$10$. For every set $S$ of $\ell/10$ nodes, every node $u\in S$, and every node $v\not\in S$, we introduce a variable $x_{S,u,v}$. Thus, in total we have $t\leq \binom{\ell}{\ell/10}\cdot \ell^2\leq 2^{0.47\ell}$ variables. Consider the following degree-$10$ polynomial $P$ of~$t$ variables. For each partition of $[\ell]=S_1 \sqcup \dotsb \sqcup S_{10}$ into sets of size $\ell/10$, and for each $v_1\in S_1,\dotsc,v_{10}\in S_{10}$, and $v_{11}\in[\ell]$ we add the monomial
\[
x_{S_1,v_1,v_2}\cdot x_{S_2,v_2,v_3}\cdots x_{S_{10},v_{10},v_{11}} \,.
\]
Now for an instance $G$ of \problem{Hamiltonian path}, where $G$ is a graph on $\ell$ nodes, we assign the following values to the $t$ variables of the polynomial~$P$. If there is a Hamiltonian path visiting each node of~$S$ exactly once, starting at the node $u$ and ending at a node neighboring~$v$, then we set ${x_{S,u,v}=1}$, otherwise we set ${x_{S,u,v}=0}$. Note that the polynomial~$P$ evaluates to~$0$ if and only if the original graph has no Hamiltonian paths. Indeed, a Hamiltonian path can be partitioned into $10$~parts of length~$\ell/10$ with starting nodes $v_1,\ldots, v_{10}$ such that all variables of the monomial $x_{S_1,v_1,v_2}\cdot x_{S_2,v_2,v_3}\cdots x_{S_{10},v_{10},v_{11}}$ are assigned ones. The converse is also true: if all variables of some monomial are assigned ones, then this gives us a Hamiltonian path in~$G$. 

\paragraph{Time complexity of polynomial formulations.}
Now we bound the running time for computing the polynomial~$P$. This polynomial can be constructed by listing all $10$-tuples of $(\ell/10)$-subsets of $[\ell]$ together with nodes $v_1,\ldots,v_{11}$, which can be done in time $O\left(\binom{\ell}{\ell/10}^{10} \ell^{O(1)}\right)\leq2^{4.7\ell}$. This polynomial is not particularly helpful for solving an instance of \problem{Hamiltonian Path} on $\ell$ nodes since the problem can be solved in time~$2^\ell\ell^{O(1)}$ while only constructing this polynomial takes time $2^{4.7\ell}$. Nevertheless, we'll later use the self-reducibility property of \problem{TAUT} to reuse one polynomial to solve exponentially many instances of \problem{Hamiltonian Path}.

Let us now bound the running time for computing the assignment of the variables for an instance~$G$ on $\ell$ nodes.  In order to find the assignment of the variables, it's sufficient to solve $t$ instances of Hamiltonian path on $\ell/10$ nodes. Since each instance can be solved in time $2^{\ell/10}\ell^{O(1)}$, we bound the total running time by 
\[t\cdot 2^{\ell/10}\ell^{O(1)} \leq \binom{\ell}{\ell/10}\cdot 2^{\ell/10}\cdot\ell^{O(1)}=O(2^{0.47\ell+\ell/10})=O(1.49^\ell)\;.\]
Thus, in time~$2^{4.7\ell}$ we construct a degree-$10$ polynomial with at most $t\leq 2^{0.47\ell}$ variables, such that each instance of \problem{Hamiltonian path} on $\ell$ nodes can be reduced in time $1.49^\ell$ to evaluating this polynomial at a $0/1$ point. Note that since the polynomial has $2^{4.7\ell}$ monomials with coefficients one, and we only evaluate it at $0/1$ points, the maximum value we can obtain is $\leq2^{4.7\ell}$, so we can as well assume that our polynomial $P$ is over $\mathbb{Z}_p$ for $2 \cdot 2^{4.7\ell} \leq p \leq 4\cdot 2^{4.7\ell}$. 

\paragraph{The first circuit lower bound.}
If the constructed polynomial $P$ of $t$ variables doesn't have arithmetic circuits of size $t^{1.2}$, then we have our first circuit lower bound. Indeed, we have an explicit family of constant-degree $t$-variate polynomials $P_t$ that require arithmetic circuits of size at least $t^{1.2}$. To see that this family of polynomials is explicit, recall that this $t$-variate polynomial can be constructed in time polynomial in the number of variables:
$O\left(\binom{\ell}{\ell/10}^{10} \ell^{O(1)}\right) \leq t^{11}$. Thus, in the following we assume that the polynomial $P$ does have circuits of size $t^{1.2}$. From this (together with the assumed $1.5^n$-SETH hardness of \problem{Hamiltonian Path}) we will prove the second circuit lower bound. In fact, we will refute NSETH, which, as discussed earlier, implies a super-linear circuit lower bound for Boolean series-parallel circuits~\cite{DBLP:journals/iandc/JahanjouMV18}. Therefore, in the rest of this section, we will show how to solve the \problem{$k$-TAUT} problem in non-deterministic time $2^{(1-\eps)n}$ for constant $\eps>0$.

\paragraph{Approximation of MACP.}
We have a $t$-variate polynomial~$P$ that solves \problem{Hamiltonian Path} on graphs with $\ell$ nodes, and we know that this polynomial has circuits of size $t^{1.2}$, but we don't have those circuits. We use the result of Strassen~\cite{strassen1973vermeidung} that asserts that if there is a circuit of size $t^{1.2}$ computing a degree-$10$ polynomial, then there exists a circuit of size $O(t^{1.2})$ computing the same polynomial, where every gate of the circuit computes a polynomial of degree at most~$10$ (see \cref{cor:bounded-degree}). We use the power of non-deterministic algorithm to guess such a circuit, but we need to verify that the guessed circuit indeed computes the polynomial~$P$. For this, we write down the polynomials computed at each gate of the circuit. Since in our modified circuit, each gate computes a \mbox{degree-$10$} polynomial, each such polynomial can be written in time $O(t^{10})$. Performing all operations (modulo~$p$) over the polynomials in this circuit will take time $O(t^{1.2}\cdot t^{20}\cdot \log^2p)\leq 2^{21.2\cdot0.47\ell}\cdot\ell^{O(1)}\leq2^{10\ell}$. By comparing the list of the monomials of the polynomial~$P$ with the monomials computed at the output gate, we verify the guessed circuit.

\paragraph{Non-deterministic algorithm for \problem{$k$-TAUT}.}
Below we present an algorithm solving every \problem{$k$-TAUT} instance on~$n$ variables in non-deterministic time $2^{(1-\eps)n}$ for constant $\eps$ independent of~$k$. Recall that this refutes NSETH and, thus, implies the second circuit lower bound.

We consider all assignments of the first $(1-\alpha)n$ Boolean variables for $\alpha=1/20$. This gives us a set of $2^{(1-\alpha)n}$ instances of \problem{$k$-TAUT} on $\alpha n$ variables each. Solutions to these $2^{(1-\alpha)n}$ instances will give us the solution to the original instance. Each of these instances of \problem{$k$-TAUT} will later be reduced (via the assumed fine-grained reduction from \problem{$k$-SAT} to \problem{Hamiltonian Path}) to (possibly exponentially many) instances of \problem{Hamiltonian Path}. Since we assume $1.5^n$-SETH hardness of \problem{Hamiltonian Path}, an instance of \problem{$k$-TAUT} on $\alpha n$ variables must be reduced to instances of \problem{Hamiltonian Path} on at most $\ell$ nodes, where $\ell$ satisfies $1.5^\ell \leq 2^{\alpha n}$. In particular, $\ell<n/11$. In this proof overview we assume for simplicity that all instances produced by the fine-grained reduction have size $\ell$, (in the full proof of \cref{thm:main} we generate $\ell$ different polynomials, one for each integer up to $\ell$, to~solve \problem{Hamiltonian Path} instances of any size $\leq \ell$).

First, we construct the polynomial $P$ that solves the instances of \problem{Hamiltonian Path} on $\ell$ nodes. Recall that we can construct this polynomial in time $2^{4.7\ell}\ll 1.9^n$ (we only need to write down this polynomial once, and then we will use this polynomial for all instances of \problem{Hamiltonian Path}). Then we find a circuit of size $O(t^{1.2})$ computing~$P$ in non-deterministic time $2^{10\ell}\ll1.9^n$. Now we are ready to solve all $2^{(1-\alpha)n}$ instances of \problem{$k$-TAUT} on $\alpha n$ variables. Indeed, each such instance we reduce to (a number of instances of) \problem{Hamiltonian Path} on $\ell$ nodes. We find the assignment of the variables of~$P$ in time $1.49^{\ell}$, and evaluate the circuit of size $O(t^{1.2})$ modulo~$p$ in time $O(t^{1.2}\log^2p)\leq 1.49^\ell$. Since we solve each instance of \problem{Hamiltonian Path} on $\ell$ nodes in time $1.49^\ell$, $1.5^n$-SETH hardness of \problem{Hamiltonian Path} implies a $2^{(1-\delta)\alpha n}$-algorithm for \problem{$k$-TAUT} on $\alpha n$ variables (for constant $\alpha>0$ independent of~$k$). Therefore, each of the $2^{(1-\alpha)n}$ instances of \problem{$k$-TAUT} is solved in time $2^{(1-\delta)\alpha n}$, and the total (non-deterministic) running time to solve the original instance of \problem{$k$-TAUT} is $2^{(1-\alpha)n}\cdot2^{(1-\delta)\alpha n}=2^{(1-\alpha\delta)n}$.

\paragraph{Summary.} We~start by~showing that \problem{Hamiltonian Path} can be~expressed as a constant-degree polynomial~$P$ in exponentially many variables. That is, checking whether the input graph has a Hamiltonian path boils down to~evaluating~$P$.

Assuming that~$P$ can be~computed by a~small arithmetic circuit~$C$
	(if~this is not the case, we~have an arithmetic circuit lower bound), we~show that there exists a~\emph{homogeneous} 
	circuit~$C'$ computing~$P$ that is not much larger than~$C$.
	We~guess~$C'$ non-deterministically. The fact that~$C'$
	is~homogeneous and that $P$~has constant degree allows~us to efficiently verify the correctness of~$C'$
	by~checking the computation at~each gate manually.
	
	The circuit $C'$ allows us to~solve \problem{Hamiltonian Path} efficiently. Since 
	we have a~fine-grained reduction from \problem{$k$-SAT} to~\problem{Hamiltonian Path},
	this implies a~fast non-deterministic algorithm for \problem{$k$-TAUT}. The obtained algorithm refutes NSETH, which, in~turn, gives super-linear lower bounds for Boolean series-parallel circuits.

\subsubsection{Parameterized Problems}
In the (parameterized) \problem{$k$-Path} problem, the goal is to check if the given graph on $n$ nodes contains a simple path on $k$ nodes. Similarly to the case of \problem{Hamiltonian Path}, the best known algorithm for \problem{$k$-Path} runs in time~$2^k$~\cite{williams2009finding}. The result sketched in the previous section (conditionally) rules out $\lambda^n$-SETH hardness of \problem{Hamiltonian Path} for every $\lambda>1$. Yet, this leaves a possibility to prove a $2^k$ lower bound on the complexity of \problem{$k$-Path} for \emph{some} function $k\vcentcolon=k(n)$. We extend our framework to parameterized problems to prove the same barriers for SETH-hardness for parameterized problems. 

While most of the machinery from the previous section extends to the case of parameterized problems, the new issue arises in the polynomial formulations of the problems. Indeed, if we apply the polynomial formulation from the previous section to the \problem{$k$-Path} problem, then we inevitably have at least $\binom{n}{k/10} \geq n^{\Omega(k)}$ variables. Such a polynomial formulation can only rule out $\left(n^{\Omega(k)}n^{O(1)}\right)$-SETH hardness results, which are of no interest since the problem can be solved in time $2^{k}n^{O(1)}$. To overcome this issue, we design a different polynomial formulation of the parameterized \problem{$k$-Path} problem. This formulation uses a certain pseudorandom object called a splitter, on which we elaborate below.

\paragraph{Splitters.} A family~$H$ of functions $f\colon[n]\to[k]$ is a $k$-perfect hash family if for every $S\subseteq[n]$ of size $|S|=k$, there is a function $h\in H$ which is injective on~$S$. An $(n,k,ck)$-splitter is a relaxation of this notion where the functions~$f$ have range $[ck]$ for a constant $c\geq 1$.  There are known constructions of splitters~\cite{NSS95} and $k$-perfect hash families~\cite{AYZ95} of size $e^{k(1+o(1))}k\log{n}$, but our polynomial formulations require splitters of size $e^{\frac{k}{g(c)}}$ for an unbounded function~$g$. While a simple probabilistic argument shows that a set of $\approx e^{\frac{k}{c}}\cdot k\log{n}$ random functions forms an $(n,k,ck)$-splitter with high probability, in \cref{sec:splitters}, we show how to efficiently construct an explicit family of such functions of size $e^{\frac{k}{c}(1+o(1))}\cdot k\log{n}$.

\paragraph{Polynomial formulations.} The color-coding technique, introduced in~\cite{AYZ95}, solves \problem{$k$-Path} on a graph with node set~$[n]$ as~follows:
    assign a~random color $c \in [k]$ to~every node; then,
    all nodes of a~$k$-path receive different colors with
    probability about $e^{-k}$; at~the same time, one can find
    such a~colorful path in~time $2^kn^{O(1)}$. This gives
    us a~randomized $2^{O(k)}n^{O(1)}$-time algorithm for 
    \problem{$k$-Path}. This algorithm can be derandomized by utilizing a~$k$-perfect family~$H$ of~hash functions $f \colon [n] \to [k]$: go~through all~$f \in H$ and assign the~color~$f(v)$ to every node $v$~\cite{AYZ95}. Since $H$ guarantees that for every $k$-path there is a~coloring $f \in H$ that assigns
    different colors to all nodes of~the path, one of the hash functions~$f$ will lead to a $k$-path. 
    \cite{AYZ95} gives a~construction of~$H$ of~size
    $e^{k(1+o(1))}$, but this would result in a polynomial formulation with $t\geq e^{k(1+o(1))}$ variables. Recall that in order to prove a lower bound of $t^{\gamma}$ on the size of arithmetic circuits, we need $t^\gamma$ to be much less than the assumed complexity of \problem{$k$-Path}. Thus, for our purposes, we~need a~family of hash function of~much smaller size.
    To~achieve this, 
    we~allow a~larger number of~colors. This will decrease the number of variables as desired at the cost of increasing the degree and the time needed to compute the coefficients of the polynomial. Fortunately, our construction is robust enough to tolerate this drawback. 
    
    We take a~family~$H$ of $(n, k, c k)$-splitters of~size $e^{\frac{k}{c}}$ that can be computed in time $2^{9k}n^{O(1)}$ described above.
    Given a~coloring $f \colon [n] \to [c k]$, we
    find an~$f$-colorful $k$-path in~time 
    $2^{c k}n^{O(1)}$~\cite{AYZ95}.
    
    The main idea of the polynomial formulation then is the following. For a~coloring $f \colon [n] \to [c k]$, an~$f$-colorful $k$-path~$\pi$ can be~partitioned into 
    $c$~paths $\pi_1, \dotsc, \pi_{c}$ such that each path uses at~most $k/c$ colors. Then, a~partition $\pi_1, \dotsc, \pi_{c}$ is~valid if the paths are color-disjoint and there~is an~edge from the last node of~$\pi_i$ to~the first node of~$\pi_{i+1}$, for all~$i \in [c-1]$. Such valid paths will lead to monomials of degree~$c$, and the resulting polynomial will have only $e^{\frac{k}{c}} n^{O(1)}$ variables for arbitrarily large constant~$c$.

\subsection{Related Work}
The closest in spirit to our result is the work of Carmosino et al.~\cite{DBLP:conf/innovations/CarmosinoGIMPS16}. The work~\cite{DBLP:conf/innovations/CarmosinoGIMPS16} proves that if APSP, $3$-SUM, or one of a few other problems is SETH-hard, then NSETH is false. In~this paper, we relax the implication but prove a stronger barrier for a much wider class of problems: if one of the problems listed in~\cref{op:finegrained} is SETH-hard, then NSETH is false or we have arbitrarily large polynomial lower bounds for arithmetic problems. Another important difference between~\cite{DBLP:conf/innovations/CarmosinoGIMPS16} and the present work is that~\cite{DBLP:conf/innovations/CarmosinoGIMPS16} rules out \emph{some} hardness results (say, $n^{3/2}$-SETH hardness of \problem{$3$-SUM}), while the present work rules out \emph{all} exponential SETH-hardness results~$\lambda^n$ for an explicit constant $\lambda>1$.

The work of Kabanets and Impagliazzo~\cite{DBLP:journals/cc/KabanetsI04} is the closest one to our main result in terms of techniques. In~fact, we view our main result (\cref{thm:main}) as a fine-grained version of~\cite{DBLP:journals/cc/KabanetsI04} with the following modifications. The work~\cite{DBLP:journals/cc/KabanetsI04} shows that if \problem{Polynomial Identity Testing} (PIT) has a~deterministic polynomial time algorithm, then either $\NEXP \not \subseteq \Ppoly$, or the \problem{Permanent} problem requires super-polynomial size arithmetic circuits. To~show this, they guess 
	small arithmetic circuits and verify them 
	using the assumed deterministic algorithm for PIT
	and downward self-reducibility of \problem{Permanent}.
	In~this paper, efficient 
	verification is~possible due to the low degree of the~polynomial computed by the circuit and the~homogenization trick allowing us to~consider only circuits with a specific structure.
	
Finally, several works provided efficient algorithms for problems studied in fine-grained complexity (e.g., Gr{\o}nlund and Pettie~\cite{DBLP:journals/jacm/GronlundP18} for \problem{$3$-SUM} and Williams~\cite{W16} for \problem{Multipoint Circuit Evaluation}) using self-reducibility of the problem and a construction that improves on the trivial running time for small instances of the problem. In~this paper,
	we~also use the self-reducibility trick: first, in the allocated time we construct arithmetic circuits solving all instances of \problem{$k$-SAT} on $\alpha n$ variables significantly faster than in $2^{\alpha n}$, and then we~reduce 
	an $n$-variate instance of \problem{$k$-SAT} to a~series of instances on $\alpha n$ variables for a~sufficiently small $\alpha>0$.

\subsection{Structure} The rest of the paper is organized as follows. In \cref{sec:prelims}, we give all necessary background material, including the definitions of SETH-hardness and polynomial formulations. In \cref{sec:macp}, we define an approximate version of the Minimum Arithmetic Circuit Problem (MACP), and provide a non-deterministic algorithm for this version of the problem. In \cref{sec:mainresultsection}, we prove the main result of this paper: SETH-hardness of a problem admitting polynomial formulations would imply one of the two aforementioned circuit lower bounds. \cref{sec:splitters} contains a construction of a pseudorandom object needed for certain polynomial formulations: deterministic splitters over alphabets of linear size. In \cref{sec:formulationssection}, we give polynomial formulations of all problems considered in this paper. Finally, \cref{sec:apx} contains proofs of technical claims omitted in the main part of the paper.

\section{Preliminaries}\label{sec:prelims}
For two sets $S$ and $T$, by $S \sqcup T$ we denote their disjoint union. All logarithms in this paper are base $2$, i.e., $\log 2^n = n$.
Recall that the $O^*(\cdot)$ notation suppresses polynomial factors, e.g., $n^2 2^n = O^*(2^n)$.

We~use square brackets in~two ways: for a~positive integer~$p$, $[p]=\{1, \dotsc, p\}$; for a~predicate~$P$, $[P]=1$ if $P$~is~true and $[P]=0$ otherwise (this~is Iverson bracket).

\subsection{Boolean Circuits}\label{sec:spckts}
\begin{definition}
A Boolean circuit~$C$ with variables $x_1,\ldots,x_n$ is a directed acyclic graph as follows. Every node has in-degree zero or two. The in-degree zero nodes are labeled either by variables $x_i$ or constants~$0$ or~$1$. The in-degree two nodes are labeled by binary Boolean functions that map $\{0,1\}^2$ to $\{0,1\}$. The only gate of out-degree zero is the output of the circuit.
\end{definition}
A Boolean circuit $C$ with variables $x_1,\ldots,x_n$ computes a Boolean function $f\colon\{0,1\}^n\to\{0,1\}$ in a natural way. We define the size of~$C$ as the number of gates in~it, and the Boolean circuit complexity of a function as the minimum size of a circuit computing it.

A circuit is called \emph{series-parallel} if there exists a numbering $\ell$ of the circuit's nodes such that for every wire $(u,v), \ell(u)<\ell(v)$, and no pair of wires $(u,v), (u',v')$ satisfies $\ell(u)<\ell(u')<\ell(v)<\ell(v')$.

The best known lower bound on the size of a Boolean circuit for functions in~\P{} is~${3.1n-o(n)}$~\cite{31n}. In fact, this bound remains the best known bound even for the much larger class of languages $\E^{\NP}$ even against the~restricted model of~series-parallel circuits. A long-standing open problem in Boolean circuit complexity is to find an explicit language that cannot be computed by linear-size circuits from various restricted circuit classes~\cite[Frontier~3]{Valiant, AB2009}.

\spckts*

\subsection{Arithmetic Circuits}\label{sec:ackts}
\begin{definition}
An arithmetic circuit $C$ over a ring~$R$ and variables $x_1,\ldots,x_n$ is a directed acyclic graph as follows. Every node has in-degree zero or two. The in-degree zero nodes are labeled either by variables $x_i$ or elements of $R$. The in-degree two nodes are labeled by either $+$ or $\times$. Every gate of out-degree zero is called an output gate.
\end{definition}
We will typically take $R$ to be $\Z$ or $\Z_p$ for a prime number $p$.
A single-output arithmetic circuit~$C$ over~$R$ computes a polynomial over $R$ in a natural way. We say that $C$ computes a polynomial $P(x_1,\ldots,x_n)$ if the two polynomials are \emph{identical} (as opposed to saying that $C$ computes $P$ if the two polynomials agree on every assignments of $(x_1,\ldots,x_n)\in R^n$). We define the size of~$C$ as the number of edges in~it, and the arithmetic circuit complexity of a polynomial as the minimum size of a circuit computing it.

While it's known~\cite{s73,bs83} that the polynomial $x_1^r+\ldots+x_n^r$ requires arithmetic circuits over $\F$ of size $\Omega(n\log(r))$ (if $r$ doesn't divide the characteristic of~$\F$), one of the biggest challenges in algebraic complexity is to prove stronger lower bounds on the arithmetic circuit complexity of an explicit polynomial of constant degree.
\ackts*
\subsection{SETH Conjectures}
Below, we~state rigorously two SETH-conjectures that we will use in this work.

\begin{itemize}
    \item Strong exponential time hypothesis 
    (SETH)~\cite{ipz98,ip99}: for every $\varepsilon>0$, there exists~$k$ such that
        \[\text{\problem{$k$-SAT}} \not \in \TIME[2^{(1-\varepsilon)n}] \, .\]
    \item Non-deterministic SETH (NSETH)~\cite{DBLP:conf/innovations/CarmosinoGIMPS16}: for every $\varepsilon>0$, there exists~$k$
    such that
    \[\text{\problem{$k$-TAUT}} \not \in \NTIME[2^{(1-\varepsilon)n}]\,,\] where \problem{$k$-TAUT}
    is~the language of~all $k$-DNFs that are tautologies.
\end{itemize}
  \cite{DBLP:journals/iandc/JahanjouMV18} proved that if SETH is false, then $\E^{\NP}$ requires series-parallel Boolean circuits of~size~$\omega(n)$. \cite[Corollary B.3.]{DBLP:conf/innovations/CarmosinoGIMPS16} extended this result and showed that refuting NSETH is sufficient for such a circuit lower bound.
    \begin{theorem}[{\cite[Corollary B.3.]{DBLP:conf/innovations/CarmosinoGIMPS16}}]\label{thm:nseth}
    If NSETH is false then $\E^{\NP}$ requires series-parallel Boolean circuits of size $\omega(n)$.
    \end{theorem}

\subsection{Fine-grained Reductions}\label{sec:fine-grained-reductions}
In this section we rigorously define fine-grained reductions and SETH-hardness. 
First, we give the definition of fine-grained reductions from~\cite[Definition~6]{DBLP:conf/iwpec/Williams15} with a small modification that specifies the dependence of $\delta$ on $\varepsilon$ (we'll need this modification when defining SETH-hardness as SETH-hardness is a \emph{sequence} of reductions from \problem{$k$-SAT} for every value of $k$).
\begin{definition}[Fine-grained reductions]\label{deg:fg}
    Let $P,Q$ be~problems, $p, q \colon \mathbb{Z}_{\ge 0} \to \mathbb{Z}_{\ge 0}$ be non-decreasing functions and ${\delta\colon\R_{>0}\to\R_{>0}}$.
    We say that $(P, p(n))$ \emph{$\delta$-fine-grained reduces} to~$(Q, q(n))$
    and write $(P,p(n)) \le_{\delta} (Q,q(n))$,
    if~for every $\varepsilon>0$ and $\delta=\delta(\eps)>0$, there exists
    an~algorithm~$\mathcal A$ for~$P$ with oracle access to~$Q$,
    a~constant~$d$, a~function $t(n) \colon \mathbb{Z}_{\ge 0} \to \mathbb{Z}_{\ge 0}$, such that on~any instance of~$P$ of size~$n$, the algorithm~$\mathcal A$
    \begin{itemize}
        \item runs in time at~most $d(p(n))^{1-\delta}$;
        \item produces at~most $t(n)$~instances of~$Q$ adaptively: every instance depends on~the previously produced instances
        as~well as~their answers of~the oracle for~$Q$;
        \item the sizes $n_i$ of the produced instances satisfy the inequality
        \[\sum_{i=1}^{t(n)}q(n_i)^{1-\varepsilon} \le d(p(n))^{1-\delta} \, .\]
    \end{itemize}
\end{definition}
We say that $(P, p(n))$ \emph{fine-grained reduces} to~$(Q, q(n))$
    and write $(P,p(n)) \le (Q,q(n))$ if $(P,p(n)) \le_{\delta} (Q,q(n))$ for some function $\delta\colon\R_{>0}\to\R_{>0}$.

It is~not difficult to~see that
if $(P,p(n)) \le (Q,q(n))$, then any improvement over the running time $q(n)$ for the problem~$Q$
implies an~improvement over the running time~$p(n)$ for the problem~$P$: for any $\varepsilon>0$, there is $\delta>0$, such that if $Q$~can be~solved in time $O(q(n)^{1-\varepsilon})$, then $P$~can be~solved in time $O(p(n)^{1-\delta})$.

\begin{definition}[SETH-hardness]\label{def:seth-hardness}
For a constant $\lambda>1$, we say that a problem $P$ is \emph{$\lambda^n$-SETH-hard} if there exists a function $\delta\colon\R_{>0}\to\R_{>0}$ and for every $k\in\N$, \[(\problem{$k$-SAT},2^n)\leq_{\delta} (P,\lambda^n) \,.\]
\end{definition}
If a problem $P$ is $\lambda^n$-SETH-hard, then any algorithm solving $P$ in time $\lambda^{(1-\varepsilon)n}$ implies an algorithm solving \problem{$k$-SAT} in time $2^{(1-\delta(\eps))n}$ for all~$k$, thus, breaking SETH.

We now similarly define SETH-hardness of parameterized problems. 
\begin{definition}[SETH-hardness of parameterized problems]\label{def:sethparam}
    Let $P$ be a parameterized problem with a parameter~$k$, $\lambda>1$ be a constant, and ${\delta\colon\R_{>0}\to\R_{>0}}$.
    We say that $P$ is \emph{$\lambda^k$-SETH-hard} if for every $q,d\in\N$, $\varepsilon>0$, and $\delta=\delta(\eps)>0$, there exists
    an~algorithm~$\mathcal A$ for \problem{$q$-SAT} with oracle access to~$P$, a~function $t(n) \colon \mathbb{Z}_{\ge 0} \to \mathbb{Z}_{\ge 0}$, such that on~any instance of~\problem{$q$-SAT} of size~$n$, the algorithm~$\mathcal A$
    \begin{itemize}
        \item runs in time at~most $O(2^{(1-\delta)n})$;
        \item produces at~most $t(n)$~instances of~$P$ adaptively: every instance depends on~the previously produced instances
        as~well as~their answers of~the oracle for~$P$;
        \item the length $\ell_i$ and parameters $k_i$ of the produced instances satisfy the inequality
        \[\sum_{i=1}^{t(n)}\lambda^{(1-\varepsilon)k_i}\cdot \ell_i^{d} \le O(2^{(1-\delta)n}) \, .\]
    \end{itemize}
\end{definition}
It particular, if a parameterized problem $P$ is $\lambda^k$-SETH-hard, then any algorithm solving $P$ in time $\lambda^{(1-\varepsilon)k} |x|^d$ implies an algorithm solving \problem{$q$-SAT} in time $O(2^{(1-\delta(\eps))n})$ for all~$q$, thus, breaking SETH.

\subsection{Polynomial Formulations}
In this work, we consider polynomials over a ring~$R$, where $R$ is typically $\Z$ or $\Z_p$ for a prime number $p$. By a~\emph{family of~polynomials~$\mathcal P$}
we~mean an~infinite sequence
of~polynomials $P_{i_1}, P_{i_2}, \dotsc, $ such that $i_1<i_2<\dotsb$ and $P_n$ is a~multivariate polynomial
depending on $n$~variables.
We say that $\mathcal P$~has \emph{degree~$d(n)$}
if, for every~$n$, every monomial of~$P_n$ has total degree
at~most~$d(n)$.

\begin{definition}[$\Delta$-explicit family of~polynomials]
	\label{def:explicitfamily}
	For a~constant~$\Delta$, we~say that $\mathcal P$~is \emph{$\Delta$-explicit}, if, for all~$n$,
	the degree of~$P_n$ is~at~most $\Delta$ and all coefficients of~$P_n$ can be~computed (simultaneously) in~time $O(n^{\Delta})$.
\end{definition}

Equipped with this definition, we're in a position to state the required properties of polynomial formulations that will allow us to prove barriers for hardness proofs. For the case of non-parameterized problems, we define polynomial formulations as follows.

Let $A$ be a computational problem, and for every $n\in\N$, let $I_n$ be the set of instances of~$A$ of size~$n$. A polynomial formulation of a computational problem~$A$ is a $\Delta$-explicit family of polynomials ${\mathcal P}=(P_s)_{s\geq1}$ and a family of maps $\phi=(\phi_n)_{n\geq1}$ where $\phi_n\colon I_n \to \Z^{s(n)}$ satisfying the following. In order to check if $x\in I_n$ is a yes instance of~$A$, it suffices to map $y=\phi(x) \in \Z^{s(n)}$ and evaluate the corresponding polynomial $P_{s(n)}(y)$.
\begin{definition}[Polynomial formulations]
    Let $A$ be a computational problem and for every $n\in\N$, let $I_n$ be the set of instances of~$A$ of size~$n$. Let $\Delta$ be a constant, $T \colon \mathbb N \to \mathbb N$ be a~time bound, and $\mathcal P=(P_1,P_2,\dotsc)$ be a~\mbox{$\Delta$-explicit} family of~polynomials  over $\Z$. We say that $\mathcal P$
    is~a~\emph{$\Delta$-polynomial formulation of~$A$} of complexity~$T$, if there exist
    \begin{itemize}
        \item a non-decreasing function $s \colon \mathbb N \to \mathbb N$ satisfying $s(n) \le T(n)$, and an algorithm computing $s(n)$ in~time $T(n)$;
        \item a family of mappings $\phi=(\phi_1, \phi_2, \dotsc)$, where $\phi_n \colon I_n \to \Z^{s(n)}$, and an algorithm evaluating $\phi_n$ at any point in time $T(n)$
    \end{itemize}
    such that the following holds.
    For every~$n\in\N$ and every $x \in I_n$,
    \begin{itemize}
    		\item $P_{s(n)}(\phi_{n}(x)) \neq 0 \Leftrightarrow x \text{ is a yes instance of } $A$\,;$
    		\item $|P_{s(n)}(\phi_n(x))| < 2^{s(n)}$.
    \end{itemize}
\end{definition}

In order to define polynomial formulations of \emph{parameterized} problems, we need to make the following changes. An instance of the problem is now a pair $(x,k) \in I_n \times \N$, where $k$ is the parameter. The function $s$ now depends on $x$ and the value of~$k$. Similarly, each $\phi_n$ now also depends on~$k$. The time bounds on the evaluation of $s$ and $\phi$ are now $T(k) |x|^{O(1)}$ rather than $T(n)$. Since the size~$n$ of the instance~$x$ doesn't appear in these time bounds anymore (it's now replaced by the length of the bit representation $|x|$ of $x$), we don't need the index~$n$ in~$\phi_n$, and we merge the functions $(\phi_1,\phi_2,\ldots)$ into one function~$\phi$.

\begin{definition}[Polynomial formulations of parameterized problems]
Let $A$ be a parameterized computational problem and let $\mathcal I \times \N$ be the set of all instances of~$A$, where for an instance $(x,k)\in \mathcal I \times N,\, k$ is the value of the parameter. Let $\Delta$ be a constant, $T \colon \mathbb N \to \mathbb N$ be a~time bound, and $\mathcal P=(P_1,P_2,\dotsc)$ be a~\mbox{$\Delta$-explicit} family of~polynomials over $\Z$. We say that $\mathcal P$
    is~a~\emph{$\Delta$-polynomial formulation of~$A$} of complexity~$T$, if there exist
    \begin{itemize}
        \item a function $s \colon \mathcal I \times \mathbb{N} \to \mathbb N$ satisfying $s(x, k) \le T(k) |x|^{\Delta}$, and an algorithm computing $s(x, k)$ in~time $T(k) |x|^{\Delta}\,;$
        \item a function $\phi \colon \mathcal I \times \mathbb{N} \to \Z^*$ such that $\phi(x, k) \in \mathbb{Z}^{s(x, k)}$, and an algorithm computing $\phi(x, k)$ in~time~$T(k) |x|^{\Delta}\,$
    \end{itemize}
    such that the following holds. For every~$(x,k)\in \mathcal I \times \N$,
    \begin{itemize}
    	    \item $P_{s(x,k)}(\phi(x,k)) \neq 0 \Leftrightarrow (x,k) \text{ is a yes instance of } $A$\,;$
    	    \item $|P_{s(x, k)}(\phi(x, k))| < 2^{T(k)|x|^{\Delta}}\,.$
    \end{itemize}
\end{definition}

\subsection{Computational Problems}\label{sec:problems}
In this work, we~show barriers to proving hardness for the following non-parameterized and parameterized problems.
\paragraph{Non-parameterized problems.}\label{sec:problemsexact}
For each problem below, the specified parameter~$n$
is used as the~default size measure
when bounding the complexity of the problem.
It~is well known that each of these problems
can be~solved in time $2^{O(n)}$.
\begin{itemize}
    \item \problem{$k$-SAT}: given a~formula~$F$ in $k$-CNF over $n$~variables,
    check if $F$~has a~satisfying assignment.
    \item \problem{MAX-$k$-SAT}: given a~formula~$F$ in $k$-CNF over $n$~variables and an~integer~$t$, check if it~is possible to~satisfy at~least $t$~clauses of~$F$.
    \item \problem{Hamiltonian Path}: given a~directed graph~$G$ with $n$~nodes, check whether $G$~contains a~cycle visiting every node exactly once.
    \item \problem{Graph Coloring}: given a~graph~$G$ with $n$~nodes and an~integer~$t$, check whether $G$~can be~colored properly using at~most $t$~colors.    
    \item \problem{Set Cover}: given a~set family $\mathcal F \subseteq 2^{[n]}$ of size $n^{O(1)}$ and
    an~integer~$t$, check whether one can cover~$[n]$ with at~most $t$~sets from~$\mathcal F$.
    \item \problem{Independent Set}: given a~graph~$G$ with $n$~nodes and an~integer~$t$, check whether $G$~contains an~independent set of~size at~least~$t$.
	\item \problem{Clique}: given a~graph~$G$ with $n$~nodes and an~integer~$t$, check whether $G$~contains a~clique of~size at~least~$t$.
    \item \problem{Vertex Cover}: given a~graph~$G$ with $n$~nodes and an~integer~$t$, check whether $G$~contains a~vertex cover of~size at~most~$t$.
	\item \problem{$3$d-Matching}: given a~3-uniform 3-partite hypergraph~$G$ with 
	parts of~size~$n$ and an~integer~$t$, check whether $G$~contains a~matching 
	of~size at~least~$t$.
\end{itemize}
\paragraph{Parameterized problems.}\label{sec:problemsparam}
Each of~the problems below comes with a~parameter~$k$ and we~are interested
to~know how the complexity of the problem grows as a~function of the input length
and~$k$. We~say that a~problem with a~parameter~$k$ 
belongs to~the class \FPT{} if it~can be~solved
in~time $O^*(f(k))$ for some computable function~$f$. Similarly to the case of non-parameterized problems, we denote the size of an instance~$x$ of a problem (the number of nodes in the input graph, the number of variables in the input formula) by~$n$, and we denote the length (the length of the binary representation of the instance~$x$) by $|x|$.
\begin{itemize}
    \item \problem{$k$-Path}: given a~graph~$G$, check whether 
    $G$~contains a~simple path with $k$~nodes.
    \item \problem{$k$-Vertex Cover}: given a~graph~$G$, check whether $G$~contains a~vertex cover of~size at~most~$k$.
    \item \problem{$k$-Tree}: given a~graph~$G$ and a~tree~$T$ with $k$~nodes, check whether there exists a~(not necessarily induced) copy of~$T$ in~$G$.
    \item \problem{$k$-Steiner Tree}: given a~graph~$G(V,E)$ with (integer non-negative) edge weights and a~subset $T \subseteq V$ of~its nodes of~size~$k$, and an~integer $0 \le t \le |V|^{O(1)}$,
    check whether there is a~tree in~$G$ of~weight at~most~$t$ containing all nodes from~$T$.
    \item \problem{$k$-Internal Spanning Tree}: given a~graph~$G$,
    check whether there is a~spanning tree of~$G$ with at~least $k$~internal nodes.
    \item \problem{$k$-Leaf Spanning Tree}: given a~graph~$G$,
    check whether there is a~spanning tree of~$G$ with at~least $k$~leaves.
    \item \problem{$k$-Nonblocker}: given a~graph~$G$, check whether $G$~contains a~subset of nodes of~size at~least~$k$ whose complement is a~dominating set in~$G$.
    \item \problem{$k$-Path Contractibility}: given a~graph~$G$, check whether it~is
    possible to~contract at~most $k$~edges in~$G$ to~turn it~into a~path.
    \item \problem{$k$-Cluster Editing}: given a~graph~$G$, check whether it~is 
    possible to~turn~$G$ into a~cluster graph (a~set of~disjoint cliques) using 
    at~most $k$~edge modifications (additions and deletions).
    \item \problem{$k$-Set Splitting}: given a~set family $\mathcal F \subseteq 2^{[n]}$ of size $n^{O(1)}$, check whether there exists a~partition of~$[n]$ into two sets
    that splits at~least $k$~sets from $\mathcal F$.
\end{itemize}

\section{Minimum Arithmetic Circuit Problem}\label{sec:macp}
In~this section, we~show that for polynomials of~constant degree one can
find arithmetic circuits of size close to~optimal in nondeterministic polynomial time.

\begin{definition}
    Let $\mathcal P=(P_1,P_2, \dotsc)$ be a~family of polynomials.
    The \emph{minimum arithmetic circuit problem for $\mathcal P$}, denoted by $\operatorname{MACP}_{\mathcal P}(n,s)$, is: given $n, s \in \mathbb N$, find an~arithmetic circuit of~size at most~$s$ computing~$P_n$, or~report that there is no~such circuit.
\end{definition}

It~is known that when $\mathcal P$ is the family of permanent polynomials,
$\operatorname{MACP}_{\mathcal P}$ can be~solved
in time $(ns)^{O(1)}$ either by an~MA-protocol or by a~nondeterministic Turing
machine with an~oracle access to~the polynomial identity testing problem
(PIT)~\cite{DBLP:journals/cc/KabanetsI04}. In~our setting,
we do~not have oracle access to~PIT
nor do~we have any randomness. To~make~up for that, we consider the following approximate version of $\operatorname{MACP}$.

\begin{definition}
    Let $p$~be a~prime number, $\mathcal P=(P_1,P_2, \dotsc)$ be a~family of polynomials over $\Z_p$, and $c \ge 1$ be an~integer parameter.
    The problem $\operatorname{Gap-MACP}_{\mathcal P, c, p}(n, s)$ is: given $n, s \in \mathbb N$, output an~arithmetic circuit over $\Z_p$ of~size at most~$cs$ computing~$P_n$, if $P_n$ can be~computed in~$\Z_p$ by a~circuit of size at~most~$s$;
    output anything otherwise.
\end{definition}
	The reason we~allow an~abitrary output in~case $P_n$
	does not have a~circuit of~size~$s$ is~the following.
	Right after solving the $\operatorname{Gap-MACP}$ problem, we~will verify 
	that the found {circuit} is~correct. Thus, 
	if instead of a~circuit of~size~$cs$ 
	computing~$P_n$, we~are given a circuit that doesn't compute $P_n$ correctly,
	we~will reject~it in~the verification stage.
	The parameter~$p$ here is~needed to~have control over the maximum
	value of~coefficients when expanding the guessed circuit as
	a~polynomial.

We will use the following result proven by Strassen~\cite{strassen1973vermeidung}
    (see also \cite[Chapter~7.1]{DBLP:books/daglib/0090316} or
    \cite[Theorem~2.2]{DBLP:journals/fttcs/ShpilkaY10}).
Recall that a polynomial is \emph{homogeneous} if all its monomials have the same degree. We say that a circuit is homogeneous if all its gates compute homogeneous polynomials. For a polynomial~$P$, the homogeneous part of~$P$ of degree~$i$ is the sum of all monomials of $P$ of degree exactly~$i$.
\begin{theorem}[{\cite{strassen1973vermeidung}}]\label{lem:homogeneous}
There exists a constant $\mu'>0$ such that the following holds.
If a \mbox{degree-$\Delta$} polynomial~$P$ can be computed by an arithmetic circuit of size $s$, then there exists a homogeneous circuit~$C'$ of size at most $\mu' \Delta^2 s$ computing~$P$ such that the $\Delta+1$ outputs of~$C'$ compute the homogeneous parts of~$P$.
\end{theorem}
We use \cref{lem:homogeneous} to conclude that at the expense of increasing the circuit size by a factor of~$O(\Delta^2)$, we can assume that an arithmetic circuit computing a degree-$\Delta$ polynomial contains only gates computing polynomials of degree at most~$\Delta$.
\begin{corollary}\label{cor:bounded-degree}
    There exists a constant $\mu>0$ such that the following holds.
If a degree-$\Delta$ polynomial~$P$ can be computed by an arithmetic circuit of size $s$, then $P$ can be computed by a (single-output) circuit~$C$ of size at most $\mu \Delta^2 s$ such that all gates of~$C$ compute polynomials of degree at most~$\Delta$.
\end{corollary}
\begin{proof}
In order to construct the circuit~$C$ we take the homogeneous circuit~$C'$ guaranteed by \cref{lem:homogeneous}, remove all gates computing polynomials of degree greater than~$\Delta$, and sum up all $\Delta+1$ output gates of~$C'$ in the output of~$C$. Since sums and products of degree-$(\Delta+1)$ homogeneous polynomials can't compute non-trivial polynomials of degree $\leq\Delta$, removing gates computing polynomials of degree greater than~$\Delta$ doesn't affect the output gates of~$C'$.
Since the outputs of $C'$ compute the homogeneous parts of~$P$, the output of~$C$ computes~$P$, which finishes the proof of the corollary.
\end{proof}
    
We now prove that for polynomials of bounded degree, $\operatorname{Gap-MACP}$ can be solved in non-deterministic polynomial time.
\begin{lemma}
    \label{lemma:macp}
    There exists a~constant $\mu>0$ such that for every $\Delta$-explicit family of~polynomials~$\mathcal P$ and every prime number $p$, 
    \[\operatorname{Gap-MACP}_{\mathcal P, \mu \Delta^2, p}(n, s) \in \NTIME[O(\Delta^{2}sn^{2\Delta}\log^2p)] \, .\]
\end{lemma}
\begin{proof}
    We present a non-deterministic algorithm that, given a polynomial~$P$ of circuit complexity~$s$, finds a circuit of size at most $cs$ for $c=\mu\Delta^2$. 
    
    First we note that \cref{cor:bounded-degree} guarantees the existence of a circuit~$C$ over~$\Z_p$ of size $cs$ computing~$P$ such that each gate of~$C$ computes a polynomial of degree at most~$\Delta$. 
    
    We non-deterministically guess such a circuit $C$, and verify if it computes $P$ correctly. If it does, we output~$C$, and we output an empty circuit otherwise. It remains to show that in the specified time we can verify that $C$ computes~$P$. 
    To~do this, we start with the circuit inputs and proceed to its output, and~write down the polynomial over~$\Z_p$ computed by~each gate as~a~sum or~product of~polynomials of~its input gates.
    There are at most $2n^{\Delta}$ monomials in a~polynomial of~degree at~most~$\Delta$. Computing sums and products of~such polynomials boils down 
    to~at~most $O(n^{2\Delta})$ arithmetic operations with their coefficients. As~every coefficient of a~polynomial over $\mathbb{Z}_p$ is~specified by~$\log p$ bits, any such arithmetic operation takes time $O(\log^2 p)$.
	Putting it all together, we~expand each of the~$\mu \Delta^2 s$ gates,
    expanding each gate takes $O(n^{2\Delta})$ arithmetic operations, each arithmetic operation takes time $O(\log^2p)$. Thus, 
    the total time of~expanding~$C$ in~$\mathbb{Z}_p$ is
	\[O(\mu \Delta^2 s \cdot n^{2\Delta} \cdot \log^2p) \, .\]

    To~nondeterministically solve $\operatorname{Gap-MACP}_{\mathcal P, \mu \Delta^2, p}(n,s)$, 
    we~guess~$C$ and expand~it in~$\mathbb{Z}_p$ as~discussed above.
    Since $P_n$ is from a~$\Delta$-explicit family~$\mathcal P$,
	it~can be~written as a~sum of~monomials in~time $O(n^{\Delta})$ (recall \cref{def:explicitfamily}).
	Then, it~remains to~compare the coefficients of the two sequences of~monomials.
\end{proof}

\section{Main Results}\label{sec:mainresultsection}
In this section, we state the main results of this paper. First, we state  \cref{lemma:polynomialformulation,lemma:parampolynomialformulation} asserting that every problem defined in \cref{sec:problems} admits polynomial formulations, we'll prove these lemmas in \cref{sec:formulations,sec:paramformulations}. Then, in \cref{sec:mainresult,sec:mainresultparam} we prove barriers to proving hardness results for non-parameterized and parameterized problems admitting polynomial formulations. Finally, we conclude that such barriers hold for for all problems from \cref{sec:problems}.

\begin{restatable}{lemma}{polynomialformulation}
    \label{lemma:polynomialformulation}
    For every $c>1$, there is $\Delta=\Delta(c)$, such that each of the following problems
    \begin{quote}
        \problem{$k$-SAT}, 
        \problem{MAX-$k$-SAT}, 
        \problem{Hamiltonian Path},
        \problem{Graph Coloring},
        \problem{Set Cover},
        \problem{Independent Set},
        \problem{Clique},
        \problem{Vertex Cover},
        \problem{$3$d-Matching}
    \end{quote}
    admits a~$\Delta$-polynomial formulation of~complexity~$c^n$.
\end{restatable}

\begin{restatable}{lemma}{parampolynomialformulation}
    \label{lemma:parampolynomialformulation}
    For every $c>1$, there is $\Delta=\Delta(c)$, such that each of the following parameterized problems
    \begin{quote}
        \problem{$k$-Path}, 
        \problem{$k$-Vertex Cover}, 
        \problem{$k$-Tree}, 
        \problem{$k$-Steiner Tree}, 
        \problem{$k$-Internal Spanning Tree}, 
        \problem{$k$-Leaf Spanning Tree}, 
        \problem{$k$-Nonblocker}, 
        \problem{$k$-Path Contractibility}, 
        \problem{$k$-Cluster Editing}, 
        \problem{$k$-Set Splitting}
    \end{quote}
    admits a~$\Delta$-polynomial formulation of~complexity~$c^k$.
\end{restatable}

\subsection{Non-parameterized Problems}\label{sec:mainresult}
In the theorem below we prove that if a problem admits constant-degree polynomial formulations of complexity $2^{\gamma n}$ for every $\gamma>0$, then SETH-hardness of the problem would imply a circuit lower bound.
\begin{theorem}\label{thm:main}
Let $A$ be a computational problem. Assume that for every $c>1$, there is $\Delta=\Delta(c)$ such that $A$ admits a~$\Delta$-polynomial formulation of~complexity~$c^n$. If $A$ is $\lambda^n$-SETH-hard for a constant $\lambda>1$, then at~least one of~the following circuit lower bounds holds:
	\begin{itemize}
		\item $\E^{\NP}$ requires series-parallel Boolean circuits of~size~$\omega(n)$;
		\item for every constant $\gamma>1$, there exists an~explicit family of constant-degree polynomials over~$\Z$ that requires arithmetic circuits of~size $\Omega(n^{\gamma})$.
	\end{itemize}
\end{theorem}

\begin{proof}
Let $\lambda>1$ be the constant from the theorem statement, $\gamma>1$ be an arbitrary constant, and $\sigma=\log(\lambda)/(6\gamma)$. For $n\in\N$, let $I_n$ be the set of all instances of~$A$ of size~$n$. Let $\mathcal P$ be a~$\Delta$-polynomial formulation of~$A$ of~complexity $2^{\sigma n}$,
	for constant $\Delta=\Delta(\sigma)>0$.
	We assume that $A$ is $\lambda^n$-SETH-hard: there is a function $\delta\colon\R_{>0}\to\R_{>0}$ such that for every $k\in\N$, $(\problem{$k$-SAT},2^n)\leq_{\delta} (P,\lambda^n)$. 
	
	We'll prove that at least one of the two circuit lower bounds holds. 
	If ${\mathcal P}=(P_t)_{t\geq1}$ does not have arithmetic circuits over $\Z$ of size $t^\gamma$ for infinitely many values of~$t$, then we have an explicit family of constant-degree polynomials that requires arithmetic circuits of~size $\Omega(t^{\gamma})$. 
	Hence, in the following we assume that  ${\mathcal P}$ has arithmetic circuits over $\Z$ of size $ct^\gamma$ for all values of~$t$ for a constant~$c>0$. Under this assumption, we design a non-deterministic algorithm solving \problem{$k$-TAUT} in time $2^{(1-\eps)n}$ for every~$k$. This contradicts NSETH and, by \cref{thm:nseth}, implies a super-linear lower bound on the size of series-parallel circuits computing~$\E^{\NP}$.
	
Let $\delta_0=\delta(1/2)\in(0,1)$ where $\delta$ is the function from the SETH-hardness reduction for~$A$. Let $\alpha=\frac{1}{\gamma+2\Delta+8}$, $L=2(1-\delta_0) \alpha n/\log(\lambda)$, and $T=2^{\sigma L}$. The meaning of these constants is the following. We will start with an instance of the \problem{$k$-TAUT} problem on~$n$ variables, reduce it to $2^{(1-\alpha)n}$ instances of \problem{$k$-TAUT} on $\alpha n$ variables each. Then we'll use the fine-grained reduction from \problem{$k$-SAT} to the problem~$A$ on instances of size~$\ell\leq L$. Finally, we'll use the polynomial formulation of~$A$ to reduce instances of size~$\ell$ to polynomials with $t\leq T$ variables.

Let $F$~be a~$k$-DNF formula over $n$~variables. In order to solve~$F$, we branch on all but $\alpha n$~variables. This gives us $2^{(1-\alpha)n}$ $k$-DNF formulas. By~solving \problem{$k$-SAT} on the negations of all of~these formulas, we~solve \problem{$k$-TAUT} on the original formula~$F$. 

We now apply the
    $(\problem{$k$-SAT}, 2^n)\leq_{\delta}(A, {\lambda}^n)$ fine-grained reduction to (the negation of) each of~the resulting formulas, which gives us a number of instances of~$A$.
    Let $\ell$~be the largest size of these instances. 
    From \cref{def:seth-hardness}, we~know that for $\eps=1/2$ and $\delta_0=\delta(1/2)>0$, ${\lambda}^{\ell/2}=\lambda^{(1-\eps)\ell} < 2^{(1-\delta_0) \alpha n}$, so each instance of~$A$ indeed has size less than $\ell<2(1-\delta_0) \alpha n/\log(\lambda)=L$. 

    Since $\mathcal P$ is a polynomial formulation of~$A$ of~complexity $2^{\sigma n}$, there exist $s \colon \N \to \N, s(\ell)\leq 2^{\sigma \ell}$ and $\phi=(\phi_1, \phi_2, \dotsc)$ (computable in time $2^{\sigma L}$) such that for every~$\ell$ and every $x\in I_\ell$,
	\begin{itemize}
		\item $P_{s(\ell)}(\phi_{\ell}(x)) \neq 0$ iff $x$ is a~yes instance of~$A$;
		\item $|P_{s(\ell)}(\phi_\ell(x))|<2^{s(\ell)}$.
	\end{itemize}
    Using~$\mathcal P$,
    we will solve all instances
    of~$A$ in~two stages: in~the preprocessing stage (which takes place before all the reductions), we~guess efficient arithmetic circuits for polynomials~$P_{t}$ for all 
    $t \le T$,
    in the solving stage, we~solve all instances of~$A$ using the guessed circuits. Note that we'll be using the polynomials to solve instances of~$A$ resulting from \problem{$k$-SAT} instances on $\alpha n$ variables. Since~$L$ is the largest size of such an instance of~$A$, we have that each such instance is mapped to a polynomial with at most $T=s(L)\leq 2^{\sigma L}$ variables. Therefore, finding efficient arithmetic circuits for polynomials~$P_{t}$ for all 
    $t \le T$ will be sufficient for solving the \problem{$k$-SAT} instances of size $\alpha n$.

        \paragraph{Preprocessing.} For every $t \leq T$, we find a prime $p_t$ in the interval $2^{t+1} \le p_t \le 2^{t+2}$ in non-deterministic time~$O(t^7)$~\cite{aks04,lp19}. 
        
        Now for every $t\leq T$, we reduce all coefficients of the polynomial $P_t$ modulo~$p_t$ to obtain a polynomial $Q_t$ over $\Z_{p_t}$, and let ${\mathcal Q}=(Q_1,Q_2,\ldots)$.
        For every $t \le T$, we now non-deterministically solve $\operatorname{Gap-MACP}_{\mathcal Q, \mu \Delta^2, p_t}(t, ct^{\gamma})$ using \cref{lemma:macp}.
        Since we assume that $\mathcal P$ has arithmetic circuits over $\Z$ of size $ct^\gamma$, we have that $\mathcal Q$ has arithmetic circuits over $\Z_{p_t}$ of this size. Thus, we~obtain arithmetic circuits~$C_t$ of~size at~most
        \begin{equation}\label{eq:circuitsize}
            c\mu \Delta^2t^{\gamma}
        \end{equation}
        computing~$Q_{t}$ over $\Z_p$ for all $t \le T$. Since~$C_t$ computes~$Q_t$ correctly in~$\mathbb{Z}_p$ and $|P_{s(\ell)}(\phi_\ell(x))| < 2^{s(\ell)} \leq  p_{s(\ell)}/2$ for all $x \in I_\ell$, we can use $C_t$
        to~solve~$A$ for every instance size $\ell \leq L$. By~\cref{lemma:macp}, $\operatorname{Gap-MACP}_{\mathcal Q, \mu \Delta^2, p_t}(t, ct^{\gamma})$ can be solved in (non-deterministic) time
        \[
        	O\left(\Delta^2 \cdot ct^{\gamma} \cdot t^{2\Delta} \cdot \log^2(p_t)\right)=
        		O\left(T^{\gamma+2\Delta+2}\right) \, .
        \]
        The total (non-deterministic) running time of the preprocessing stage is then bounded from above by the time needed to~find~$T$ prime numbers, write down the corresponding explicit polynomials modulo~$p_t$, and solve~$T$ instances of $\operatorname{Gap-MACP}$:
        \begin{align}\label{eq:preprocessing}
        	O\left(T(T^7+T^{\Delta+2}+T^{\gamma+2\Delta+2})\right)=
        		O\left(T^{\gamma+2\Delta+8}\right)
        		=O\left(2^{(1-\delta_0)n}\right) \, ,
        \end{align}
        where the last equality holds due to $T=2^{\sigma L}$, $L=2(1-\delta_0) \alpha n/\log(\lambda)$, $\sigma=\log(\lambda)/(6\gamma)$, and $\alpha= \frac{1}{\gamma+2\Delta+8}$.
        \paragraph{Solving.} In~the solving stage, we~solve all $2^{(1-\alpha)n}$
        instances of~\problem{$k$-SAT} by~reducing them to~$A$ and using
        efficient circuits found in~the preprocessing stage. For an instance~$x$ of~$A$ of size $\ell$, we first transform it into an input of the polynomial $y=\phi_\ell(x)\in\Z^{s(\ell)}$. Both $s(\ell)$ and $\phi_\ell(x)$ can be computed in time $O(2^{\sigma \ell})$.  Then we feed it into the circuit $Q_{s(\ell)}$. First we note that we have the circuit $Q_{s(\ell)}$ after the preprocessing stage as $s(\ell)\leq s(L) \leq 2^{\sigma L}= T$ and we have circuits $(Q_1,\ldots,Q_T)$. The number of arithmetic operations in~$\Z_{p_{s(\ell)}}$ required to evaluate the circuit is proportional to the circuit size, and each arithmetic operation takes time $\log^2(p_{s(\ell)})=O(s(\ell)^2)$. From~\eqref{eq:circuitsize} with $t\leq s(\ell)\leq2^{\sigma \ell}$, we have that we can solve an instance of~$A$ with $\ell$ inputs in time 
        \[
        O(2^{\sigma \ell}) + c\mu \Delta^2 \cdot s(\ell)^2 \cdot 2^{\sigma\gamma \ell}
        =O\left( 2^{2\sigma \ell+\sigma\gamma \ell}\right)
        =O\left( 2^{3\sigma\gamma \ell}\right)
        =O\left(\lambda^{\ell/2}\right)\,,
        \]
        where the last equality holds due to the choice of $\sigma=\log(\lambda)/(6\gamma)$.
        The~fine-grained reduction from \problem{$k$-SAT} to~$A$ implies that a~$O\left(\lambda^{n/2}\right)$-time algorithm for~$A$ gives us a~$O\left(2^{n(1-\delta_0)}\right)$-time algorithm for \problem{$k$-SAT}. Thus, since we solve each $\ell$-instance of~$A$ resulting from $2^{(1-\alpha)n}$ instances of \problem{$k$-SAT} in time $O\left(\lambda^{\ell/2}\right)$, we solve the original $n$-variate instance~$F$ of \problem{$k$-TAUT} in time
    \begin{align}\label{eq:solving}
    O\left(2^{(1-\alpha) n} \cdot (2^{\alpha n})^{1-\delta_0}\right)=O\left(2^{n(1-\alpha\delta_0)}\right) \,.
    \end{align}
    
    The total running time of the preprocessing and solving stages (see~\eqref{eq:preprocessing} and~\eqref{eq:solving}) is bounded from above by $O\left(2^{n(1-\delta_0)}\right)+O\left(2^{n(1-\alpha\delta_0)}\right)=O\left(2^{n(1-\alpha\delta_0)}\right)$, which refutes NSETH, and implies a super-linear lower bound for Boolean series-parallel circuits.
\end{proof}

We now apply \cref{thm:main} to the non-parameterized problems from \cref{sec:problemsexact} to prove \cref{thm:intromain}.
\intromain*
\begin{proof}
    This follows immediately from \cref{lemma:polynomialformulation} and \cref{thm:main}.
\end{proof}

\subsection{Parameterized Problems}\label{sec:mainresultparam}
In the next theorem we show that if a parameterized problem admits constant-degree polynomial formulations of complexity $2^{\gamma k}$ for every $\gamma>0$, then SETH-hardness of this problem would imply a circuit lower bound. The proof of \cref{thm:mainparam} follows the high level strategy of the proof of \cref{thm:main}, but takes into account the dependence on the parameter~$k$ of the parameterized problem under consideration and (arbitrary) polynomial dependence on the input length, we present the proof in \cref{apx:mainparam}.
\begin{restatable}{theorem}{mainparam}\label{thm:mainparam}
Let $A$ be a parameterized computational problem with a parameter~$k$. Assume that for every $c>1$, there is $\Delta=\Delta(c)$ such that $A$ admits a~$\Delta$-polynomial formulation of~complexity~$c^k$. If $A$ is $\lambda^k$-SETH-hard for a constant $\lambda>1$, then at~least one of~the following circuit lower bounds holds:
	\begin{itemize}
		\item $\E^{\NP}$ requires series-parallel Boolean circuits of~size~$\omega(n)$;
		\item for every constant $\gamma>1$, there exists an~explicit family of constant-degree polynomials over~$\Z$ that requires arithmetic circuits of~size $\Omega(n^{\gamma})$.
	\end{itemize}
\end{restatable}

We apply \cref{thm:mainparam} to the parameterized problems from \cref{sec:problemsparam} to prove \cref{thm:intromainparam}.
\intromainparam*
\begin{proof}
    This follows immediately from \cref{lemma:parampolynomialformulation} and \cref{thm:mainparam}.
\end{proof}
While \cref{thm:intromainparam} conditionally rules out $\lambda^k$ lower bounds for certain parameterized problems, we remark that the same machinery can be applied to conditionally rule out lower bounds of the form $n^{\lambda k}$ for constant $\lambda>0$. We do not include rigorous proofs of such results in the paper as this would require us to generalize \cref{thm:mainparam} to work with functions of~$k$ that may have different forms (exponential in $k$ or in~$k\log{n}$) at the expense of clarity of presentation. We note that the same techniques show that $n^{\gamma k}$-SETH hardness of \problem{$k$-Clique} or \problem{$k$-Independent Set} for a constant $\gamma>0$ would also imply one of the two circuit lower bounds. The polynomial formulations of parameterized \problem{$k$-Clique} and \problem{$k$-Independent Set} are identical to the polynomial formulations of the non-parameterized \problem{Independent Set} problem presented in \cref{sec:formulations} with the only difference that the sizes of the sets $S$ are now bounded by $2k/\theta$ instead of $2n/\theta$. This leads to $\binom{n}{\leq 2k/\theta}=n^{O(k)}$ variables in polynomial formulations and rules out $n^{\lambda k}$ lower bounds for the parameterized versions of \problem{$k$-Clique} or \problem{$k$-Independent Set}.

\section{Deterministic Splitters over Alphabets of Linear Size}\label{sec:splitters}
Our polynomial formulations of some of the problems (such as \problem{$k$-Path} and \problem{$k$-Tree}) will require \emph{deterministic} constructions of certain splitters. This section is devoted to designing such splitters. 
\begin{definition}
An $(n,k,\ell)$-splitter $H$ is a family of functions $f\colon[n]\to[\ell]$ such that for every set $S\subseteq[n]$ of size $|S|=k$, there exists a function $f\in H$ that splits $S$ evenly:
\[
\forall j\in[\ell],\;\; \floor{k/\ell} \leq f^{-1}(j) \leq \ceil{k/\ell} \;.
\]
\end{definition}
The set~$[\ell]$ in this definition is called the alphabet. If $\ell\geq k$, an $(n,k,\ell)$-splitter $H$ is a family of functions from $[n]$ to $[\ell]$ such that for every $S\subseteq[n],\; |S|=k$, there exists an $f\in H$ which is injective on~$S$. If $\ell=k$, then such a splitter is called a family of perfect-hash functions.

In this section, we present $(n,k,ck)$-splitters of size $\widetilde{O}(e^{\frac{k}{c}(1+o(1))})$ that can be computed in deterministic time $2^{9k}n^{O(1)}$.

It is easy to verify that a random set of $\approx e^{\frac{k}{c}}k\log{n}$ functions forms an $(n,k,ck)$-splitter with high probability. It's known~\cite{F84,A86} that a good linear code over the alphabet $[\ell]$ with relative distance $1-\Theta(1/k^2)$ implies a splitter with related parameters. \cite{F84,A86} use this observation to deterministically construct splitters of size $k^{O(1)}\log{n}$ for alphabets of size $\ell\geq k^2$. Although we can't use this splitter directly as we're working with alphabets of size $\ell=ck\ll k^2$, we'll use this primitive as one of the building blocks.

\cite[Theorem 3(iii)]{NSS95} gives a deterministic $(n,k,ck)$-splitter of size $\widetilde{O}(e^{k(1+o(1))})$. We follow the high-level approach of \cite{NSS95} with certain low-level modifications to design a splitter of size $\widetilde{O}(e^{\frac{k}{c}(1+o(1))})$.

Our final construction of an $(n,k,ck)$-splitter will be a certain composition of splitters with various parameters. First, we give three auxiliary constructions of splitters with different parameters that will later be used in our main construction.

We say that an $(n,k,\ell)$-splitter is explicit if the truth table of every function can be computed in deterministic time $(n\ell)^{O(1)}$. In particular, all functions of an explicit splitter~$H$ can be printed in time $|H|(n\ell)^{O(1)}$.

\subsection{\texorpdfstring{$(n,k,k^2)$-splitters}{(n,k,k2)-splitters}}
We present an efficient deterministic way to build $(n,k,k^2)$-splitters from~\cite{F84,A86} that will later effectively allow us to reduce the domain size from~$n$ to~$k^2$.

\begin{proposition}[\cite{F84,A86}]\label{prop:splitter1}
There is an explicit $(n,k,k^2)$-splitter $A(n,k,k^2)$ of size $O(k^6\log{k}\log{n})$.
\end{proposition}
\begin{proof}
There exist explicit linear codes~\cite{ABNNR92} over the alphabet~$[k^2]$ with at least~$n$ codewords, relative distance $\delta\geq1-2/k^2$, and length $m=O(k^6\log{k}\log{n})$. Below we show that viewing such a code as a set of~$m$ functions from~$[n]$ to~$[k^2]$ gives us the desired construction of an $(n,k,k^2)$-splitter. 

Assume towards a contradiction that there exists a set~$T$ of~$k$ codewords such that for each of the~$m$ coordinates, a pair of codewords from~$T$ takes the same value at this coordinate. Then the sum of the $\binom{k}{2}$ pairwise distances between the codewords does not exceed $\binom{k}{2}\cdot m - m$. By averaging, there is a pair of codewords with distance at most
\[
\frac{\binom{k}{2}\cdot m - m}{\binom{k}{2}}=m\left(1-\frac{1}{\binom{k}{2}}\right)<m\left(1-\frac{2}{k^2}\right)\;,
\]
which contradicts the assumption $\delta\geq1-2/k^2$ on the relative distance of the code.
\end{proof}

\subsection{\texorpdfstring{$(k^2,k,\log{k})$-splitters}{(k2,k,logk)-splitters}}
Now we present explicit splitters of small size for the case of small alphabet $\ell=\log{k}$.
\begin{proposition}[{\cite[Lemma~4]{NSS95}}]\label{prop:splitter2}
There is an explicit $(k^2,k,\log{k})$-splitter $B(k^2,k,\log{k})$ of size $k^{2\log{k}}$.
\end{proposition}
\begin{proof}
Let $\ell=\log{k}$. For each sequence $0=i_0<i_1<\ldots<i_{\ell}=k^2$, define $f\colon[k^2]\to[\ell]$ by
\[
f(x)=t \;\;\;\;\; \text{iff} \;\;\;\;\; i_{t-1}<x\leq i_t
\;.
\]
This construction is explicit and has size $\binom{k^2}{\ell-1}\leq k^{2(\ell-1)}\leq k^{2\log{k}}$. In order to show that this is a $(k^2,k,\ell)$-splitter, consider a set $S=\{j_1,\ldots,j_k\}\subseteq[n]$, where $j_1<\ldots<j_k$, and note that the function $f$ defined by the set
\[
i_1=j_{k/\ell}, i_2=j_{2k/\ell},\ldots,i_{\ell-1}=j_{(\ell-1)k/\ell} \;,
\]
splits the set~$S$ evenly.
\end{proof}

\subsection{\texorpdfstring{$(k^2,k/\log{k},ck/\log{k})$-splitters}{(k2,klogk,cklogk)-splitters}}
Now we present a splitter with good parameters which is not explicit. We will later use it with small values of parameters so even though this splitter is not explicit, it will be possible to compute it in the allocated time. This primitive is based on the construction from~\cite[Theorem~2(i)]{NSS95}.

\begin{lemma}\label{lem:splitter3}
There is an $(n,k,ck)$-splitter $C(n,k,ck)$ of size $\widetilde{O}(e^{\frac{k}{c}(1+o(1))})$ that can be constructed deterministically in time $O((kn)^{3k})$.
\end{lemma}
\begin{proof}
Let~$T$ be a $k$-wise independent set of vectors of length~$n$ over the alphabet~$[ck]$. There are explicit constructions of such sets of size $|T|\leq n^k$~\cite{AS08}. 

First we show that there exists $t\in T$ which, if viewed as a function $h\colon[n]\to[ck]$, splits at least an $e^{-k/c}$ fraction of $k$-sets of any family of $k$-sets ${\cal{F}} \subseteq \binom{[n]}{k}$. Indeed, a set $S\in{\cal F}$ is split by~$h$ if $h$ is injective on~$S$. The probability that~$h$ is injective on a fixed set of size~$k$ is
\[
\left(1-\frac{1}{ck}\right)\left(1-\frac{2}{ck}\right)\cdots\left(1-\frac{k-1}{ck}\right)
\geq e^{-\frac{1}{ck}-\left(\frac{1}{ck}\right)^2-\frac{2}{ck}-\left(\frac{2}{ck}\right)^2-\ldots-\frac{k-1}{ck}-\left(\frac{k-1}{ck}\right)^2}
\geq e^{-k/c} \;.
\]
Now, we iteratively greedily pick a vector from~$T$ splitting at least $e^{-k/c}$-fraction of the remaining $k$-sets.

\paragraph{Size of the splitter.} The size of the resulting splitter is at most smallest~$t$ satisfying
\[
\binom{n}{k}(1-e^{-k/c})^t \leq 1\;.
\]
That is, $t\leq e^{k/c} k\log{n}$.

\paragraph{Running time.} The running time of each step of the greedy algorithm is at most $n\cdot\binom{n}{k}|T|$. And the total running time is at most
\[
t\cdot n\cdot \binom{n}{k}\cdot |T|
\leq e^{k/c} \cdot n^k \cdot n^k\cdot (kn)^{O(1)}
\leq (kn)^{3k}\;.
\]
\end{proof}

\subsection{Main Construction}
Equipped with the three auxiliary constructions above, we're in a position to present the main result of this section.
\begin{theorem}\label{thm:splitters}
For every $c\geq1$, there is an $(n,k,ck)$-splitter of size $O(e^{\frac{k}{c}(1+o(1))}\log{n})$ that can be constructed deterministically in time
$2^{9k}n^{O(1)}$.
\end{theorem}
\begin{proof}
Let $A=A(n,k,k^2), B=B(k^2,k,\log{k}), C=C(k^2,k/\log{k},ck/\log{k})$ be the splitters from \cref{prop:splitter1,prop:splitter2,lem:splitter3}, respectively. 
Without loss of generality, we assume that $k$ is a multiple of $\log{k}$. 
We define our $(n,k,ck)$-splitter $H$ as follows. For every function $a\in A$, every function $b\in B$, and every $\log{k}$-tuple of functions $(h_1,\ldots,h_{\log{k}})$ from~$C$, $H$ contains the function $f\colon[n]\to[ck]$, where
\begin{align}\label{eq:splitters}
f(x) = \frac{ck}{\log{k}}\cdot b(a(x)) + h_{b(a(x))}(a(x)) \;.
\end{align}
\paragraph{Correctness.} Let $S\subseteq[n]$ be a set of size $|S|=k$. We will show that there exist functions $a \in A$, $b \in B$, and $(h_1,\ldots,h_{\log{k}})$ in~$C^{\log k}$ such that for $f$ defined in~\eqref{eq:splitters} and every pair of distinct $s_1,s_2\in S$, $f(s_1)\neq f(s_2)$. Equivalently, for $s\in S$, if $y=b(a(s))$ and $S_y=\{s\in S\colon b(a(s))=y\}$, then $h_y$ is injective on~$a(S_y)$.

Since $A$ and $B$ are splitters, there exist $a\in A$ and $b\in B$ such that for every $y\in[\log{k}]$, $|S_y| \leq k/\log{k}$. Now since $C$ is a splitter and $|a(S_y)|\leq|S_y|\leq k/\log{k}$, we have that for every~$y\in[\log{k}]$, there exists $h_y\in C$ such that $h_y$ is injective on $a(S_y)$. Therefore, the function~$f$ defined with the selected $a, b, h_1,\ldots,h_{\log{k}}$ satisfies the requirement that $f(s_1)\neq f(s_2)$ for all distinct $s_1,s_2\in S$.

\paragraph{Size of the splitter.} By the definition of~$H$,
\[
|H| = |A| \cdot |B| \cdot |C|^{\log{k}}
 = O(k^6\log{k}\log{n}) \cdot O(k^{2\log{k}}) \cdot \left(O\left(e^{\frac{k}{c\log{k}}(1+o(1))} \right) \right)^{\log{k}}
 = O\left(e^{\frac{k}{c}(1+o(1))}\log{n}\right) \;.
\]
\paragraph{Running time.} The splitters from \cref{prop:splitter1,prop:splitter2} are explicit, and the splitter from \cref{lem:splitter3} takes time $\left(k^2 \cdot \frac{k}{\log{k}}\right)^{\frac{3k}{\log{k}}}=O(2^{9k})$.
\end{proof}

\section{Polynomial Formulations}\label{sec:formulationssection}
\renewcommand{\descriptionlabel}[1]{\hspace{\labelsep}\textit{#1}}

In this section, we prove \cref{lemma:polynomialformulation,lemma:parampolynomialformulation}: we give polynomial formulations of all problems from \cref{sec:problems}.

\subsection{Non-parameterized Problems}\label{sec:formulations}
\polynomialformulation*
\begin{proof}
    Let $A$ be one of the problems from the list above, and for every $n\in\N$, let $I_n$ be the set of instances of~$A$ of size~$n$. We~construct a~family of mappings $\phi=(\phi_1, \phi_2, \dotsc)$, where $\phi_n \colon I_n \to \{0,1\}^{s(n)}$ and a~family of~polynomials $\mathcal P=(P_1, P_2, \dotsc)$, following the 
    same five-step pattern.

    \begin{quote}
    \begin{description}
        \item[Idea.] We~provide a~high-level idea 
        of~encoding a~problem as a~polynomial. Start by~fixing a~parameter $\theta=\theta(c)$ that will be~chosen as a~large enough constant. 
        In~the analysis, we~write $n/\theta$ instead of~$\lceil n/\theta \rceil$: this affects the bounds negligibly and
        at~the same time simplifies the bounds.
        Then, a~solution of~size~$n$ that we~are looking for can be~broken into ``blocks'' of~size $n/\theta$; for each potential block, we~introduce a~$0/1$-variable; then, for each candidate solution, we~introduce a~monomial that is non-zero if~it~is indeed a~solution.
        
        \item[Variables.] We~introduce a~set~$X$ of~$s(n)$
        variables. They are used to~specify the function~$\phi_n$
        that maps an~instance $I \in I_n$ to a~vector in~$\mathbb{Z}^{s(n)}$. To~do this, we~specify a~$0/1$-value
        that $\phi_n(I)$ assigns to~every variable~$x \in X$.
        
        \item[Complexity.] We~bound the number of~variables $s(n)$ of~the constructed polynomial~$P_{s(n)}$ as~well~as the time needed to~compute the mapping~$\phi_n$
        by~\[\left(2^{n/\theta} \cdot \binom{n}{n/\theta}\right)^{O(1)} \, .\]
        In~all the cases, it~will be~straightforward to~compute~$s(n)$ in~the allocated time.
        
        \item[Polynomial.] We~specify the polynomial $P_{s(n)}(X)$ as a~sum of~$2^{O(n)}$ monomials (where the hidden constant in $O(n)$ depends on $\theta=\theta(c)$ only), each having coefficient~$1$. It~is usually straightforward from the definition of the polynomial that 
        $I$~is a~yes-instance of~$A$
        iff $P_{s(n)}(\phi_n(I))>0$.
        
        \item[Degree.] We~show that the degree~$\Delta$ 
        of~$\mathcal P$
        depends on~$\theta$ only.
    \end{description}
    \end{quote}
    
    Below, we~show that the five steps above ensure that $\mathcal P$ is indeed a~polynomial formulation of~$A$.
    \begin{itemize}
        \item By~choosing a~large enough $\theta=\theta(c)$, we ensure that, for all large enough~$n$,
	    \[|X|=s(n)=\left(2^{n/\theta} \cdot \binom{n}{n/\theta}\right)^{O(1)} < c^n \, .\]
	    
	    \item Since $P_{s(n)}(X)$ is a~sum of~$2^{O(n)}$ monomials, computing the coefficients of all monomials in~$P_{s(n)}$ takes time $2^{O(n)}$. Since $|X|=c^n$, $2^{O(n)}=|X|^{O(1)}$.
	    Since the degree of~$\mathcal P$ is~$\Delta$,
	    $\mathcal P$ is a~$\Delta$-explicit family of~polynomials.
	    
	    \item Recall that $\phi_n$ maps an instance of the problem to a vector from~$\{0,1\}^{s(n)}$, and that all the coefficients of the polynomials $P_{s(n)}$ are ones. Thus, for every $I \in I_n$, $|P_{s(n)}(\phi_n(I))|$
		is at~most the number of~monomials in~$P_{s(n)}$,
		i.e., $|X|^{O(1)}$, and hence at~most $2^{|X|}$. Since $\phi_n(I)$ can be~computed in~time $c^n$, we~conclude that
		$\mathcal P$ is~indeed a~polynomial formulation of~$A$.
    \end{itemize}

        \paragraph{\problem{Hamiltonian Path}.} Given a~directed graph $G(V,E)$ with $n$~nodes, check whether it~contains a~Hamiltonian path.
        
        \begin{description}
            \item[Idea.] 
            One can break a~Hamiltonian path $\pi$ into $\theta$ node-disjoint paths $\pi_1, \dotsc, \pi_{\theta}$ of length $n/\theta$ each. We say that $\pi_1, \dotsc, \pi_{\theta}$ is a~valid partition iff $\pi_i$'s are simple paths of~length $n/\theta$ sharing no~nodes and, for every~$i$, there is an~edge joining the last node of~$\pi_i$ with the first node of~$\pi_{i+1}$.
            
            \item[Variables.] Introduce $s(n)=O\left(n^2\binom{n}{n/\theta}\right)$ variables:
            \[X=\{x_{S,u,v} \colon S \subseteq V, |S|=n/\theta, u \in S, v \in V \setminus S\} \, .\]
            The mapping~$\phi_n(G)$ 
            assigns the following $0/1$-value 
            to a~variable $x_{S,u,v}$:
            \[[\text{there is a~Hamiltonian path in $G[S]$ that starts at~$u$ and ends
            at~a~node adjacent to~$v$}]\]
            (here and below, $[\cdot]$ is the Iverson bracket: for 
            a~predicate~$Y$, $[Y]=1$ if $Y$ is~true and $[Y]=0$ otherwise).
            
            \item[Complexity.] The mapping $\phi_n(G)$~can be~computed in~time $O^*(\binom{n}{{n/\theta}}2^{n/\theta})$ because \problem{Hamiltonian Path} on a graph with $n$ nodes can be~solved in~time $O^*(2^n)$.
            
            \item[Polynomial.] For every partition $V=S_1 \sqcup \dotsb \sqcup S_{\theta}$ into disjoint subsets of~size~${n/\theta}$ and
            every ${\theta}$~nodes $v_1, \dotsc, v_{\theta+1}$, add to~$P_{s(n)}$ a~monomial
            \[x_{S_1,v_1,v_2} \cdot x_{S_2,v_2,v_{3}} \dotsb x_{S_\theta,v_{\theta},v_{\theta+1}} \, .\]
            The number of~monomials added to~$P_{s(n)}$ 
            is~at~most~$\theta^nn^{\theta}=2^{O(n)}$.
            
            \item[Degree.] The degree of~$\mathcal P$ is~$\theta$.
        \end{description}

        \paragraph{\problem{3d-Matching}.} Given a~3-uniform 3-partite hypergraph~$G(V_1 \sqcup V_2 \sqcup V_3, E)$ with 
        	parts of~size~$n$ (that is, $|V_1|=|V_2|=|V_3|=n$ and $E \subseteq V_1 \times V_2 \times V_3$) and an~integer~$t$, check whether $G$~contains a~matching 
        	of~size at~least~$t$.
        	
        \begin{description}
            \item[Idea.] 
             One can break a~matching~$M$ of~size~$t$ into $\theta$~matchings $M_1, \dotsc, M_{\theta}$ 
         	of~size at~most $n/\theta$. Then, $M_1, \dotsc, M_{\theta}$ is a~valid partition iff $M_i$'s are
         	node-disjoint matchings.
         	\item[Variables.] Introduce $s(n)=\binom{n}{\leq n/\theta}^3=O^*\left(\binom{n}{n/\theta}^3\right)$ variables:      	\[X=\{x_{A,B,C} \colon A \subseteq V_1, B \subseteq V_2, C \subseteq V_3, |A|=|B|=|C| \le n/\theta\} \, .\]
         	 The function $\phi_n(G)$ assigns the following value to $x_{A,B,C}$:
         	 \[[\text{the induced subgraph $G[A \sqcup B \sqcup C]$ contains a~perfect matching}] \, .\]
         	 \item[Complexity.]
         	 This~can be~computed in~time $O^*\left(8^{n/\theta}\binom{n}{n/\theta}^3\right)$ since \problem{3d-Matching}
     	        is~solvable in~time $O^*(8^n)$ in 3-partite graphs with parts of~size~$n$.\footnote{This is~done by a~straightforward dynamic programming algorithm: for $A \subseteq V_1, B \subseteq V_2, C \subseteq V_3$, let $M(A, B, C)$ be the maximum size of a~matching in $G[A \cup B \cup C]$; then, $M(A, B, C) = \max_{\{a, b, c\} \in E} M(A \setminus a, B \setminus b, C \setminus c) + 1.$}
     	    \item[Polynomial.] The polynomial~$P_{s(n)}$ is~defined as~follows. For every 
         	$A_1, \dotsc, A_\theta \subseteq V_1$,
         	$B_1, \dotsc, B_\theta \subseteq V_2$,
         	$C_1, \dotsc, C_\theta \subseteq V_3$, such that all $A_i$'s, $B_i$'s, and $C_i$'s
         	are pairwise disjoint, have size at most $n/\theta$, and that
         	\[\left|\bigcup_{i \in [\theta]}A_i\right| = t \, ,\]
         	add to~$P_{s(n)}$ a~monomial
         	\[\prod_{j \in [\theta]}x_{A_j,B_j,C_j} \, .\] 
         	The number of~monomials added to~$P_{s(n)}$
     	    is at most $\binom{n}{n / \theta}^{3\theta} = 2^{O(n)}$.
     	    \item[Degree.] The degree of~$\mathcal P$ is~$\theta$.
        \end{description}

        \paragraph{\problem{Independent Set}.} Given a~graph~$G(V,E)$ with $n$~nodes and an~integer~$t$, check whether $G$~contains an~independent set of~size at~least~$t$.
        
        \begin{description}
            \item[Idea.] 
            An~independent set~$I$ of size~$t$ can 
        	be~partitioned into $\theta$ sets $I_1, \dotsc, I_{\theta}$ of size at~most~$n/\theta$. Then, $I_1, \dotsc, I_{\theta}$ is a~valid partition iff their total size is at least~$t$ and for all $i \neq j$, $I_i \cup I_j$ is an~independent set.
        	\item[Variables.] Introduce
            $s(n)=\binom{n}{\le 2n/\theta}=O^*(\binom{n}{2n/\theta})$ variables:
            \[X=\{x_S \colon S \subseteq V \text{ and } |S|\le 2n/\theta\} \, .\]
            The mapping~$\phi_n(G)$ assigns to~$x_S$ the value
            \[[\text{$S$ is an~independent set of~$G$}] \, .\]
            \item[Complexity.] The mapping $\phi_n(G)$ can be~computed in~time $O^*(2^{2n/\theta}\binom{n}{2n/\theta})$,
            since \problem{Independent Set} can be~solved in~time $O^*(2^n)$.
            \item[Polynomial.] For every $S_1, \dotsc, S_{\theta} \in \binom{V}{\le n/\theta}$ such that $S_i \cap S_j =\emptyset$, for all $i \neq j$, and $|\cup_{i \in [\theta]} S_i| = t$, add to~$P_{s(n)}$ a~monomial
            \[\prod_{1 \le i < j \le \theta}x_{S_i \cup S_j} \, .\]
            The number of~monomials added to~$P$ is at~most~$\theta^n=2^{O(n)}$.
            \item[Degree.] The degree of~$\mathcal P$ is~$O(\theta^2)$. 
        \end{description}
        
        \paragraph{\problem{Vertex Cover} and \problem{Clique}.} 
        These two problems are close relatives of~\problem{Independent Set}: 
        the complement of an~independent set in~a~graph
        is a~vertex cover of~this graph; a~clique in a~graph
        is an~independent set in the~complement of the graph.
        Thus, for \problem{Vertex Cover} and \problem{Clique} one can use the polynomial formulation of \problem{Independent Set}.

        \paragraph{\problem{MAX-$k$-SAT}.} Given a~$k$-CNF formula $F=C_1 \land \dotsb \land C_m$ over $n$~variables and an~integer~$t$, check whether it~is~possible to~satisfy at~least $t$~clauses of~$F$.
        
        \begin{description}
            \item[Idea.] 
            An~assignment $\mu \in \{0,1\}^n$ satisfying at~least~$t$ clauses can be~partitioned into $\theta$
            subassignments $\mu_1, \dotsc, \mu_{\theta} \in \{0,1\}^{n/\theta}$. Then, for each clause~$C$, one can assign at most $k$~subassignments that are ``responsible'' for~$C$: these are the subassignments that contain the $k$ variables from~$C$. Then, $\mu_1, \dotsc, \mu_{\theta}$ is a~valid 
        	partition if~the total number of clauses satisfied 
        	by~their $k$-tuples of~subassignments is at~least~$t$.
        	
        	\item[Variables.] Partition
            the set of~variables of~$F$ into $\theta$~blocks $V_1, \dotsc, V_{\theta}$
            of~size~$n/\theta$.
            For each clause~$C$, assign $k$~blocks such that all variables of~$C$
            belong to~these blocks: formally, let $b(C) \subseteq [{\theta}]$, $|b(C)|=k$, and the set of~variables of~$C_i$ is a~subset of~$\cup_{i \in b(C)}V_i$.
            
            Introduce $s(n)=\binom{n/{\theta}}{k}2^{nk/{\theta}}t$ variables:
            \[X=\{x_{B,\tau,r} \colon B \subseteq [{\theta}], |B|=k, \tau \in \{0,1\}^{nk/{\theta}}, 0 \le r \le t\} \, .\]
            For $B \subseteq [{\theta}]$, $|B|=k$, by~$c(B)$ define the set of~clauses~$C$ of~$F$ such that $b(C)=B$.
            The mapping $\phi_{n}(F)$ assigns the following value to $x_{B, \tau, {r}}$:
            \[[\text{$\tau$~satisfies at~least $r$~clauses from~$c(B)$}] \,.\]
            
            \item[Polynomial.] 
            Let $\mathcal F$ be a~class of~functions $f \colon 2^{[{\theta}]} \to \mathbb{Z}_{\ge 0}$ such that $\sum_{B \subseteq [{\theta}], |B|=k}f(B)=t$. For $\mu \in \{0,1\}^n$ and $B \subseteq [{\theta}]$, let $\mu_B$ be a~projection on~coordinates $\cup_{i \in B}V_i$. 
            For every $f \in \mathcal F$ and $\mu \in \{0,1\}^n$, 
            add to~$P_{s(n)}$ a~monomial
            \[\prod_{B \subseteq [{\theta}], |B|=k}x_{B,\mu_B,f(B)} \, .\]
            Clearly, $|\mathcal F| \le t^{2^\theta} \le n^{O(k2^\theta)}$ and $|\{B \colon B \subseteq [\theta], |B|=k\}| \le 2^{\theta}$. Hence, the number of~monomials added to~$P$ is at~most $O^*(2^n)=2^{O(n)}$.
            \item[Degree.] The degree of~$P$ is $\binom{\theta}{k} \le 2^{\theta}$.
        \end{description}

        \paragraph{\problem{$k$-SAT}.} \problem{$k$-SAT} is a~special case of~\problem{MAX-$k$-SAT}.

        \paragraph{\problem{Graph Coloring}.} Given a~graph~$G(V,E)$ with $n$~nodes and an~integer~$t$, check whether $G$~can be~colored properly using at~most $t$~colors.
        
        \begin{description}
            \item[Idea.] 
Partition~$V$ into $\theta$~blocks $V_1, \dotsc, V_{\theta}$ of size $n/\theta$. We~would like 
    		to~construct a~$t$-coloring of~$V$ from~colorings of
    		the blocks. However, a~$t$-coloring may contain a~color
    		whose color class is large (much larger than $n/\theta$). For this reason, the polynomial formulation below is a~bit trickier than the previous ones: we~have to~treat color classes differently depending on~their size.
    		\item[Variables.] Introduce $s(n)=O(n^2\binom{n}{2n/\theta})$ variables:
            \[X=\{x_{S,r} \colon S \subseteq V, |S| \le 2n/\theta, 0 \le r \le t\} \, .\]

            The mapping~$\phi_{n}(G)$ assigns to~$x_{S,r}$ the value
            \[[\chi(G[S]) \le r] \, .\]
            As~the chromatic number of a~$n$-node graph can be~found
            in~time~$O^*(2^n)$~\cite{DBLP:journals/siamcomp/BjorklundHK09}, the mapping $\phi_n(G)$ can be computed in time \[O^*\left(2^{2n/\theta}\binom{n}{2n/\theta}\right) \, .\]
            
            \item[Polynomial.] The polynomial $P_{s(n)}$ is~defined as~follows:
            \[P_{s(n)}(X)=\sum_{\substack{V=S_1 \sqcup \dotsb \sqcup S_{\theta}\\|S_i|\le 2{n/\theta} \text{ for } i \in [{\theta}]}} \sum_{\substack{t_1+\dotsb+t_p=t\\t_i \ge 0 \text{ for } i \in [{\theta}]}} \prod_{i \in [p]}x_{S_i,t_i} \, .\]
            We~claim that $P_{s(n)}(\phi_n(X)) > 0$ iff $G$~can be~properly colored using $t$~colors such that every color induces an~independent set of~size
            at~most~${n/\theta}$. Indeed, if~there is~such a~coloring, one can greedily pack color classes into groups of size at~most~$2{n/\theta}$ and obtain the required partition $V=S_1 \sqcup \dotsb \sqcup S_{\theta}$.
            
                    Thus, it~remains to~consider the case when there exists a~$t$-coloring where at~least one of the color classes has size more than~${n/\theta}$.
        Since each color class induces an~independent set in the~graph,
        we~are going to~reuse the ideas from polynomial formulation
        of~\problem{Independent Set}. Namely, assume that $T_1, \dotsc, T_l$ are all large color classes in a~$t$-coloring: for every $i \in [l]$, $T_i \subseteq V$ is an~independent set of~size more than~${n/\theta}$.

        Introduce the following additional variables:
        \[Y=\{y_S \colon S \subseteq V, |S| \le {2n/\theta}\} \, .\]
        The mapping~$\phi_n(G)$ assigns the following values to~$y_S$:
        \[[\text{$S$ is an~independent set in~$G$}] \, .\]
        
                For every $L \subseteq V$ of~size $l>{n/\theta}$, fix its partition (say, the lexicographically first one) $L_1 \sqcup \dotsb \sqcup L_{\lceil l\theta/n \rceil}$ into subsets of~size~${n/\theta}$: all sets are disjoint and all of~them have size~${n/\theta}$ except for possibly the last one. The following monomial expresses the fact that $L$~is an~independent set in~$G$:
        \[M(L)=\prod_{1 \le i < j \le \lceil l\theta/n \rceil}y_{L_i \cup L_j} \, .\]

        The final polynomial looks as~follows:
        \[Q_{s(n)}(X,Y)=\sum_{l=0}^{{\theta}}\sum_{\substack{T_1, \dotsc, T_l \subseteq V\\T_i \cap T_j = \emptyset \text{ for all }i\neq j\\|T_i|>{n/\theta} \text{ for all }i}} \sum_{\substack{S_1 \sqcup \dotsb \sqcup S_{\theta}=V\setminus \cup_{i \in [l]}T_i\\|S_i| \le 2{n/\theta} \text{ for all }i \in [p]\\t_1+\dotsb+t_p=t-l}}\prod_{i \in [l]}M(T_i)\prod_{i \in [{\theta}]}x_{S_i,t_i} \, .\]
        
        The number of~monomials added to~$Q$ is at~most $O^*(n^{\theta}\theta^n)=2^{O(n)}$.
        \item[Degree.] The degree of~$Q_{s(n)}$ is at~most $2\theta^2+\theta$. 
        \end{description}

        \paragraph{\problem{Set Cover}.} Given a~set family $\mathcal F = \{F_1, \dotsc, F_m\} \subseteq 2^{[n]}$, $m = n^{O(1)}$ and an~integer~$t$, check whether one can cover $[n]$ with at~most $t$~sets from~$\mathcal F$.
        
        \begin{description}
            \item[Idea.] 
Partition the universe $[n]$ into $\theta$ blocks of~size $n/\theta$. Each of~these blocks is~either covered~by at~most~$t$ sets or is~covered by a~single large set (of size at least $n/\theta$) that also possibly intersects other blocks.
            
            \item[Variables.]
            Introduce $s(n)=O(n^{O(1)}\binom{n}{2n/\theta})$ variables:
            \begin{align*}
                X&=\{x_{S,r} \colon S \subseteq [n], |S| \le 2n/\theta, 0 \le r \le t\} \, ,\\
                Y&=\{y_{S,i} \colon S \subseteq [n], |S| \le 2n/\theta, 1 \le i \le m\} \, .
            \end{align*}
    
            The mapping~$\phi_{s(n)}(\mathcal F)$ assigns the following values to~$x_{S,r}$ and $y_{S,i}$:
            \begin{align*}
                &[\text{$S$ can be covered by $r$ sets from~$\mathcal F $}] \, ,\\
                &[S \subseteq F_i] \, ,\\
            \end{align*}
            
            \item[Complexity.] As~\problem{Set Cover} problem can be~solved
            in~time~$O^*(2^n)$~\cite{DBLP:journals/siamcomp/BjorklundHK09}, the mapping~$\phi_{s(n)}(\mathcal F)$ can be~computed in~time \[O^*\left(2^{\frac{2n}{\theta}}\binom{n}{\frac{2n}{\theta}}\right) \, .\]
            
            \item[Polynomial.]
            For every $L \subseteq V$ of~size $l>n/\theta$, fix its partition (say, the lexicographically first one) $L_1 \sqcup \dotsc \sqcup L_{\lceil \theta l/n \rceil}$ into subsets of~size~$n/\theta$: all sets are disjoint and all of~them have size~$n/\theta$ except for possibly the last one. The following monomial expresses the fact that $L \subseteq F_q$:
            \[M(L,q)=\prod_{i=1}^{\lceil \theta l/n \rceil}y_{L_i,q} \, .\]
    
            Finally, the polynomial $Q_{s(n)}(X,Y)$ is~defined as~follows:
            \[\sum_{l=0}^{\theta}\sum_{\substack{q_1, \dotsc, q_l \in [m]\\ \forall i \neq j: q_i \neq q_j}}\sum_{\substack{T_1, \dotsc, T_l \subseteq V\\T_i \cap T_j = \emptyset \text{ for all }i\neq j\\|T_i|>n/\theta \text{ for all }i}} \sum_{\substack{S_1 \sqcup \dotsb \sqcup S_{\theta}=V\setminus \cup_{i \in [l]}T_i\\t_1+\dotsb+t_p=t-l}}\prod_{i \in [l]}M(T_i,q_i)\prod_{i \in [\theta]}x_{S_i,t_i} \, .\]
            
            \item[Degree.]
            The degree of~this polynomial is at~most $(\theta+1)\theta$. The number of~monomials added to~$Q_{s(n)}$
            is at~most $O^*(n^{O(\theta)}\theta^n)=2^{O(n)}$.
        \end{description}
\end{proof}

\subsection{Parameterized Problems}\label{sec:paramformulations}

\subsubsection{Technical Lemmas}
    For parameterized polynomial formulations, we utilize the following technical lemma. We~provide its proof in~\cref{apx:treelemma}.
    \begin{restatable}{lemma}{treedecomposition}
    \label{st:tree_decomposition}
        Let $T(V,E)$ be a~tree, $M \subseteq V$ be a~set of $k$~nodes, and $\theta > 1$ be an integer. Then $E$~can be~partitioned into
        $m \le \theta$ blocks $E=E_1 \sqcup \dotsb \sqcup E_m$
        such that, for each $i \in [m]$, $E_i$ induces a~subtree $T_i$
        of~$T$ with at~most $\frac{2k}{\theta - 1} + 2$ nodes from~$M$.
    \end{restatable}
    
    The following definition resembles a~block structure of a~graph: for a~family of~sets $X_1, \dotsc, X_m$ we~introduce $m$~nodes
    and connect the $i$-th of these nodes with all elements of~$X_i$
    that belong to~at least one other~$X_j$.
    
    \begin{definition}
        Given $m$~sets $X_1, \dotsc, X_m$, we introduce a~set $S = \{s_1, \dotsc, s_m\}$ such that $S$~does not intersect any $X_i$, and a set $C = \bigcup\limits_{i \neq j} (X_i \cap X_j)$.
By a~\emph{subset graph} $B(X_1, \dotsc, X_m)$ we denote the following graph $G(V,E)$: $V=S \sqcup C$, $E = \{ \{s_i, c\} \colon c \in X_i \cap C \}$.
        Each $v \in S$ is~called a~\emph{set node}, 
        and each $v \in C$ is called a~\emph{connector node}.
    \end{definition}

    \begin{remark}
    \label{re:b_tree}
        Let $T_1, \dotsc, T_m$ be the trees resulting from 
        applying \cref{st:tree_decomposition} to a~tree~$T$.
        Then, $B = B(V(T_1), \dotsc, V(T_m))$ is a~tree
        containing $m$~set nodes and at~most~$m - 1$ connector nodes.
    \end{remark}

\subsubsection{Polynomial Formulations}
\parampolynomialformulation*
\begin{proof}
We will design polynomial formulations $P_{s(x,k)}$ for the above problems where $s(x,k)=s(n,k)$ will be a function of $n=|x|$ and $k$. Naturally, we would like to have different polynomials $P$ for different values of $(n,k)$, alas, $s$ is not necessarily injective. One way to overcome this issue is to consider a two-dimensional sequence of polynomials $P_{s(n,k),k}$ as we'll always have that $s(\cdot,k)$ is injective for every~$k$. But this approach would cause technical issues in the proof of arithmetic lower bounds in \cref{thm:mainparam}. Instead, we still consider a sequence $(P_1,P_2,\ldots)$ of polynomials, but slightly modify the function $s$. Given a function $s(n, k)$, we define the following Cantor pairing function~\cite{HU79} of $s(n,k)$ and $k$:
\[
s'(n,k)=(s(n,k)+k)(s(n,k)+k+1)/2+k \,.
\]
We note that $s'(\cdot,\cdot)$ is injective because $s(\cdot,k)$ is injective for every~$k$. In all polynomial formulations below, we can always switch from $s(x,k)$ to $s'(x,k)$ to resolve the aforementioned issue. Indeed, since $s'(x,k)\geq s(x,k)$, we can just add $s'(x,k)-s(x,k)$ dummy variables to the polynomial. Also, since $s'(x,k)\leq(s(x,k)+k)^2$, it suffices to replace the bound $s(x,k)\leq T(k) |x|^\Delta$ by the bound $s'(x,k)\leq (s(x,k)+k)^2 \leq T'(k) |x|^{\Delta'}$ for $T'(x)=T(k)^2\cdot k^2$ and $\Delta'=\Delta$. Since we're proving polynomial formulations for \emph{all} $T=c^k, c>1$, the change in $T'$ doesn't affect the statement of the lemma.

In~all polynomial formulations below we~follow the same five-step
pattern as~in \cref{lemma:polynomialformulation} with the following three differences.
\begin{description}
    \item[Kernel.] For most of the~considered problems, we~start by~applying a~kernel. Recall that a~kernel replaces, in~polynomial time,
    an~instance $(x,k)$ of a~parameterized problem~$A$
    with an~equivalent instance $(x',k')$ of~$A$ such that 
    $|x'|, k' \le g(k)$, for some computable function~$g$.
    To~simplify the presentation of~polynomial formulations,
    we~identify $(x,k)$ with $(x',k')$. This allows~us
    to~assume from the beginning that $n=|x| \le g(k)$.
    \item[Variables.] The value of $s(x,k)$ will~be
    a~constant degree of~a~product~of $2^{O(k/\theta)}$
    and $\binom{O(k)}{O(k/\theta)}$ as~well~as $n^{O(1)}$.
    By~choosing a~large enough~$\theta=\theta(c)$, we~ensure that
    $s(x,k) \le c^kn^{O(1)}$.
    \item[Polynomial.] We~ensure that the polynomial $P_{s(x,k)}$
    is a~sum of~at~most $n^{O(1)}2^{O(k)}$ monomials.
\end{description}

    \paragraph{\problem{$k$-Vertex Cover}.} Given a~graph~$G(V,E)$, check whether $G$~contains a~vertex cover of~size at~most~$k$.
    
    \begin{description}
        \item[Kernel.] As~there exists a~kernel of~size~$2k$
        for \problem{$k$-Vertex Cover}~\cite{chen2001vertex}, we~assume that $|V|=n \le 2k$. 
        \item[Idea.] 
Partition~$V$ into $\theta$ blocks of~size $n/\theta$: $V=V_1 \sqcup \dotsb \sqcup V_{\theta}$. This induces a~partition
        of any~$S \subseteq V$ into~$\theta$ blocks of~size at~most~$n/\theta$: $S=S_1 \sqcup \dotsb \sqcup S_\theta$ where $S_i = S \cap V_i$. Then, $S$~is a~vertex cover of~$G$
        iff, for all $1 \le i<j \le \theta$, $S_i \sqcup S_j$ is a~vertex cover of~$G[V_i \sqcup V_j]$.
        \item[Variables.] Introduce $s(G,k)=O(n^2\binom{n}{\le n/\theta}^2)=O^*(\binom{2k}{2k/\theta})$ variables:
        \[X=\{x_{i,j,A,B} \colon i,j \in [\theta], A \subseteq V_i, B \subseteq V_j\} \, .\]
        The mapping $\phi(G,k)$ assigns the following value to a~variable $x_{i,j,A,B}$:
        \[[\text{$A \sqcup B$ is a~vertex cover of~$G[V_i \sqcup V_j]$}] \, .\]
        \item[Complexity.] The mapping $\phi(G,k)$ can be~computed
        in time $O^*(\binom{2k}{2k/\theta})$.
        \item[Polynomial.] For every $S \subseteq V$ of~size at~most~$k$, add the following monomial to~$P_{s(G,k)}$:
        \[\prod_{1 \le i < j \le \theta}x_{i,j,S_i,S_j} \, .\]
        The number of~monomials added to~$P_{s(G,k)}$ is $\binom{n}{\le k}=O^*(\binom{2k}{k})$.
        \item[Degree.] The degree of~$\mathcal P$ is at most~$\theta^2$.
    \end{description}
    
    \paragraph{\problem{$k$-Steiner Tree}.} Given a~graph~$G(V,E)$ with (integer non-negative) edge weights and a~subset $S \subseteq V$ of~its nodes of~size~$k$ (called terminals), and an~integer $0 \le t \le |V|^{O(1)}$, check whether there is a~tree in~$G$ of~weight at~most~$t$ containing all nodes from~$S$. 
    \begin{description}
\item[Idea.] 
Assume that $S = \{1, \dots, k\}$: relabel nodes if needed. 
Consider a~Steiner tree~$T$ that we are looking for. Using \cref{st:tree_decomposition} for~$T$ and~$M = S$ one can find subtrees $T_1,\dotsc, T_m$ of~$T$ for some $m \leq \theta$ such that each $T_i$ contains at most $\frac{2k}{\theta-1}+2 \le 3k/\theta$ terminals. 
        Let $V_1, \dots, V_m$ be the corresponding sets of~nodes, that is, $V_i = V(T_i)$, let $S_i = S \cap V_i$ for every $1 \le i \le m$, let $\ell_i$ be the weight of $T_i$, and let $B = B(V_1, \dots, V_m)$. By \cref{re:b_tree}, $B$ is a~tree and it contains $m$ subset nodes and at most $m - 1$ connector nodes. Observe that $B$~can also be obtained as $B(A_1, \dots, A_m)$, where $A_i = V_i \cap \bigcup\limits_{j \neq i} V_j$. It~is significant since each $|A_i|$ is bounded by $\theta - 1$ in contrast to $|V_i|$ which is bounded by~$n$. Note also that $\sum \ell_i \le L$ and for every~$i$, a~subtree $T_i$ has weight $\ell_i$ and connects nodes from $S_i \cup A_i$.
        
        To construct the polynomial, we go over all possible $\ell_i$, $\{S_i\}$ and $\{A_i\}$ such that $\sum \ell_i \le L$, $\bigcup S_i = S$ and $B(A_1, \dots, A_m)$ is connected, and check whether for every $i$ we can connect nodes from $S_i \cup A_i$ with a tree of weight at most $\ell_i$. Observe that those trees can intersect and their union can give us a~connected subgraph that is not a~tree, but we still can obtain a~proper Steiner tree just by~taking a~spanning tree of that subgraph.
        
        \item[Variables.] Introduce $s(G,k)=\mathcal{O}(\binom{k}{\le 3k/\theta + \theta} \cdot n^{\theta} \cdot n^{O(1)}) +  n^{O(1)} = O^*(\binom{k}{4k/\theta})$ variables:
        \begin{align*}
          X = \{& x_{S', A, \ell} \colon  S' \subseteq S, |S'| \le 3k/\theta, A \subseteq V, |A| \le \theta - 1, \ell \le n^{O(1)}\}; \\
          Y = \{& y_{L'} \colon L' \le n^{O(1)} \} \,.
        \end{align*}
        
        The mapping $\phi(G,k)$ assigns the following values to variables:
        \begin{align*}
            x_{S', A, \ell} \mapsto [& \text{there exists a Steiner tree in $G$ for the set of terminals } S' \cup A\\ 
            & \text{of the weight at most $\ell$];} \\ 
            y_{L'} \mapsto [& L' \le L] \, .
        \end{align*}
        
        \item[Complexity.] The mapping $\phi(G,k)$ can be~computed in time $O^*(\binom{k}{4k/\theta} \cdot 2^{4k/\theta})$, since \problem{$k$-Steiner tree} can be solved in~time $O^*((2 + \delta)^k)$ for any $\delta > 0$ \cite{fuchs2007dynamic}.
        
        \item[Polynomial.] 
        The polynomial looks as follows:
        \[P_{s(G, k)}(X, Y) = \sum\limits_{\substack{
            m \le \theta, \\ 
            \ell_1, \dots, \ell_m, \\
            \sum \ell_i \le n^{O(1)}
            }} \sum\limits_{\substack{
            S_1 \cup \dots \cup S_m = S, \\
            |S_i| \le 3k/\theta
            }} \sum\limits_{\substack{
            A_1, \dots, A_m \subseteq V, \\
            |\bigcup A_i| \le \theta - 1, \\
            B(A_1, \dots, A_m) \\
            \text{is connected}
            }} y_{\sum \ell_i}
            \prod\limits_{i = 1}^{m} x_{S_i, A_i, \ell_i} 
        \, .\]
        The number of~monomials in~$P_{s(G,k)}$ is at most $\theta  \cdot n^{O(1)} \cdot \binom{k}{\le 3k/\theta}^{\theta} \cdot \binom{n}{\theta}^{\theta} = 2^{O(k)} n^{O(1)}$.
        \item[Degree.] The degree of~$\mathcal P$ is $\theta + 1$.
    \end{description}
    
    \paragraph{\problem{$k$-Internal Spanning Tree}.} Given a~graph~$G$,
    check whether there is a~spanning tree of~$G$ with at~least $k$~internal nodes.
    \begin{description}
        \item[Kernel.] As~there exists a~kernel for \problem{$k$-Internal Spanning Tree} of~size~$3k$ \cite{fomin2013linear} we~assume that $|V|=n \le 3k$.  
        
        \item[Idea.] 
Let $T$~be a~tree we are looking for. Using \cref{st:tree_decomposition} for~$T$ and $M = V(T)$, we obtain $m \le \theta$ subtrees~$T_i$ of~size at~most 
        $\frac{2n}{\theta - 1} + 2 \le 3n / \theta$. 
        For each $i \in [m]$, let $S_i=V(T_i)$
        and $k_i$ be~the number of~internal nodes of~$T$ that 
        belong to~$S_i$.   
        By~\cref{re:b_tree}, $B = B(S_1, \dots, S_m)$ is a~tree. 
        Let $S$~be the set of~set nodes and $C$~be the set of~connector nodes of~$B$. 
        When summing up all~$k_i$, we~count each internal node~$v$ of~$T$ that belongs to~$C$ exactly $\deg_B(v)$ times.
        Then, we~count $\sum\limits_{v \in C} \deg_B(v) - |C| = (|S| + |C| - 1) - |C| = m - 1$ extra nodes, so $\sum k_i \ge k + (m - 1)$.
        
        To construct the polynomial, we go~over all possible $S_1, \dots, S_m \subseteq V$ such that $\bigcup S_i = V$ and $B(S_1, \dots, S_m)$ is a~tree and over all possible $k_1, \dots, k_m$ such that $\sum k_i \ge k + (m - 1)$.
        For fixed $\{S_i\}$ and $\{k_i\}$, it suffices to check that for every~$i$ there exists a~spanning tree of $G[S_i]$ with at least $k_i$ internal nodes, where we also should consider $S_i \cap \bigcup\limits_{j \neq i} S_j$ as internal nodes (to do that, we add a~leaf to each node from $S_i \cap \bigcup\limits_{j \neq i} S_j$).

        \item[Variables.] Introduce $s(G,k)=\binom{n}{\le 3n/\theta} \cdot \binom{n}{\le \theta} \cdot 3n/\theta = O^*(\binom{3k}{9k/\theta})$ variables:
        \[X=\{x_{S, A, k'} \colon A \subseteq S \subseteq V, |S| \le 3n/\theta, |A| \le \theta - 1, k' \le |S|\} \, .\]
        
        Let $G_{S, A}$ be a~graph obtained from $G[S]$ by adding 
        a~leaf to~each node $v \in A$.
        
        The mapping $\phi(G,k)$ assigns the following value 
        to~a~variable $x_{S, A, k'}$:
        \[[\text{there exists a~spanning tree of $G_{S, A}$ that contains at least $k'$ internal nodes}] \, .\]
        
        \item[Complexity.] The mapping $\phi(G,k)$ can be~computed
        in time $O^*(\binom{3k}{9k/\theta} \cdot 8^{3n/\theta}) = O^*(\binom{3k}{9k/\theta} \cdot 8^{9k/\theta})$, since \problem{$k$-Internal Spanning Tree} can be solved in time $O^*(8^k)$ \cite{fomin2013linear}.
        
        \item[Polynomial.] The polynomial looks as follows:
        \[P_{s(G, k)}(X) = \sum\limits_{m \in [\theta]} 
            \sum\limits_{\substack{
            S_1, \dots, S_m \subseteq V, \\
            1 < |S_i| \le 3n/\theta, \\ 
            \bigcup S_i = V, \\
            B(S_1, \dots, S_m) \text{ is a tree}
            }} \sum\limits_{\substack{
            k_1, \dots, k_m, \\
            k_i \le |S_i|, \\
            \sum k_i \ge k + (m - 1) \\
            }} \prod\limits_{i = 1}^m x_{S_i, \bigcup\limits_{j \neq i} (S_i \cap S_j), k_i} 
        \]

        The number of~monomials in~$P_{s(G,k)}$ is at most $\theta \cdot \binom{n}{3n/\theta}^{\theta} \cdot (3n/\theta)^{\theta} = \binom{3k}{9k/\theta}^{\theta} k^{O(1)} = 2^{O(k)}$.
        
        \item[Degree.] The degree of~$\mathcal P$ is~$\theta$.
    \end{description}

    \paragraph{\problem{$k$-Leaf Spanning Tree}.} Given a~graph~$G$,
    check whether there is a~spanning tree of~$G$ with at~least $k$~leaves.
    \begin{description}
        \item[Kernel.] As~there exists a~kernel for \problem{$k$-Leaf Spanning Tree} of~size~$5.75k$ \cite{fellowscoordinatized}, we~assume that $|V|=n \le 5.75k \le 6k$.  
        
        \item[Idea.] 
Similarly to~the polynomial formulation of~\problem{$k$-Internal Spanning Tree},
        we~go over all possible $S_1, \dots, S_m \subseteq V$ such that $\bigcup S_i = V$ and $B(S_1, \dots, S_m)$ is a~tree, and go~over all possible $k_1, \dots, k_m$ such that $\sum k_i \ge k$. For fixed $\{S_i\}$ and $\{k_i\}$, 
        it~suffices to~check that for every~$i$ there exists 
        a~spanning tree of $G[S_i]$ with at least $k_i$ leaves, where we should consider $S_i \cap \bigcup\limits_{j \neq i} S_j$ as internal nodes (to do that, we add a~leaf to each node from $S_i \cap \bigcup\limits_{j \neq i} S_j$).
        
        \item[Variables.] Introduce $s(G, k) = O(\binom{n}{\le 3n / \theta} \cdot n^{\theta} \cdot 3n/\theta) = O^*(\binom{6k}{18k/\theta})$ variables:
        \[X=\{x_{S, A, k'} \colon A \subseteq S \subseteq V, |S| \le 3n/\theta, |A| \le \theta - 1, k' \le |S|\} \, .\]

        The mapping $\phi(G,k)$ assigns the following value 
        to~a~variable $x_{S, A, k'}$:
        \[[\text{there exists a spanning tree of $G_{S, A}$ that contains at least $k' + |A|$ leaves}] \, .\]
        We add~$|A|$ as we~obtain a~graph with $|A|$~dummy leaves
        that appear in~any spanning tree.
        
        \item[Complexity.] The mapping $\phi(G,k)$ can be~computed
        in time $O^*(\binom{6k}{18k/\theta} \cdot 4^{3n / \theta + \theta}) = O^*(\binom{6k}{18 / \theta} \cdot 4^{20k/\theta})$, since \problem{$k$-Leaf Spanning Tree} can be solved in time $O^*(4^{k})$ \cite{kneis2008new}.
        
        \item[Polynomial.] The polynomial looks as follows:
        \[P_{s(G, k)}(X) = \sum\limits_{m \le \theta}
            \sum\limits_{\substack{
            S_1, \dots, S_m \subseteq V, \\
            1 < |S_i| \le 3n/\theta, \\ 
            \bigcup S_i = V, \\
            B(S_1, \dots, S_m) \text{ is a tree}
            }} \sum\limits_{\substack{
            k_1, \dots, k_m, \\
            k_i \le |S_i|, \\
            \sum k_i \ge k, \\
            }} \prod\limits_{i = 1}^m x_{S_i, \bigcup\limits_{j \neq i} (S_i \cap S_j), k_i} 
        \]
        
        The number of~monomials added to~$P_{s(G,k)}$ is $O(\theta \cdot \binom{n}{\le 3n/\theta}^{\theta} \cdot (3n/\theta)^{\theta}) = \binom{6k}{18k/\theta}^{\theta}k^{O(1)} = 2^{O(k)}$.
        
        \item[Degree.] The degree of~$\mathcal P$ is $\theta$.
    \end{description}

    \paragraph{\problem{$k$-Nonblocker}.} Given a~graph~$G(V, E)$, check whether $G$~contains a~subset of nodes of~size at~least~$k$ whose complement is a~dominating set in~$G$.
    \begin{description}
        \item[Kernel.] As~there exists a~kernel of~size~$\frac{5k}{3}$
        for \problem{$k$-Nonblocker}~\cite{dehne2006nonblocker}, we~assume that $|V|=n \le \frac{5k}{3}$.  
        \item[Idea.] 
Let $N \subseteq V$ be a~nonblocker for $G$ with at~least 
        $k$~nodes, and let $D = V \setminus N$ be a~dominating set for $G$. Let $D^1 \sqcup \dots \sqcup D^{\theta}$ be 
        a~partition of~$D$ into $\theta$ 
        blocks of~size at~most~$n/\theta$.
        Let $N^1 \sqcup \dots \sqcup N^{\theta}$ be a~partition of $N$ such that each node from $N^i$ is adjacent to some node from $D^i$, and, for each $i \in [\theta]$, let $N^i_1 \sqcup \dots \sqcup N^i_{\theta}$ be a~partition of $N^i$ into $\theta$ blocks of~size at~most $n/\theta$. 
        This way, we get
        a~partition $\bigsqcup N^i_j \sqcup \bigsqcup D^i$ of~$V$ into $\theta^2 + \theta$ blocks of~size at~most $n/\theta \le \frac{5k}{3\theta}$.
        
        To construct the polynomial~$P_{s(G,k)}$, we go~through all such partitions, and for each partition we~check that for every~$i$ each node from $\bigsqcup\limits_j N^i_j$ is adjacent to~a~node from $D^i$.
        
        \item[Variables.] Introduce $s(G,k)=O(\binom{n}{\le n / \theta}^2) = O^*(\binom{5k/3}{5k/3\theta}^2)$ variables:
        \[X=\{x_{N, D} \colon N, D \subseteq V, |N|, |D| \le n/\theta\}\, .\]
        The mapping $\phi(G,k)$ assigns the following value 
        to~a~variable $x_{N, D}$:
        \[[ \forall v \in N \; \exists u \in D: \{v, u\} \in E]=[\text{every node in~$N$ is~dominated by a~node in~$D$}] \, .\]
        \item[Complexity.] The mapping $\phi(G,k)$ can be~computed in time $O^*(\binom{5k/3}{5k/3\theta}^2)$.
        \item[Polynomial.] For every $N \subseteq V$ of~size 
        at~least~$k$ and a~partition $\bigsqcup N^i_j \sqcup \bigsqcup D^i$ of~$V$ into $\theta^2 + \theta$ blocks 
        of~size at~most $n/\theta$ such that $\bigsqcup\limits_{i, j} N^i_j = N$, we add the following monomial to~$P_{s(G,k)}$:
        \[\prod\limits_{i, j \le \theta} x_{N^i_j, D^i} \, .\]
        The number of~monomials added to~$P_{s(G,k)}$ is $O((\theta^2 + \theta)^n) = 2^{O(k)}$.
        \item[Degree.] The degree of~$\mathcal P$ is $\theta^2$.
    \end{description}
    
    \paragraph{\problem{$k$-Path Contractibility}.} Given a~graph~$G$, check whether it~is
    possible to~contract at~most $k$~edges in~$G$ to~turn it~into a~path.
    \begin{description}
        \item[Kernel.] As~there exists a~kernel of~size~$5k + 3$
        for \problem{$k$-Path Contractibility}~\cite{heggernes2014contracting}, we~assume that $|V|=n \le 5k + 3 \le 6k$.
        
        \item[Idea.] Consider the path resulting from edge contraction.
        Each node of the path corresponds to a connected set of $G$ that contracts to that node.
        Let $\mathcal{S}$ be a family of those sets of size at least $n/\theta$, and let $\mathcal{T}$ be a family of those sets of size less than $n/\theta$. We apply the following procedure to $\mathcal{T}$.
        While $\mathcal{T}$ contains two sets $T_1$ and $T_2$ of size less than $n/(2\theta)$ and $T_1 \sqcup T_2$ is a connected set in $G$, we replace $T_1$ and $T_2$ in $\mathcal{T}$ by $T_1 \sqcup T_2$.

        Let $s := |\mathcal{S}|$ and $t := |\mathcal{T}|$. 
        Let $\mathcal{V} = \mathcal{S} \sqcup \mathcal{T}$, and let $\mathcal{V} = \{V_1, \dots, V_m\}$ where $V_i$ are numbered according to their order in the path. 
        This way, we get a partition $V_1 \sqcup \dots \sqcup V_m$ of $V$ into $m = s + t$ blocks.
        In order to obtain an upper bound for $m$, we consider a partition $\mathcal{T} = \mathcal{T}_{\ell} \sqcup \mathcal{T}_s$, where $\mathcal{T}_{\ell}$  contains sets of size at least $n / (2\theta)$ and $\mathcal{T}_s$ contains sets of size less than $n / (2\theta)$. Observe that $s \le n / (n / \theta) = \theta$, $|\mathcal{T}_{\ell}| \le n / (n / (2\theta)) = 2\theta$, and that after applying the above procedure to $\mathcal{T}$, for every $V_i \in \mathcal{T}_{s}$, $V_{i - 1}$ and $V_{i + 1}$ (if they exist) belong to $\mathcal{S} \sqcup \mathcal{T}_{\ell}$, so $|\mathcal{T}_{s}| \le |\mathcal{S}| + |\mathcal{T}_{\ell}| + 1 \le \theta + 2\theta + 1 \le 4\theta$. Hence, $t = |\mathcal{T}_{\ell}| + |\mathcal{T}_{s}| \le 6\theta$, and $m = s + t \le 7\theta$.
        
        Let $\mathcal{S} = \{S_1, \dots, S_s\}$ and $\mathcal{T} = \{T_1, \dots, T_t\}$. Sets from $\mathcal{S}$ can be too large, so we cover each of them with subsets of size at most $3n /\theta$. For every $i \in [s]$ we apply \cref{st:tree_decomposition} to a spanning tree of $G[S_i]$ and obtain its subtrees on node sets $S^1_i, \dots, S^{\ell}_i$, where $\ell \le \theta$, $\forall j \; |S^j_i| \le 3n/\theta$, $\bigcup\limits_{j}S^j_i = S_i$, and $B(S^1_i, \dots, S^{\ell}_i)$ is connected.
        
        Let $k_1, \dots, k_m$ be the numbers of contracted edges in $V_1, \dots, V_m$, respectively. 
        Let us partition $\{k_i\}$ into $\{k^s_i\}$ and $\{k^t_i\}$, where $k^s_i$ and $k^t_j$ are numbers of contracted edges in $S_i$ and $T_j$, respectively. Observe that $\forall i \in [s], \; k^s_i = |S_i| - 1$. That means that $\sum k_i \le k \Leftrightarrow \sum k^t_i + \sum (|S_i| - 1) \le k$. 
        
        Let $A_1, \dots, A_m$ and $B_1, \dots, B_m$ be such sets of nodes that $\forall i \in [m - 1] \; B_i$ consists of nodes from $V_i$ that are adjacent to a node from $V_{i + 1}$, and $A_{i + 1}$ consists of nodes from $V_{i + 1}$ that are adjacent to a node from $V_{i}$. Let us partition $\{A_i\}$ into $\{A^s_i\}$ and $\{A^t_i\}$ such that $\forall i \in [s] \; A^s_i \subseteq S_i$ and $\forall i \in [t] \; A^t_i \subseteq T_i$. Similarly, we partition $\{B_i\}$ into $\{B^s_i\}$ and $\{B^t_i\}$.
        
        Consider $a_1, \dots, a_m$ and $b_1, \dots, b_m$ such that $\forall i \in [m] \; a_i \in A_i$, $b_i \in B_i$ and $\forall i \in [m - 1]$ there is an edge between $b_i$ and $a_{i + 1}$.
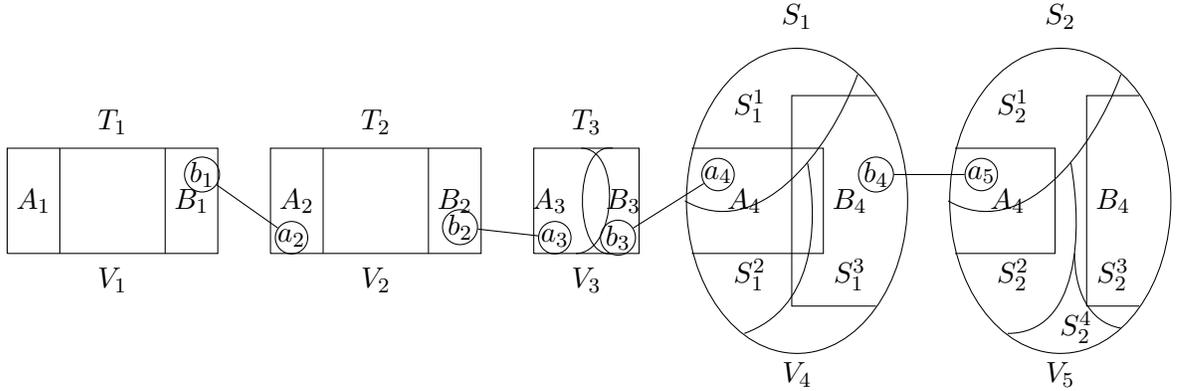
\begin{center}
\begin{tikzpicture}[scale=0.7]
 \coordinate (a1) at (0,0) {};
 \coordinate (a2) at (0,2) {};
 \coordinate (a3) at (4,2) {};
 \coordinate (a4) at (4,0) {};
 \coordinate (a5) at (1,0) {};
 \coordinate (a6) at (1,2) {};
 \coordinate (a7) at (3,0) {};
 \coordinate (a8) at (3,2) {};
 \node (a9) at (0.5,1) {$A_1$};
 \node (a10) at (3.5,1) {$B_1$};
 \node[circle,draw,minimum size = 0.5mm,inner sep=0mm] (a12) at (3.7,1.5) {$b_1$};
 \node (a13) at (2,2.5) {$T_1$};
 \node (a14) at (2,-0.5) {$V_1$};
 
 \draw (a1) -- (a2) -- (a3) -- (a4) -- (a1);
 \draw (a5) -- (a6);
 \draw (a7) -- (a8);

 \coordinate (b1) at (0+5,0) {};
 \coordinate (b2) at (0+5,2) {};
 \coordinate (b3) at (4+5,2) {};
 \coordinate (b4) at (4+5,0) {};
 \coordinate (b5) at (1+5,0) {};
 \coordinate (b6) at (1+5,2) {};
 \coordinate (b7) at (3+5,0) {};
 \coordinate (b8) at (3+5,2) {};
 \node (b9) at (0.5+5,1) {$A_2$};
 \node (b10) at (3.5+5,1) {$B_2$};
 \node[circle,draw,minimum size = 0.5mm,inner sep=0mm] (b11) at (0.4+5,0.3) {$a_2$}; 
 \node[circle,draw,minimum size = 0.5mm,inner sep=0mm] (b12) at (3.6+5,0.5) {$b_2$}; 
 \node (b13) at (2+5,2.5) {$T_2$};
 \node (b14) at (2+5,-0.5) {$V_2$};

 \draw (b1) -- (b2) -- (b3) -- (b4) -- (b1);
 \draw (b5) -- (b6);
 \draw (b7) -- (b8);
 \draw (a12) -- (b11);
 
 \coordinate (c1) at (0+10,0) {};
 \coordinate (c2) at (0+10,2) {};
 \coordinate (c3) at (2+10,2) {};
 \coordinate (c4) at (2+10,0) {};
  \node[circle,draw,minimum size = 0.5mm,inner sep=0mm] (c5) at (0.4+10,0.3) {$a_3$}; 
 \node[circle,draw,minimum size = 0.5mm,inner sep=0mm] (c6) at (1.6+10,0.3) {$b_3$}; 
 \coordinate (c7) at (0.8+10,0) {};
 \coordinate (c8) at (1.5+10,0) {};
 \coordinate (c9) at (0.9+10,2) {};
 \coordinate (c10) at (1.5+10,2) {};
 \node (c11) at (0.3+10,1) {$A_3$};
 \node (c12) at (1.7+10,1) {$B_3$};
 \node (c13) at (1+10,2.5) {$T_3$};
 \node (c14) at (1+10,-0.5) {$V_3$}; 

 \draw (c1) -- (c2) -- (c3) -- (c4) -- (c1);
 \draw (b12) -- (c5);
 \draw (c7) to[out=0,in=0] (c9);
 \draw (c8) to[out=190,in=170] (c10);

  \draw (2+13,1) ellipse (2.1 and 2.9);
 \coordinate (d1) at (0+13,0) {};
 \coordinate (d2) at (0+13,2) {};
 \coordinate (d3) at (2.5+13,2) {};
 \coordinate (d4) at (2.5+13,0) {};
 \coordinate (d5) at (1.9+13,-1) {};
 \coordinate (d6) at (1.9+13,3) {};
 \coordinate (d7) at (3.5+13,3) {};
 \coordinate (d8) at (3.5+13,-1) {};

 \coordinate (d9) at (-0.13+13,1) {};
 \coordinate (d10) at (3.15+13,3.4) {};
 \coordinate (d11) at (1+13,-1.52) {};
 \coordinate (d12) at (2.2+13,1.72) {};
 
 \node (d13) at (1+13,1) {$A_4$};
 \node (d14) at (3+13,1) {$B_4$};
 \node (d15) at (2+13,4.5) {$S_1$};
 \node (d16) at (1.1+13,2.8) {$S^1_1$};
 \node (d17) at (1.1+13,-0.4) {$S^2_1$};
 \node (d18) at (3+13,-0.4) {$S^3_1$};
 \node (d19) at (2+13,-2.3) {$V_4$};

 \node (d20) [circle,draw,minimum size = 0.5mm,inner sep=0mm] at (0.5+13,1.5) {$a_4$};
 \node (d21) [circle,draw,minimum size = 0.5mm,inner sep=0mm] at (3.5+13,1.5) {$b_4$};

\draw (c6) -- (d20);
 \draw (d2) -- (d3) -- (d4)-- (d1);
 \draw (d7) -- (d6) -- (d5) -- (d8);
 \draw (d9) to [out = 330, in = 250] (d10);
 \draw (d11) to [out = 20,in = 280] (d12);

 \draw (2+18,1) ellipse (2.1 and 2.9);
 \coordinate (e1) at (0+18,0) {};
 \coordinate (e2) at (0+18,2) {};
 \coordinate (e3) at (1.9+18,2) {};
 \coordinate (e4) at (1.9+18,0) {};
 \coordinate (e5) at (2.5+18,-1) {};
 \coordinate (e6) at (2.5+18,3) {};
 \coordinate (e7) at (3.5+18,3) {};
 \coordinate (e8) at (3.5+18,-1) {};

\coordinate  (e9) at (2.27+18,0) {};
 \coordinate (e10) at (3.15+18,-1.42) {};
 \coordinate (e11) at (-0.13+18,1) {};
 \coordinate (e12) at (3.15+18,3.4) {};
 \coordinate (e13) at (1+18,-1.52) {};
 \coordinate (e14) at (2.2+18,1.72) {};

 \node (e15) at (1+18,1) {$A_4$};
 \node (e16) at (3+18,1) {$B_4$};
 \node (e17) at (2+18,4.5) {$S_2$};
 \node (e18) at (1.1+18,2.8) {$S^1_2$};
 \node (e19) at (1.1+18,-0.4) {$S^2_2$};
 \node (e20) at (3+18,-0.4) {$S^3_2$};
 \node (e20) at (2.3+18,-1.4) {$S^4_2$};
 \node (e22) at (2+18,-2.3) {$V_5$};

 \node (e23) [circle,draw,minimum size = 0.5mm,inner sep=0mm] at (0.5+18,1.5) {$a_5$};

\draw (d21) -- (e23);
 \draw (e2) -- (e3) -- (e4)-- (e1);
 \draw (e7) -- (e6) -- (e5) -- (e8);
  \draw (e9) to [out =270, in = 170] (e10); 
 \draw (e11) to [out = 330, in = 250] (e12);
 \draw (e13) to [out = 0,in = 280] (e14);

\end{tikzpicture}
\end{center}
        
        To construct the polynomial, we go through all possible $\{V_i\}$ and $\{k_i\}$, and check that every $V_i$ is a connected set, for every $i < j$ there is an edge between $V_i$ and $V_j$ iff $i + 1 = j$, and for every $V_i$ of size less than $n / \theta$, the graph $G[V_i]$ can be contracted to a path $p_1, \dots, p_h$ where only nodes corresponding to $p_1$ can be adjacent to $V_{i - 1}$ and only nodes corresponding to $p_h$ can be adjacent to $V_{i + 1}$. 
        To do that we also go through all possible $\{A_i\}$ and $\{B_i\}$ where $A_i, B_i \subseteq V_i$, $\{a_i\}$ and $\{b_i\}$ where $a_i \in A_i$, $b_i \in B_i$, and $\{S^j_i\}$, and consider $\{T_i\}$, $\{A^t_i\}$ and $\{B^t_i\}$ consistent with $\{V_i\}$, $\{A_i\}$ and $\{B_i\}$.  
        
        \item[Variables.] Introduce $s(G, k) = O(\binom{n}{\le 3n/\theta} + \binom{n}{\le n/\theta}^3 k + \binom{n}{\le 3n/\theta}^{4} + n^2) = O^*(\binom{6k}{18k/\theta}^4)$ variables:
        \[X = \{x_{S} \colon S \subseteq V, |S| \le 3n/\theta\} \, ;\]
        \[Y = \{y_{T, t}^{A, B} \colon A, B \subseteq T \subseteq V, |T| \le n/\theta, t \le k\} \, ;\]
        \[Z = \{z_{V_1, V_2}^{B, A} \colon V_1, V_2 \subseteq V, B \subseteq V_1, A \subseteq V_2, |V_1|, |V_2| \le 3n/\theta\} \, ;\]
        \[W = \{w_{b, a} \colon b, a \in [n] \} \, .\]
        
        The mapping $\phi(G,k)$ assigns the following values to the variables:
        \begin{align*}
          x_{S} \mapsto [& G[S] \text{ is connected}] \, ; \\
          y_{T, t}^{A, B} \mapsto [& \text{it is possible to contract at most } t \text{ edges of } G[T] \\ 
          & \text{to obtain a path such that } A \text{ and } B \text{ contract to} \\
          & \text{the first and the last nodes in the path, respectively}] \, ; \\
          z_{V_1, V_2}^{B, A} \mapsto [& \text{the set of endpoints of edges between $V_1$ and $V_2$ is a subset of $B \cup A$}] \, ; \\
          w_{b, a} \mapsto [& \text{$b$ and $a$ are adjacent in $G$}] \, . 
        \end{align*}

        \item[Complexity.] Each variable $x_{S}$, $z_{V_1, V_2}^{B, A}$ and $w_{b, a}$ can be computed in polynomial time.
        \begin{lemma}
            $y_{T, t}^{A, B}$ can be computed in $O^*(2^{|T|})$ time.
        \end{lemma}
        \begin{proof}
            Let $p_1, \dots, p_h$ be a path obtained from $G[T]$ by contracting edges.      
            Let $T^1 \sqcup \dots \sqcup T^h$ be a partition of $T$ into sets where $T^i$ contracts to $p_i$.
            Let $\tau: T \to \{0, 1\}$ be a 2-coloring function that colors $v \in T^{i}$ to $[i \mod 2]$.
            We observe that given $T$ and $\tau$, we can obtain the partition $T^1 \sqcup \dots \sqcup T^h$ by finding connected components of each color.
            
            To compute $y_{T, t}^{A, B}$, we go through all $2^{|T|}$ 2-colorings and check that if we contract the corresponding $\{T^i\}$ we obtain a path, $\sum (|T^i| - 1) \le t$, $A \subseteq T^1$ and $B \subseteq T^h$.
        \end{proof}
        
        The mapping $\phi(G,k)$ can be~computed
        in time $O^*(\binom{6k}{18k/\theta}^4 \cdot 2^{6k/\theta})$.
        
        \item[Polynomial.] Before we present the polynomial, we introduce the following monomials.
        
        \begin{itemize}
            \item $C_{\{S^j_i\}}(X)$ checks that all sets $S^j_i$ are connected.
            Using it for proper $\{S^j_i\}$, we check that every $S_i$ is connected.
            \[
            C_{\{S^j_i\}}(X) =
            \prod\limits_{i, j}  x_{S^j_i}
            \]
            
            \item $P_{\{T_i\}, \{k^t_i\}}^{\{A^t_i\}, \{B^t_i\}}(Y)$ checks that, for every~$i$, there is 
            a~way to contract at most $k^t_i$ edges of $G[T_i]$ to obtain a path such that $A^t_i$ and $B^t_i$ contract to the first and the last nodes of that path, respectively. 
            \[
            P_{\{T_i\}, \{k^t_i\}}^{\{A^t_i\}, \{B^t_i\}}(Y) = 
            \prod\limits_{i} y_{T_i, k^t_i}^{A^t_i, B^t_i}
            \]
            
            \item $N_{\{V_i\}}^{\{S^j_i\}, \{T_i\}}(Z)$ checks that $\forall a, b$ such that $|a - b| > 1$, there is no edge between $V_a$ and $V_b$.
            \[
            N_{\{V_i\}}^{\{S^j_i\}, \{T_i\}}(Z) = 
            \prod\limits_{\substack{
            a, b, \\
            a + 1 < b
            }} \prod\limits_{\substack{
            M_1, M_2 \in \{T_i\} \cup \{S^j_i\}, \\
            M_1 \subseteq V_a, M_2 \subseteq V_b
            }} z^{\varnothing, \varnothing}_{M_1, M_2}
            \]
            
            \item $A^{\{S^j_i\}, \{T_i\}}_{\{V_i\}, \{A_i\}, \{B_i\}}(Z)$ checks that for every $i \in [m - 1]$ the set of endpoints of edges between $V_i$ and $V_{i + 1}$ is a subset of $B_i \cup A_{i + 1}$.
            Using it with the previous monomial, we check that $\forall i \in [t]$, regardless of how we contract $G[T_i]$ to a path, there is no edge connecting an internal node of that path to a node outside $T_i$.
            \[
            A^{\{S^j_i\}, \{T_i\}}_{\{V_i\}, \{A_i\}, \{B_i\}}(Z) =
            \prod\limits_{i} 
            \prod\limits_{\substack{
                M_1, M_2 \in \{T_i\} \cup \{S^j_i\}, \\
                M_1 \subseteq V_i, M_2 \subseteq V_{i + 1}
                }} z^{M_1 \cap B_i, M_2 \cap A_{i + 1}}_{M_1, M_2}
            \]
            
            \item $P_{\{a_i\}, \{b_i\}}(W)$ checks that $\forall i \in [m - 1]$, there is an edge $(b_i, a_{i + 1})$ between $V_i$ and $V_{i + 1}$.
            \[
            P_{\{a_i\}, \{b_i\}}(W) =  
            \prod\limits_{i} w_{b_i, a_{i + 1}}
            \]
        \end{itemize}

        Now, we can present the final polynomial $P_{s(G, k)}(X, Y, Z, W)$:
        for every partition $V_1 \sqcup \dots \sqcup V_m$ of $V$ into $m \le 7\theta$ blocks, $\{A_i\}$, $\{B_i\}$ such that $\forall i \in [m] \; A_i, B_i \subseteq V_i$, $\{a_i\}$ and $\{b_i\}$ such that $\forall i \in [m] \; a_i \in A_i$, $b_i \in B_i$, the corresponding $\{S_i\}$, $\{T_i\}$, $\{A^t_i\}$ and $\{B^t_i\}$, $\{k^t_i\}$ such that $\sum k^t_i + \sum (|S_i| - 1) \le k$  and a family $\{S^j_i\}$ of $\ell \le \theta$ sets of size at most $3n / \theta$ such that $\bigcup\limits_j S^j_i = S_i$ and $B(S^1_i, \dots, S^{\ell}_i)$ is connected, we add the following monomial:
        
        \[C_{\{S^j_i\}}(X) \cdot
        P_{\{T_i\}, \{k^t_i\}}^{\{A^t_i\}, \{B^t_i\}}(Y) \cdot
        N_{\{V_i\}}^{\{S^j_i\}, \{T_i\}}(Z) \cdot
        A^{\{S^j_i\}, \{T_i\}}_{\{V_i\}, \{A_i\}, \{B_i\}}(Z) \cdot
        P_{\{a_i\}, \{b_i\}}(W)
        \, .\]
        
        The number of~monomials added to~$P_{s(G,k)}$ is 
        $O((5\theta)^{n} \cdot 2^{n} \cdot  2^{n} \cdot n^{5\theta} \cdot n^{5\theta} \cdot k^{5\theta} \cdot \binom{n}{3n/\theta}^{\theta}) = 2^{O(n)} = 2^{O(k)}$.

        \item[Degree.] The degree of~$\mathcal P$ is $\theta^{O(1)}$.
    \end{description}

\paragraph{\problem{$k$-Set Splitting}.} Given a~family $\mathcal{F} \subseteq 2^{[n]}$ of~size $m = n^{O(1)}$ and an~integer~$k$, decide whether there exists a~partition $A \sqcup B=[n]$ 
that splits at~least $k$~sets of~$\mathcal F$.

\begin{description}
    \item[Kernel.] As~there exists a~kernel with~$2k$ sets over the universe of~size~$k$
    for \problem{$k$-Set Splitting}~\cite{LS09}, we~assume that $|\mathcal{F}| = m \le 2k$ and $n \le k$
    
    \item[Idea.] 
Recall that a~partition $A \sqcup B=[n]$ splits a~set~$S \subseteq [n]$ if there are exists $a,b \in S$ such that $a \in A$ and $b \in B$.
    For a~subfamily $\mathcal{F}' \subseteq \mathcal{F}$, 
    if $(A, B)$ splits exactly $q$~sets of~$\mathcal{F}'$, 
    then there exist two sets $A' \subseteq A$, $B' \subseteq B$ such that $|A'|, |B'| \leq q \leq |\mathcal{F}'|$ and pair $(A', B')$ splits at~least $q$~sets 
    of~$\mathcal{F}'$. Such sets $A'$ and $B'$ can be~constructed 
    by~picking a~pair $(a, b)$ from each split set of~$\mathcal{F}'$. If~there are $k$~sets in~$\mathcal{F}$ that are split by~$(A, B)$, one can partition them into $\theta$ subfamilies $\mathcal{F}_1, \dotsc, \mathcal{F}_\theta$ of size at most~$k / \theta$ such that $(A, B)$ splits all sets in each $\mathcal{F}_i$. Then, for each $\mathcal{F}_i$, one can choose pair of sets $A_i \subseteq A, B_i \subseteq B$ of size at most $k / \theta$ that splits all sets 
    in~$\mathcal{F}_i$.
    
    \item[Variables.] Introduce $s(G,k) = O(\binom{n}{\le k/\theta}^2\binom{m}{\le k/\theta}) = O^*(\binom{2k}{k/\theta}^3)$ variables:

    \[X=\left\{x_{A, B, L} \colon A, B \subseteq [n], |A|,|B| \le k/\theta, A \cap B = \varnothing, L \subseteq [m], |L| \le k/\theta\right\}\, .\]
    For a~set family~$\mathcal{S}$ of size~$m$ and a~set of~indices $L \subseteq [m]$ denote by~$\mathcal{S}^L$ a~subfamily of $\mathcal{S}$ defined by~indices of~$L$.
    The mapping $\phi(\mathcal F,k)$ assigns the following value 
    to~a~variable $x_{A, B, L}$:
    \[[\text{$(A, B)$ splits all sets in $\mathcal{F}^L$}] \, .\]
    
    \item[Complexity.] The mapping $\phi(\mathcal F,k)$ can be~computed
    in time $O^*(\binom{2k}{k/\theta}^3)$.
    \item[Polynomial.] For every partition $A \sqcup B = [n]$, sets $A_1, \dotsc, A_{\theta} \subseteq A$, $B_1, \dotsc, B_{\theta} \subseteq B$ such that $|A_i|, |B_i| \le k / \theta$ for all $i$ and disjoint sets of indices $L_1, \ldots, L_\theta \subseteq [m]$ such that $|L_i| \le k / \theta$ for all $i$, we add the following monomial to $P_{s(G, k)}$:
    \[\prod\limits_{\substack{i \in [\theta]}} x_{A_i, B_i, L_i}\]
    The number of monomials added to $P_{s(\mathcal F, k)}$
    is at most $O(\binom{n}{k/\theta}^{2\theta}\binom{m}{k/\theta}^\theta) = 2^{O(k)}$.
    \item[Degree.] The degree of~$\mathcal P$ is $\theta$.
\end{description}

\paragraph{\problem{Cluster Editing}.} Given a~graph $G(V, E)$ with $n$~nodes and an~integer~$k$, decide whether $G$~can be~transformed into a~cluster graph (i.e., a~set of disjoint cliques) using 
at~most $k$~edge modifications (additions and deletions).

\begin{description}
    \item[Kernel.] As~there exists a~kernel with~$2k$ nodes for \problem{Cluster Editing}~\cite{CM12}, we~assume that $n \le 2k$.
    
    \item[Idea.] Consider a~solution of~\problem{Cluster Editing} for~$G$. Let $H$~be a~cluster graph resulting from~$G$
    by~at~most $k$~edge modifications. Take all cliques of~size 
    at~most $n/\theta$ and group them into $t \le 2\theta$ blocks
    each having at~most $n/\theta$ nodes. Let $T_1, \dotsc, T_t$ be sets of nodes of those blocks. Let $S_1, \dotsc, S_s$ 
    be~sets of~nodes of~cliques of~size more than $n/\theta$.
    For every~$i \in [s]$, partition $S_i$ into at~most $\theta$~blocks of~size at~most $n/\theta$: $S_i=\bigsqcup S_i^j$.

    To~construct the polynomial, we go~through all possible $\{T_i\}$ and $\{S_i^j\}$ and check whether 
    it~is possible to~modify
    the graph to~obtain a~cluster graph that is consistent with $\{T_i\}$ and $\{S_i^j\}$. To do this, let $\mathcal{T} = \{T_i\}$, $\mathcal{S} = \{S_i\}$, $\widetilde{\mathcal{S}} = \bigsqcup\limits_{i, j} \{S_i^j\}$, $\mathcal{V} = \mathcal{T} \sqcup \mathcal{S}$, and $\mathcal{C} = \mathcal{T} \sqcup \widetilde{\mathcal{S}}$.
    Let $\mathcal{V} = \{V_1, \dots, V_{t + s}\}$ and $\mathcal{C} = \{C_1, \dots, C_m\}$, where $m \le 2\theta + \theta^2$.
    
    Let $k = \sum\limits_{1 \le a \le b \le m} k_{a, b}$, where:
    \begin{itemize}
        \item there are at~most $k_{a, a}$ edge modifications in every $C_a \in \mathcal{T}$;
        \item there are at~most $k_{b, b}$ edge additions in every $C_b \in \widetilde{\mathcal{S}}$;
        \item there are at~most $k_{a, b}$ edge additions between every pair of different $C_a$ and $C_b$, where $C_a, C_b \subseteq S_i$ for some $i$;
        \item there are at~most $k_{a, b}$ edge deletions between every pair of $C_a$ and $C_b$, where $C_a$ and $C_b$ belong to different sets from $\mathcal{V}$.
    \end{itemize}
    
    Now, we can go through all possible $\mathcal{T}$, $\mathcal{S}$, $\mathcal{V}$, and $\mathcal{C}$ that are consistent with each other, and $\{k_{a, b}\}$ such that $\sum k_{a, b} = k$, and check whether for every $a, b \in [m]$ there can be done at most $k_{a, b}$ edge modifications according to cases listed above.

    \item[Variables.] Introduce $s(G, k) = O(\binom{n}{\le n / \theta}^2 \cdot k) = O^*(\binom{2k}{2k / \theta}^2)$ variables:
    
    \[X = \{ x_{T, k'} \colon T \subseteq V, |T| \le n / \theta, k' \in [k] \}\,;\]
    \[Y = \{ y_{S, k'} \colon S \subseteq V, |S| \le n / \theta, k' \in [k] \}\,;\]
    \[Z = \{ z_{S_1, S_2, k'} \colon S_1 \sqcup S_2 \subseteq V, S_1, S_2 \le n / \theta, k' \in [k] \}\,;\]
    \[W = \{ w_{C_1, C_2, k'} \colon C_1 \sqcup C_2 \subseteq V, C_1, C_2 \le n / \theta, k' \in [k] \}\,.\]

    The mapping $\phi(G,k)$ assigns the following values to the variables:
    \begin{align*}
      x_{T, k'} \mapsto [& \le k' \text{ edge modifications are needed to make $G[T]$ a cluster graph}] \, ; \\
      y_{S, k'} \mapsto [& \le k' \text{ edge modifications are needed to make $G[S]$ a clique}] \, ; \\
      z_{S_1, S_2, k'} \mapsto [& \le k' \text{ edge modifications are needed to connect each pair of nodes} \\  
      & \text{$(v_1, v_2)$ by an edge, where $v_1 \in S_1$, $v_2 \in S_2$}] \, ; \\
      w_{C_1, C_2, k'} \mapsto [& \le k' \text{ edge modifications are needed to get rid of all edges $(v_1, v_2)$} \\
      & \text{where $v_1 \in C_1$, $v_2 \in C_2$}] \, . 
    \end{align*}

    \item[Complexity.] Each variable from $Y \cup Z \cup W$ can be computed in polynomial time. Each variable $x_{T, k'}$ can be computed in time $O^*(3^{|T|})$ by dynamic programming
    \footnote{Let $dp[S]$ be the minimum number of edge modifications that turn $G[S]$ into a cluster graph. Then $dp[S] = \min\limits_{X \subseteq S} (|\overline{E}(X)| + |E(X, S \setminus X)| + dp[S \setminus X])$ and one can compute $dp[S]$ for all $S \subseteq T$ in time $O^*(3^{|T|})$.}. 
    The mapping $\phi(G, k)$ can be computed in time $O^*(\binom{2k}{2k / \theta}^2 \cdot 3^{n / \theta}) = O^*(\binom{2k}{2k / \theta}^2 \cdot 3^{2k / \theta})$.
    
    \item[Polynomial.] For every $(\mathcal{T}, \mathcal{S}, \widetilde{\mathcal{S}}, \mathcal{V}, \mathcal{C}, \{k_{a, b}\})$, we add the following monomial:

    \[\prod\limits_{\substack{a \in [m] \colon \\ C_a \in \mathcal{T}}} x_{C_a, k_{a, a}} 
    \prod\limits_{\substack{b \in [m] \colon \\ C_b \in \widetilde{\mathcal{S}}}} y_{C_b, k_{b, b}}
    \prod\limits_{\substack{a, b \in [m] \colon \\ a < b, \\ \exists i \in [s] \; C_a, C_b \subseteq S_i}} z_{C_a, C_b, k_{a, b}}
    \prod\limits_{\substack{a, b \in [m] \colon \\ a < b, \\ \exists i, j \in [t + s] \colon i \neq j, \\ C_a \subseteq V_i, C_b \subseteq V_j}} w_{C_a, C_b, k_{a, b}}
    \, .\]

    \item[Degree.] The degree of~$\mathcal P$ is $\theta^{O(1)}$.
\end{description}

\paragraph{\problem{$k$-Path}.} Given a~directed graph $G(V, E)$ with $n$ nodes and integer $k$, decide whether there is a simple path with exactly $k$ nodes.

\begin{description}
    \item[Idea.] 
Let $V = [n]$. The color-coding technique, introduced by~\cite{AYZ95}, solves \problem{$k$-Path} as~follows:
    assign a~random color $c \in [k]$ to~every node; then,
    all nodes of a~$k$-path receive different colors with
    probability about $e^{-k}$; at~the same time, one can find
    such a~colorful path in~time $n^{O(1)}2^k$. This gives
    a~randomized $2^{O(k)}n^{O(1)}$-time algorithm for 
    \problem{$k$-Path}. A~way to~derandomize~it
    is~to~use a~$k$-perfect family~$\mathcal{F}$ of~hash functions $f \colon [n] \to [k]$: go~through all~$f \in \mathcal F$ and for each node $v \in [n]$, assign a~color~$f(v)$ (see \cite[Section~4]{AYZ95}). The key feature of~the family~$\mathcal F$~is: for any subset of~nodes~$S \subseteq [n]$ of~size~$k$, there exists a~function $f \in \mathcal{F}$ such that $f$~is injective on~$S$. It~guarantees that for any $k$-path there is a~coloring $f \in \mathcal F$ that assigns
    different colors to all nodes of~the path.
    \cite{AYZ95} gives a~construction of~$\mathcal F$ of~size
    $O(2^{O(k)}\log n)$. 
    For our purposes, we~need a~family of~smaller size.
    To~achieve this, 
    we~allow a~larger number of~colors.
    
    Take a~family~$\mathcal F$ of $(n, k, \theta k)$-splitters of~size $O(e^{(k /\theta)(1 + o(1))}n^{O(1)}) = O^*(e^{2k/\theta})$ that can be computed in time $2^{9k}n^{O(1)}$ guaranteed by~\cref{thm:splitters}.

Given a~coloring $f \colon [n] \to [\theta k]$, one can
    find an~$f$-colorful $k$-path in~time 
    $2^{\theta k}n^{O(1)}$~\cite[Lemma~3.1]{AYZ95}.
    
    The main idea of the polynomial formulation below is the following. For a~coloring $f \colon [n] \to [\theta k]$, an~$f$-colorful $k$-path~$\pi$ can be~partitioned into 
    $\theta$~paths $\pi_1, \dotsc, \pi_{\theta}$ such that each path uses at~most $k/\theta$ colors. Then, a~partition $\pi_1, \dotsc, \pi_{\theta}$ is~valid if the paths are color-disjoint and there~is an~edge from the last node of~$\pi_i$ to~the first node of~$\pi_{i+1}$, for all~$i \in [\theta-1]$.

    \item[Variables.] Introduce $s(G,k) = O\left(e^{2k/\theta}{\binom{\theta k}{\le k / \theta}}n^{O(1)}\right) = O^*\left(e^{2k/\theta}{\binom{\theta k}{k / \theta}}\right)$ variables:

    \[X=\{x_{f, C, u, v} \colon f \in \mathcal{F}, u, v \in V, C \subseteq [\theta k], |C| \le k / \theta\}\, ,\]
    \[Y=\{y_{u, v} \colon u, v \in V\}\, .\]
    
    The mapping $\phi(G,k)$ assigns the following values:
    \[x_{f, C, u, v} \mapsto [\text{there is an~$f$-colorful path from~$u$ to~$v$ that uses colors from~$C$ only}]\,,\]
    \[y_{u, v} \mapsto [(u, v) \in E]\, .\]
    
    \item[Complexity.] The mapping $\phi(G,k)$ can be~computed
    in time $O\left(2^{k / \theta}e^{2k/\theta}{\binom{\theta k}{\le k / \theta}}n^{O(1)}\right) = O^*\left((2e^2)^{k / \theta}{\binom{\theta k}{k / \theta}}\right)$.
    
    \item[Polynomial.] For every $f \in \mathcal{F}$, disjoint $C_1, \ldots, C_{\theta} \subseteq [\theta k]$ such that $|C_i| \le \frac{k}{\theta}$, and distinct nodes $v_1, \ldots, v_{2\theta} \in V$, add a~monomial:
    \[\prod\limits_{i \in [\theta]}x_{f, C_i, v_{2i - 1}, v_{2i}}
      \prod\limits_{i \in [\theta - 1]}y_{v_{2i}, v_{2i + 1}}\]
    The number of monomials added to $P_{s(G, k)}$
    is at most $e^{2k/\theta}\binom{\theta k}{k / \theta}^{\theta}n^{O(1)} = 2^{O(k)}n^{O(1)}$.
    \item[Degree.] The degree of~$\mathcal P$ is $2\theta - 1$.
\end{description}

\paragraph{\problem{$k$-Tree}.} Given a tree $H$ on $k$ nodes and a graph $G$ on $n$ nodes, decide whether there is a (not necessarily induced) copy of $H$ in $G$.

\begin{description}
     \item[Idea.] Let $\mathcal{F}$ be an $(n, k, \theta k)$-splitter from \cref{thm:splitters}. Let $H_1, \dotsc, H_m$ be subtrees of~$H$ resulting from applying \cref{st:tree_decomposition} to~$H$, where $m \le \theta$ and $|V(H_i)| \le 3k / \theta$. Let $\mathcal{I} \in V(G)^{V(H)}$ be an isomorphism between $H$ and some subtree of $G$. Let $T_i = V(H_i)$, and let $S_i = \mathcal{I}(T_i)$. 
     Let $C^T$ be connector nodes of $B(T_1, \dots, T_m)$, and let $C^S$ be connector nodes of $B(S_1, \dots, S_m)$. Let $C^T_i = T_i \cap C^T$, let $C^S_i = S_i \cap C^S$, and let $\mathcal{I}_i = \mathcal{I} \Big|_{C^T_i}$.
     Let $f \in \mathcal{F}$ be a coloring of $V(G)$ that assigns different colors to nodes from $\mathcal{I}(V(H))$, and let $F_i = f(S_i)$.

     To construct the polynomial, we go through all possible $\{T_i\}$, $\{C^S_i\}$, $\{F_i\}$ and $f$ such that $T_i \subseteq [k]$, $C^S_i \subseteq [n]$, $F_i \subseteq [\theta k]$, $|T_i| = |F_i| \le 3k / \theta$, $|C^S_i| \le \theta - 1$ and $f \in \mathcal{F}$. We check that $\bigcup C^S_i$ equals to the set of connector nodes of $B(C^S_1, \dots, C^S_m)$, and that $B(T_1, \dots, T_m)$, $B(C^S_1, \dots, C^S_m)$ and $B(F_1, \dots, F_m)$ are isomorphic trees where sets with the same numbers correspond to each other. 
     We compute $C^T$ as connector nodes of $B(T_1, \dots, T_m)$ and $C^T_i$ as $T_i \cap C^T$.
     Now, we define $\mathcal{I}_i$ as a bijection between $C^T_i$ and $C^S_i$ obtained from the above isomorphism between $B(T_1, \dots, T_m)$ and $B(C^S_1, \dots, C^S_m)$ as a restriction to connector nodes. 
     
     Now, we have $(\{T_i\}, \{C^T_i\}, \{C^S_i\}, \{\mathcal{I}_i\}, \{F_i\}, f)$, and we need to check that
     \begin{itemize}
         \item for all~$i$, $H(T_i)$ is a~tree;
         \item for all~$i$, there exists $S_i \subseteq [n]$ and $\mathcal{J} \in \operatorname{Bij}(T_i, S_i)$ such that: 
         \begin{itemize}
             \item $ C^S_i \subseteq S_i$,
             \item $f(S_i) = F_i$,
             \item $H[T_i]$ is isomorphic to $G[S_i]$ according to $\mathcal{J}$,
             \item $\mathcal{J} \Big|_{C^T_i} = \mathcal{I}_i$.
         \end{itemize}
     \end{itemize}

    \item[Variables.] Introduce $s(H, G, k) = O(\binom{k}{\le 3k/\theta} k^{\theta} n^{\theta} \theta^{\theta} \binom{\theta k}{\le 3k / \theta} e^{2k/\theta} \log n)$ variables:
    \begin{align*}
       X = \{& x_{T} \colon T \subseteq [k], |T| \le 3k / \theta \}; \\
       Y = \{& y_{T, C^T, C^S, \mathcal{I}, F, f} \colon T, C^T \subseteq [k], C^S \subseteq [n], \mathcal{I} \in (C^S)^{C^T}, F \subseteq [\theta k], \\
       & f \in \mathcal{F}, |T| = |F| \le 3k / \theta, |C^T| = |C^S| \le \theta - 1\}\, .
    \end{align*}
 
    The mapping $\phi(H, G, k)$ assigns the following values to the variables:
    \begin{align*}
       x_{T} \mapsto [& H[T] \text{ is a tree}] \, ; \\
       y_{T, C^T, C^S, \mathcal{I}, F, f} \mapsto [& \text{there exists a set $S \subseteq[n]$ such that $C^S \subseteq S$, $f(S) = F$,} \\
       & \text{and $G[S]$ is a copy of $H[T]$, where nodes of $H$ in $C^T$} \\
       & \text{correspond to nodes of $G$ in $C^S$ according to $\mathcal{I}$}] \, .
    \end{align*}

    \item[Complexity.] The family $\mathcal{F}$ can be computed in time $2^{9k}n^{O(1)}$. Each variable from $X$ can be computed in polynomial time. Each variable from $Y$ can be computed in time $O^*(8^{3k/\theta})$ by the following dynamic programming
    algorithm. Let us view $H[T]$ as a rooted tree. Let $T' \subseteq T$ such that $H[T']$ is connected, let $F' \subseteq F$, and let $v \in [n]$. Let $dp[T'][F'][v] = [$there exists a subset $S \subseteq [n]$ such that $v \in S$, $|S| = |F'|$, $f(S) = F'$ and $H[T']$ is isomorphic to $G[S]$ according to some isomorphism $\mathcal{J}$ which is consistent with $\mathcal{I}$ and which maps the root of $H[T']$ into $v$]. For $T'$, let $T''$ be a subset of $T'$ such that $H[T'']$ and $H[T' \setminus T'']$ are connected, and the root of $H[T'']$ is a son of the root of $H[T']$. Then, 
    \[dp[T'][F'][v] = \bigvee\limits_{\substack{F'' \subseteq F', \\ u \in [n] \colon (u, v) \in E(G)}} 
    (dp[T''][F''][u] \wedge dp[T' \setminus T''][F' \setminus F''][v]) \,. \]
    
    \item[Polynomial.] For every $(\{T_i\}, \{C^T_i\}, \{C^S_i\}, \{\mathcal{I}_i\}, \{F_i\}, f)$, we add the following monomial:
    \[\prod\limits_{i = 1}^{m} x_{T_i} \prod\limits_{i = 1}^{m} y_{T_i, C^T_i, C^S_i, \mathcal{I}_i, F_i, f} \, . \]
    
    The number of monomials is at most $\binom{k}{3k/\theta}^{\theta} 
    \binom{k}{\theta}^{\theta} \binom{n}{\theta}^{\theta} 
    (\theta^{\theta})^{\theta} \binom{\theta k}{3k / \theta}^{\theta}
    e^{2k / \theta} \log n = 2^{O(k)}n^{O(1)}$.
    
    \item[Degree.] The degree of~$\mathcal P$ is at most $2 \theta$.
 \end{description}
    
\end{proof}

\section*{Acknowledgments}
Research is~partially supported by Huawei (grant TC20211214628)
and by the Ministry of Science and Higher Education of the Russian Federation (agreement 075-15-2019-1620 date 08/11/2019 and 075-15-2022-289 date 06/04/2022).
We~are indebted to~anonymous reviewers and to
	Shyan Shaer Akmal and Ryan Willams
	for their helpful comments that significantly helped us
	to~improve the exposition of the paper and to~fix a~number
	of~issues in~the proofs.

\bibliographystyle{alpha} 
\bibliography{main.bib}

\newcommand{\etalchar}[1]{$^{#1}$}
\begin{thebibliography}{HVHL{\etalchar{+}}14}

\bibitem[AB09]{AB2009}
Sanjeev Arora and Boaz Barak.
\newblock {\em Computational complexity: a modern approach}.
\newblock Cambridge University Press, 2009.

\bibitem[ABN{\etalchar{+}}92]{ABNNR92}
Noga Alon, Jehoshua Bruck, Joseph Naor, Moni Naor, and Ron~M. Roth.
\newblock Construction of asymptotically good low-rate error-correcting codes
  through pseudo-random graphs.
\newblock {\em IEEE Transactions on information theory}, 38(2):509--516, 1992.

\bibitem[AKS04]{aks04}
Manindra Agrawal, Neeraj Kayal, and Nitin Saxena.
\newblock {PRIMES is in P}.
\newblock {\em Annals of Mathematics}, 160:781--793, 2004.

\bibitem[Alo86]{A86}
Noga Alon.
\newblock Explicit construction of exponential sized families of
  $k$-independent sets.
\newblock {\em Discrete Mathematics}, 58(2):191--193, 1986.

\bibitem[AS08]{AS08}
Noga Alon and Joel~H. Spencer.
\newblock {\em The Probabilistic Method, Third Edition}.
\newblock Wiley, 2008.

\bibitem[AYZ95]{AYZ95}
Noga Alon, Raphael Yuster, and Uri Zwick.
\newblock Color-coding.
\newblock {\em Journal of the ACM}, 42(4):844--856, 1995.

\bibitem[BCS97]{DBLP:books/daglib/0090316}
Peter B{\"{u}}rgisser, Michael Clausen, and Mohammad~Amin Shokrollahi.
\newblock {\em Algebraic complexity theory}.
\newblock Springer, 1997.

\bibitem[Bel62]{bellman1962dynamic}
Richard Bellman.
\newblock Dynamic programming treatment of the travelling salesman problem.
\newblock {\em Journal of the ACM}, 9(1):61--63, 1962.

\bibitem[BHK09]{DBLP:journals/siamcomp/BjorklundHK09}
Andreas Bj{\"{o}}rklund, Thore Husfeldt, and Mikko Koivisto.
\newblock Set partitioning via inclusion-exclusion.
\newblock {\em SIAM Journal on Computing}, 39(2):546--563, 2009.

\bibitem[BI15]{DBLP:journals/siamcomp/BackursI18}
Arturs Backurs and Piotr Indyk.
\newblock Edit distance cannot be computed in strongly subquadratic time
  (unless {SETH} is false).
\newblock In {\em STOC 2015}, pages 51--58, 2015.

\bibitem[Bj{\"{o}}10]{B10}
Andreas Bj{\"{o}}rklund.
\newblock {Determinant Sums for Undirected Hamiltonicity}.
\newblock In {\em FOCS 2010}, pages 173--182. {IEEE}, 2010.

\bibitem[BS83]{bs83}
Walter Baur and Volker Strassen.
\newblock The complexity of partial derivatives.
\newblock {\em Theoretical computer science}, 22(3):317--330, 1983.

\bibitem[CDL{\etalchar{+}}16]{DBLP:journals/talg/CyganDLMNOPSW16}
Marek Cygan, Holger Dell, Daniel Lokshtanov, D{\'{a}}niel Marx, Jesper
  Nederlof, Yoshio Okamoto, Ramamohan Paturi, Saket Saurabh, and Magnus
  Wahlstr{\"{o}}m.
\newblock On problems as hard as {CNF-SAT}.
\newblock {\em ACM Transactions on Algorithms}, 12(3):41:1--41:24, 2016.

\bibitem[CFK{\etalchar{+}}15]{cygan2015parameterized}
Marek Cygan, Fedor~V. Fomin, {\L}ukasz Kowalik, Daniel Lokshtanov, D{\'a}niel
  Marx, Marcin Pilipczuk, Micha{\l} Pilipczuk, and Saket Saurabh.
\newblock {\em Parameterized algorithms}, volume~5.
\newblock Springer, 2015.

\bibitem[CGI{\etalchar{+}}16]{DBLP:conf/innovations/CarmosinoGIMPS16}
Marco~L. Carmosino, Jiawei Gao, Russell Impagliazzo, Ivan Mihajlin, Ramamohan
  Paturi, and Stefan Schneider.
\newblock Nondeterministic extensions of the strong exponential time hypothesis
  and consequences for non-reducibility.
\newblock In {\em ITCS 2016}, pages 261--270. {ACM}, 2016.

\bibitem[CKJ01]{chen2001vertex}
Jianer Chen, Iyad~A. Kanj, and Weijia Jia.
\newblock Vertex cover: further observations and further improvements.
\newblock {\em Journal of Algorithms}, 41(2):280--301, 2001.

\bibitem[CM12]{CM12}
Jianer Chen and Jie Meng.
\newblock A $2k$ kernel for the cluster editing problem.
\newblock {\em Journal of Computer and System Sciences}, 78(1):211--220, 2012.

\bibitem[DFF{\etalchar{+}}06]{dehne2006nonblocker}
Frank Dehne, Michael Fellows, Henning Fernau, Elena Prieto, and Frances
  Rosamond.
\newblock Nonblocker: parameterized algorithmics for minimum dominating set.
\newblock In {\em SOFSEM 2006}, pages 237--245. Springer, 2006.

\bibitem[Eri99]{e99}
Jeff Erickson.
\newblock Bounds for linear satisfiability problems.
\newblock {\em Chicago Journal of Theoretical Computer Science}, page~8, 1999.

\bibitem[FGST13]{fomin2013linear}
Fedor~V. Fomin, Serge Gaspers, Saket Saurabh, and St{\'e}phan Thomass{\'e}.
\newblock A linear vertex kernel for maximum internal spanning tree.
\newblock {\em Journal of Computer and System Sciences}, 79(1):1--6, 2013.

\bibitem[FK10]{DBLP:series/txtcs/FominK10}
Fedor~V. Fomin and Dieter Kratsch.
\newblock {\em Exact Exponential Algorithms}.
\newblock Texts in Theoretical Computer Science. An {EATCS} Series. Springer,
  2010.

\bibitem[FKM{\etalchar{+}}07]{fuchs2007dynamic}
Bernhard Fuchs, Walter Kern, D~Molle, Stefan Richter, Peter Rossmanith, and
  Xinhui Wang.
\newblock Dynamic programming for minimum {Steiner trees}.
\newblock {\em Theory of Computing Systems}, 41(3):493--500, 2007.

\bibitem[FMRU00]{fellowscoordinatized}
Michael~R. Fellows, Catherine McCartin, Frances~A. Rosamond, and Stege Ulrike.
\newblock Coordinatized kernels and catalytic reductions: An improved {FPT}
  algorithm for max leaf spanning tree and other problems.
\newblock In {\em FSTTCS 2000}, pages 240--251. Springer, 2000.

\bibitem[Fri84]{F84}
Joel Friedman.
\newblock Constructing {$O(n\log{n})$} size monotone formulae for the $k$-th
  elementary symmetric polynomial of $n$ boolean variables.
\newblock In {\em FOCS 1984}, pages 506--515. IEEE, 1984.

\bibitem[GO95]{go95}
Anka Gajentaan and Mark~H. Overmars.
\newblock On a class of {$O(n^2)$} problems in computational geometry.
\newblock {\em Computational geometry}, 5(3):165--185, 1995.

\bibitem[GP18]{DBLP:journals/jacm/GronlundP18}
Allan Gr{\o}nlund and Seth Pettie.
\newblock Threesomes, degenerates, and love triangles.
\newblock {\em Journal of the ACM}, 65(4):22:1--22:25, 2018.

\bibitem[HK62]{held1962dynamic}
Michael Held and Richard~M. Karp.
\newblock A dynamic programming approach to sequencing problems.
\newblock {\em Journal of the Society for Industrial and Applied mathematics},
  10(1):196--210, 1962.

\bibitem[HKNS15]{hkns15}
Monika Henzinger, Sebastian Krinninger, Danupon Nanongkai, and Thatchaphol
  Saranurak.
\newblock Unifying and strengthening hardness for dynamic problems via the
  online matrix-vector multiplication conjecture.
\newblock In {\em STOC 2015}, pages 21--30. ACM, 2015.

\bibitem[HU79]{HU79}
John~E. Hopcroft and Jeffrey~D. Ullman.
\newblock {\em Introduction to Automata Theory, Languages and Computation}.
\newblock Addison-Wesley, 1979.

\bibitem[HVHL{\etalchar{+}}14]{heggernes2014contracting}
Pinar Heggernes, Pim Van’t~Hof, Benjamin L{\'e}v{\^e}que, Daniel Lokshtanov,
  and Christophe Paul.
\newblock Contracting graphs to paths and trees.
\newblock {\em Algorithmica}, 68(1):109--132, 2014.

\bibitem[IP99]{ip99}
Russell Impagliazzo and Ramamohan Paturi.
\newblock The complexity of {$k$-SAT}.
\newblock In {\em CCC 1999}, pages 237--240. {IEEE}, 1999.

\bibitem[IPZ98]{ipz98}
Russell Impagliazzo, Ramamohan Paturi, and Francis Zane.
\newblock Which problems have strongly exponential complexity?
\newblock In {\em FOCS 1998}, pages 653--653. IEEE, 1998.

\bibitem[JMV15]{DBLP:journals/iandc/JahanjouMV18}
Hamidreza Jahanjou, Eric Miles, and Emanuele Viola.
\newblock Local reductions.
\newblock In {\em ICALP 2015}, pages 749--760. Springer, 2015.

\bibitem[KI03]{DBLP:journals/cc/KabanetsI04}
Valentine Kabanets and Russell Impagliazzo.
\newblock Derandomizing polynomial identity tests means proving circuit lower
  bounds.
\newblock In {\em STOC 2003}, pages 355--364. {ACM}, 2003.

\bibitem[KLR08]{kneis2008new}
Joachim Kneis, Alexander Langer, and Peter Rossmanith.
\newblock A new algorithm for finding trees with many leaves.
\newblock In {\em ISAAC 2008}, pages 270--281. Springer, 2008.

\bibitem[LMS11]{DBLP:journals/eatcs/LokshtanovMS11}
Daniel Lokshtanov, D{\'{a}}niel Marx, and Saket Saurabh.
\newblock Lower bounds based on the exponential time hypothesis.
\newblock {\em Bulletin of {EATCS}}, 105:41--72, 2011.

\bibitem[LMS18]{DBLP:journals/talg/LokshtanovMS18}
Daniel Lokshtanov, D{\'{a}}niel Marx, and Saket Saurabh.
\newblock Known algorithms on graphs of bounded treewidth are probably optimal.
\newblock {\em ACM Transactions on Algorithms}, 14(2):13:1--13:30, 2018.

\bibitem[LP19]{lp19}
{Hendrik W.} {Lenstra, Jr.} and Carl Pomerance.
\newblock Primality testing with gaussian periods.
\newblock {\em Journal of the European Mathematical Society}, 21(4):1229--1269,
  2019.

\bibitem[LS09]{LS09}
Daniel Lokshtanov and Saket Saurabh.
\newblock Even faster algorithm for set splitting!
\newblock In {\em IWPEC 2009}, pages 288--299. Springer, 2009.

\bibitem[LY22]{31n}
Jiatu Li and Tianqi Yang.
\newblock $3.1n-o(n)$ circuit lower bounds for explicit functions.
\newblock In {\em STOC 2022}, pages 1180--1193. {ACM}, 2022.

\bibitem[NSS95]{NSS95}
Moni Naor, Leonard~J. Schulman, and Aravind Srinivasan.
\newblock Splitters and near-optimal derandomization.
\newblock In {\em FOCS 1995}, pages 182--191. IEEE, 1995.

\bibitem[PW10]{DBLP:conf/soda/PatrascuW10}
Mihai P{\v a}tra\c{s}cu and Ryan Williams.
\newblock On the possibility of faster {SAT} algorithms.
\newblock In {\em SODA 2010}, pages 1065--1075. {SIAM}, 2010.

\bibitem[RV13]{roditty2013fast}
Liam Roditty and Virginia {Vassilevska Williams}.
\newblock Fast approximation algorithms for the diameter and radius of sparse
  graphs.
\newblock In {\em STOC 2013}, pages 515--524, 2013.

\bibitem[Str73a]{s73}
Volker Strassen.
\newblock Die berechnungskomplexit{\"a}t von elementarsymmetrischen funktionen
  und von interpolationskoeffizienten.
\newblock {\em Numerische Mathematik}, 20(3):238--251, 1973.

\bibitem[Str73b]{strassen1973vermeidung}
Volker Strassen.
\newblock Vermeidung von divisionen.
\newblock {\em Journal f{\"u}r die reine und angewandte Mathematik},
  264:184--202, 1973.

\bibitem[SY10]{DBLP:journals/fttcs/ShpilkaY10}
Amir Shpilka and Amir Yehudayoff.
\newblock Arithmetic circuits: {A} survey of recent results and open questions.
\newblock {\em Foundations and Trends in Theoretical Computer Science},
  5(3-4):207--388, 2010.

\bibitem[Val77]{Valiant}
Leslie~G. Valiant.
\newblock Graph-theoretic arguments in low-level complexity.
\newblock In {\em MFCS 1977}, pages 162--176, 1977.

\bibitem[{Vas}15]{DBLP:conf/iwpec/Williams15}
Virginia {Vassilevska Williams}.
\newblock Hardness of easy problems: Basing hardness on popular conjectures
  such as the strong exponential time hypothesis (invited talk).
\newblock In {\em IPEC 2015}, pages 17--29. Schloss Dagstuhl, 2015.

\bibitem[{Vas}18]{vw18}
Virginia {Vassilevska Williams}.
\newblock On some fine-grained questions in algorithms and complexity.
\newblock In {\em ICM 2018}, 2018.

\bibitem[VW10]{vww10}
Virginia {Vassilevska Williams} and Ryan Williams.
\newblock Subcubic equivalences between path, matrix and triangle problems.
\newblock In {\em STOC 2010}, pages 645--654. IEEE, 2010.

\bibitem[Wil05]{w05}
Ryan Williams.
\newblock A new algorithm for optimal 2-constraint satisfaction and its
  implications.
\newblock {\em Theoretical Computer Science}, 348(2-3):357--365, 2005.

\bibitem[Wil09]{williams2009finding}
Ryan Williams.
\newblock Finding paths of length $k$ in {$O^*(2^k)$} time.
\newblock {\em Information Processing Letters}, 109(6):315--318, 2009.

\bibitem[Wil16]{W16}
Ryan Williams.
\newblock {Strong ETH Breaks With Merlin and Arthur: Short Non-Interactive
  Proofs of Batch Evaluation}.
\newblock In {\em CCC 2016}, volume~50, pages 2:1--2:17. Dagstuhl, 2016.

\bibitem[Zam21]{DBLP:conf/icalp/Zamir21}
Or~Zamir.
\newblock Breaking the $2^n$ barrier for $5$-coloring and $6$-coloring.
\newblock In {\em ICALP 2021}, volume 198 of {\em LIPIcs}, pages 113:1--113:20,
  2021.

\end{thebibliography}

\appendix
\section{Omitted Proofs}\label{sec:apx}
\subsection{Proof of Theorem~\ref{thm:mainparam}}\label{apx:mainparam}
\mainparam*
\begin{proof}
Let $\lambda>1$ be the constant from the theorem statement, $\gamma>1$ be an arbitrary constant, and $\sigma=\log(\lambda)/(6\gamma)$. Let $\mathcal I \times \N$ be the set of all instances of~$A$, where for an instance $(x,k)\in \mathcal I \times \N$, $k$~is the value of the parameter. Let $\mathcal P$ be a~$\Delta$-polynomial formulation of~$A$ of~complexity $2^{\sigma k}$,
	for constant $\Delta=\Delta(\sigma)>0$.
	We assume that $A$ is $\lambda^k$-SETH-hard: there is a function $\delta\colon\R_{>0}\to\R_{>0}$ such that for every $q,d\in\N$, \problem{$q$-SAT} can be solved in time $2^{(1-\delta)n}$ given a~$\lambda^{(1-\eps)k} |x|^d$-algorithm for~$A$.
	
	If ${\mathcal P}=(P_t)_{t\geq1}$ does not have arithmetic circuits over $\Z$ of size $t^\gamma$ for infinitely many values of~$t$, then we have an explicit family of constant-degree polynomials that requires arithmetic circuits of~size $\Omega(t^{\gamma})$. 
	Hence, in the following we assume that  ${\mathcal P}$ has arithmetic circuits over $\Z$ of size $ct^\gamma$ for all values of~$t$ for a constant~$c>0$. Under this assumption, we design a non-deterministic algorithm solving \problem{$q$-TAUT} in time $2^{(1-\eps)n}$ for every~$q$. This contradicts NSETH and, by \cref{thm:nseth}, implies a super-linear lower bound on the size of series-parallel circuits computing~$\E^{\NP}$.
	
Let $\delta_0=\delta(1/2)\in(0,1)$ where $\delta$ is the function from the SETH-hardness reduction for~$A$. Let $\alpha= \frac{1}{(\Delta+1)(\gamma+2\Delta+8)}$, $L=2^{(1-\delta_0)\alpha n}$,  $K=2(1-\delta_0) \alpha n/\log(\lambda)$, and $T=2^{(\sigma K)} \cdot L^\Delta$. We will start with an instance of the \problem{$q$-TAUT} problem on~$n$ variables, reduce it to $2^{(1-\alpha)n}$ instances of \problem{$q$-TAUT} on $\alpha n$ variables each . Then we'll use the fine-grained reduction from \problem{$q$-SAT} to the problem~$A$ on instances of length~$\ell\leq L$ and parameter~$k\leq K$. Finally, we'll use the polynomial formulation of~$A$ to reduce instances of length~$\ell$ and~parameter~$k$ to polynomials with $t\leq T$ variables.

Let $F$~be a~$k$-DNF formula over $n$~variables. In order to solve~$F$, we branch on all but $\alpha n$~variables. This gives us $2^{(1-\alpha)n}$ $k$-DNF formulas. By~solving \problem{$q$-SAT} on the negations of all of~these formulas, we~solve \problem{$q$-TAUT} on the original formula~$F$. 

Assuming a $\lambda^{k/2}$-algorithm for~$A$ we have a $2^{(1-\delta_0)\alpha n}$-algorithm for \problem{$q$-SAT} on $\alpha n$ variables for $\delta_0=\delta(1/2)>0$. We now apply the
fine-grained reduction from \problem{$q$-SAT} to~$A$ to (the negations of) all $2^{(1-\alpha)n}$ instances of \problem{$q$-TAUT}. This gives us a number of instances of the problem~$A$. Let $\ell$~be the largest length of these instances, and $k$ be the largest value of the parameter in these instances. Since the running time of the reduction is bounded from above by~$2^{(1-\delta_0)\alpha n}$, we have that $\ell \leq 2^{(1-\delta_0)\alpha n} = L$.  From \cref{def:sethparam}, we~know that  ${\lambda}^{k/2}< 2^{(1-\delta_0) \alpha n}$, so each instance of~$A$ has the parameter value at most $k<2(1-\delta_0) \alpha n/\log(\lambda)=K$. 

    Since $\mathcal P$ is a polynomial formulation of~$A$ of~complexity $2^{\sigma k}$, there exist $s \colon \mathcal I \times \mathbb{N} \to \mathbb N$, $s(x, k)\leq 2^{\sigma k} |x|^{\Delta} \leq 2^{\sigma K} L^\Delta$ and $\phi \colon \mathcal I \times \mathbb{N} \to \Z^*$ (both computable in time $2^{\sigma K} L^\Delta$) such that for every $x\in \mathcal I$ and $k\in\N$,
	\begin{itemize}
		\item $P_{s(x,k)}(\phi(x,k)) \neq 0 \Leftrightarrow (x,k) \text{ is a yes instance of } $A$\,;$
		\item $|P_{s(x, k)}(\phi(x, k))| < 2^{2^{\sigma k}|x|^{\Delta}}\,.$
	\end{itemize}
    Using~$\mathcal P$,
    we will solve all instances
    of~$A$ in~two stages: in~the preprocessing stage (which takes place before all the reductions), we~guess efficient arithmetic circuits for polynomials~$P_{t}$ for all 
    $t \le T$,
    in the solving stage, we~solve all instances of~$A$ using the guessed circuits. Note that we'll be using the polynomials to solve instances of~$A$ resulting from \problem{$q$-SAT} instances on $\alpha n$ variables. Since~$L$ is the largest length of such an instance of~$A$ and $K$ is the largest value of the parameter, we have that each such instance is mapped to a polynomial with at most $s(x,k)\leq 2^{\sigma k} |x|^\Delta \leq 2^{\sigma K} L^\Delta=T$ variables. Therefore, finding efficient arithmetic circuits for polynomials~$P_{t}$ for all 
    $t \le T$ will be sufficient for solving all \problem{$q$-SAT} instances of size $\alpha n$ that we obtain in the reductions.

        \paragraph{Preprocessing.} For every $t \leq T$, we find a prime $p_t$ in the interval $2^{t+1} \le p_t \le 2^{t+2}$ in non-deterministic time~$O(t^7)$~\cite{aks04,lp19}. 
        
        Now for every $t\leq T$, we reduce all coefficients of the polynomial $P_t$ modulo~$p_t$ to obtain a polynomial $Q_t$ over $\Z_{p_t}$, and let ${\mathcal Q}=(Q_1,Q_2,\ldots)$.
        For every $t \le T$, we now non-deterministically solve $\operatorname{Gap-MACP}_{\mathcal Q, \mu \Delta^2, p_t}(t, ct^{\gamma})$ using \cref{lemma:macp}.
        Since we assume that $\mathcal P$ has arithmetic circuits over $\Z$ of size $ct^\gamma$, we have that $\mathcal Q$ has arithmetic circuits over $\Z_{p_t}$ of this size. Thus, we~obtain arithmetic circuits~$C_t$ of~size at~most
        \begin{equation}\label{eq:circuitsizeparam}
            c\mu \Delta^2t^{\gamma}
        \end{equation}
        computing~$Q_{t}$ over $\Z_p$ for all $t \le T$. Since~$C_t$ computes~$Q_t$ correctly in~$\mathbb{Z}_p$ and $|P_{s(x,k )}(\phi(x,k))| < 2^{s(x,k)} \leq  p_{s(x,k)}/2$ for all $x \in I_\ell$, we can use $C_t$
        to~solve~$A$ for every instance length $\ell \leq L$ and parameter $k \leq K$. By~\cref{lemma:macp}, $\operatorname{Gap-MACP}_{\mathcal Q, \mu \Delta^2, p_t}(t, ct^{\gamma})$ can be solved in (non-deterministic) time
        \[
        	O\left(\Delta^2 \cdot ct^{\gamma} \cdot t^{2\Delta} \cdot \log^2(p_t)\right)=
        		O\left(T^{\gamma+2\Delta+2}\right) \, .
        \]
        The total (non-deterministic) running time of the preprocessing stage is then bounded from above by the time needed to~find~$T$ prime numbers, write down the corresponding explicit polynomials modulo~$p_t$, and solve~$T$ instances of $\operatorname{Gap-MACP}$:
        \begin{align}\label{eq:preprocessingparam}
        	O\left(T(T^7+T^{\Delta+2}+T^{\gamma+2\Delta+2})\right)=
        		O\left(T^{\gamma+2\Delta+8}\right)
        		=O\left(2^{(1-\delta_0)n}\right) \, ,
        \end{align}
        where the last equality holds due to $T=2^{\sigma K} L^\Delta$, $L=2^{(1-\delta_0)\alpha n}$, $K=2(1-\delta_0) \alpha n/\log(\lambda)$, $\sigma=\log(\lambda)/(6\gamma)$, and $\alpha= \frac{1}{(\Delta+1)(\gamma+2\Delta+8)}$.
        \paragraph{Solving.} In~the solving stage, we~solve all $2^{(1-\alpha)n}$
        instances of~\problem{$q$-SAT} by~reducing them to~$A$ and using
        efficient circuits found in~the preprocessing stage. For an instance~$x$ of~$A$ of length $\ell$ with parameter~$k$, we first transform it into an input of the polynomial $y=\phi(x,k)\in\Z^{s(x,k)}$. Both $s(x,k)$ and $\phi(x,k)$ can be computed in time $O(2^{\sigma k}\ell^\Delta)$.  Then we feed it into the circuit $Q_{s(x,k)}$. First we note that we have the circuit $Q_{s(x,k)}$ after the preprocessing stage as $s(x,k)\leq T$ and we have circuits $(Q_1,\ldots,Q_T)$. The number of arithmetic operations in~$\Z_{p_{s(x,k)}}$ required to evaluate the circuit is proportional to the circuit size, and each arithmetic operation takes time $\log^2(p_{s(x,k)})=O(s(x,k)^2)$. From~\eqref{eq:circuitsizeparam} with $t\leq s(x,k)\leq2^{\sigma k }\ell^\Delta$, we have that we can solve an instance of~$A$ with $\ell$ inputs and parameter~$k$ in time 
        \[
        O(2^{\sigma k}\ell^\Delta) + c\mu \Delta^2 \cdot s(x,k)^2 \cdot 2^{\sigma\gamma k}\cdot\ell^{\Delta \gamma}
        =O\left( 2^{3\sigma\gamma k} \ell^{3\Delta\gamma}\right)
        =O\left(\lambda^{k/2} \ell^{3\Delta\gamma}\right)\,,
        \]
        where the last equality holds due to the choice of $\sigma=\log(\lambda)/(6\gamma)$.
        The~fine-grained reduction from \problem{$q$-SAT} to~$A$ implies that a~$O\left(\lambda^{k/2} \ell^{O(1)}\right)$-time algorithm for~$A$ gives us a~$O\left(2^{n(1-\delta_0)}\right)$-time algorithm for \problem{$q$-SAT}. Thus, since we solve each $\ell$-instance of~$A$ resulting from $2^{(1-\alpha)n}$ instances of \problem{$q$-SAT} in time $O\left(\lambda^{k/2}\ell^{O(1)}\right)$, we solve the original $n$-variate instance~$F$ of \problem{$q$-TAUT} in time
    \begin{align}\label{eq:solvingparam}
    O\left(2^{(1-\alpha) n} \cdot (2^{\alpha n})^{1-\delta_0}\right)=O\left(2^{n(1-\alpha\delta_0)}\right) \,.
    \end{align}
    
    The total running time of the preprocessing and solving stages (see~\eqref{eq:preprocessingparam} and~\eqref{eq:solvingparam}) is bounded from above by $O\left(2^{n(1-\delta_0)}\right)+O\left(2^{n(1-\alpha\delta_0)}\right)=O\left(2^{n(1-\alpha\delta_0)}\right)$, which refutes NSETH, and implies a super-linear lower bound for Boolean series-parallel circuits.
\end{proof}

\subsection{Proof of~Lemma~\ref{st:tree_decomposition}}\label{apx:treelemma}

\begin{lemma}
\label{lemma:U_u}
Let $T(V, E)$~be a~rooted tree and $M\subseteq V$~be a~set of more than $\ell$ nodes. Then there is a~subtree~$U$ of~$T$ whose topmost node is~$u$ such that $T[V(T) \setminus V(U) \cup \{u\}]$ is connected and $\frac{\ell}{2} \le |M \cap (V(U) \setminus \{u\})| \le \ell$.
\end{lemma}
\begin{proof}
We prove this lemma by~induction on~the size of the tree~$T$. The base case of $|V|=1$ holds trivially (by choosing $U=T$). For the case $|V|>1$, suppose the root~$r$ of~$T$ has $m$~children with subtrees $T_1,\ldots,T_m$. Let $a_i = |V(T_i) \cap M|$. Assume without loss of generality that $a_1\leq \ldots \leq a_m$. Consider the following three cases.
\begin{itemize}
    \item If $a_m > \ell$, then we use the induction hypothesis 
    to~find~$U$ in~$T_m$.
    \item If $\frac{\ell}{2} \le a_m \le \ell$, then we take $U$ to be the subtree $T_m$ together with the root~$r$. Then $|M \cap (V(U) \setminus \{u\})|=|M \cap V(T_m)|=a_m$.
    \item 
    If $a_m < \frac{\ell}{2}$, then, for all $i \in [m]$, $a_i < \frac{\ell}{2}$. As $\sum_{i=1}^m a_i \ge |M|-1\geq \ell$, there exists $j \in [m]$ such that $\frac{\ell}{2} \le \sum\limits_{i = 1}^{j} a_i \le \ell$. We take $U$ to be the subtrees $T_1,\ldots,T_j$ together with the root $r$. This way we have $|M \cap (V(U) \setminus \{u\})|=\sum\limits_{i = 1}^{j} |M \cap V(T_i)|=\sum\limits_{i = 1}^{j} a_i$.
\end{itemize}
\end{proof}

\treedecomposition*
\begin{proof}
Choose an~arbitrary node~$r$ of~$T$ as the root and view~$T$ as a~rooted tree. Let $\ell = \lceil \frac{2k}{\theta - 1} \rceil$. 
As long as $|V(T) \cap M| > \ell$ we repeat the following procedure. (Note that it's possible that in the very beginning $|V(T) \cap M|=k\leq \ell$, so we don't run the procedure even once.) 
Since $|V(T) \cap M| > \ell$, \cref{lemma:U_u} gives us a tree $U$ and its root~$u$. We take $T_i=U$, and delete all the vertices $V(U) \setminus \{u\}$ from $T$.

When we can't apply this procedure anymore, we have $|V(T) \cap M| \leq \ell$, and we add the remaining tree with at most~$\ell$ nodes from~$M$ as the last~$T_i$ in our collection. Note that each $T_i$ is a subtree, and each $T_i$ contains at most~$\ell+1\leq \frac{2k}{\theta-1}+2$ nodes from~$M$. Since each application of \cref{lemma:U_u} removes at least $\ell/2$ vertices from~$M$, the number of such iterations is at most $\ceil{\frac{k}{\ell/2}} = \ceil{\frac{2k}{\ell}}$. Thus, the total number of subtrees is  $m \leq \ceil{\frac{2k}{\ell}} + 1 \le \ceil{\frac{2k(\theta - 1)}{2k}} + 1 = \theta$.
\end{proof}
\end{document}